\documentclass[%
 preprint, 
nopreprintnumbers,
nofootinbib,
 amsmath,amssymb,
 aps, physrev,
floatfix,
]{revtex4-2}

\usepackage{graphicx}
\usepackage{float}
\usepackage{physics}
\usepackage{subcaption}
\usepackage{amsthm}
\usepackage{tensor}
\usepackage{wasysym}
\usepackage{mathrsfs}
\usepackage{dcolumn}
\usepackage{bm}
\usepackage{hyperref}

\newcommand{\R}{\mathbb{R}}
\newcommand{\Z}{\mathbb{Z}}
\newcommand{\N}{\mathbb{N}}
\newcommand{\C}{\mathbb{C}}
\newcommand\edth{\mathord{\textnormal{ð}}}
\newcommand{\dt}{\text{d}t}
\newcommand{\dx}{\text{d}x}
\newcommand{\dtheta}{\text{d}\theta}
\newcommand{\dphi}{\text{d}\varphi}
\newcommand{\normord}[1]{:\mathrel{#1}:}
\newtheorem{theorem}{Theorem}[section]

\begin{document}
\count\footins=1000\relax

\title{
Obstructions to Traversable Wormholes in Einstein-Dirac Theory
}%

\author{Robert J. Weinbaum}
 \email{Contact author: rweinbaum@uchicago.edu}
\affiliation{%
Leinweber Institute for Theoretical Physics, Enrico Fermi Institute, and Department of Physics,\\
The University of Chicago,\\
933 East 56th Street, Chicago, Illinois 60637, USA}%

\date{\today}

\begin{abstract}
Traversable wormholes are among the most striking hypothetical solutions of general relativity, but every such geometry requires matter that violates the averaged null energy condition (ANEC). A possible candidate for such matter is the Dirac field. Several recent works have claimed traversable wormhole solutions of the classical Einstein-Dirac-Maxwell system, but these works did not properly take into account that single-particle states do not self-interact  electromagnetically and that the mode functions of single-particle states are of positive frequency. In this paper, we show that classical, positive-frequency Dirac fields on certain fixed wormhole geometries can violate ANEC, so they are indeed candidates for sourcing traversable wormholes. We then numerically search for static, spherically symmetric, asymptotically flat traversable wormhole geometries sourced by Dirac fields of a definite frequency $\omega>0$, angular momentum quantum number $\ell$, and definite parity. We find ``partial-wormhole solutions,'' describing a regular wormhole throat with correct asymptotics at one end of the wormhole, but we find that these solutions cannot be continued to a second asymptotically flat end. In the case of reflection-symmetric wormholes, we perform an additional search where we do not assume that the Dirac solution has a definite parity. In this case, after an extensive scan of throat-forming asymptotic data, we find that the conditions necessary for a reflection-symmetric wormhole throat cannot be obtained. Taken together, these results strongly support that the Einstein-Dirac system does not admit traversable wormhole solutions when sourced by a physically
meaningful Dirac field.
\end{abstract}

\maketitle

\tableofcontents

\section{Introduction\label{ch:introduction}}

A wormhole is a hypothetical spacetime geometry describing a tunnel-like structure connecting separated regions of space, or even different parallel universes altogether (Fig. \ref{fig:wormholes}). The idea of a \textit{traversable} wormhole, in which timelike or null geodesics can cross from one end to the other, has proven captivating to a wide audience. To physicists, wormholes are a class of possible solutions in Einstein's theory of general relativity that play a role in many foundational questions about the nature of our universe. Wormholes have been invoked to better understand the nature of fundamental particles \cite{EinsteinRosen1935ParticleProblem}, causality \cite{MorrisThorneYurtsever1988WeakEnergyCondition}, and even entanglement in the context of holography \cite{MaldacenaSusskind2013CoolHorizons}. To fans of science fiction, traversable wormholes serve as a possible means by which a highly advanced civilization can travel between widely separated stars or galaxies, without centuries having to pass for the travelers \cite{Sagan1985Contact}.

\begin{figure}[!ht]
    \begin{subfigure}[t]{0.45\textwidth}
            \centering
 \includegraphics[width=\linewidth]{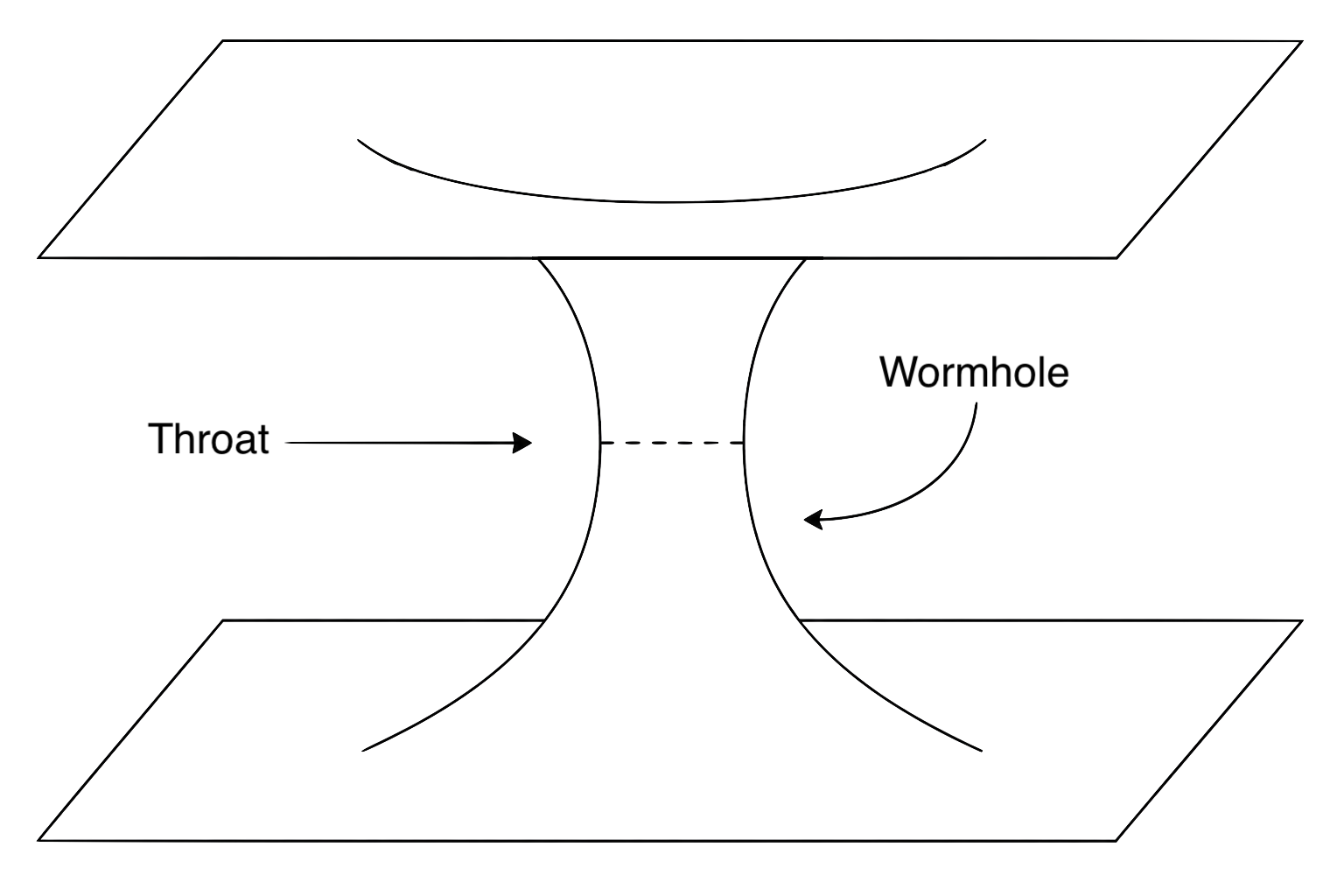}
        \caption{}
    \end{subfigure}
\hfill
     \begin{subfigure}[t]{0.45\textwidth}
            \centering
 \includegraphics[width=\linewidth]{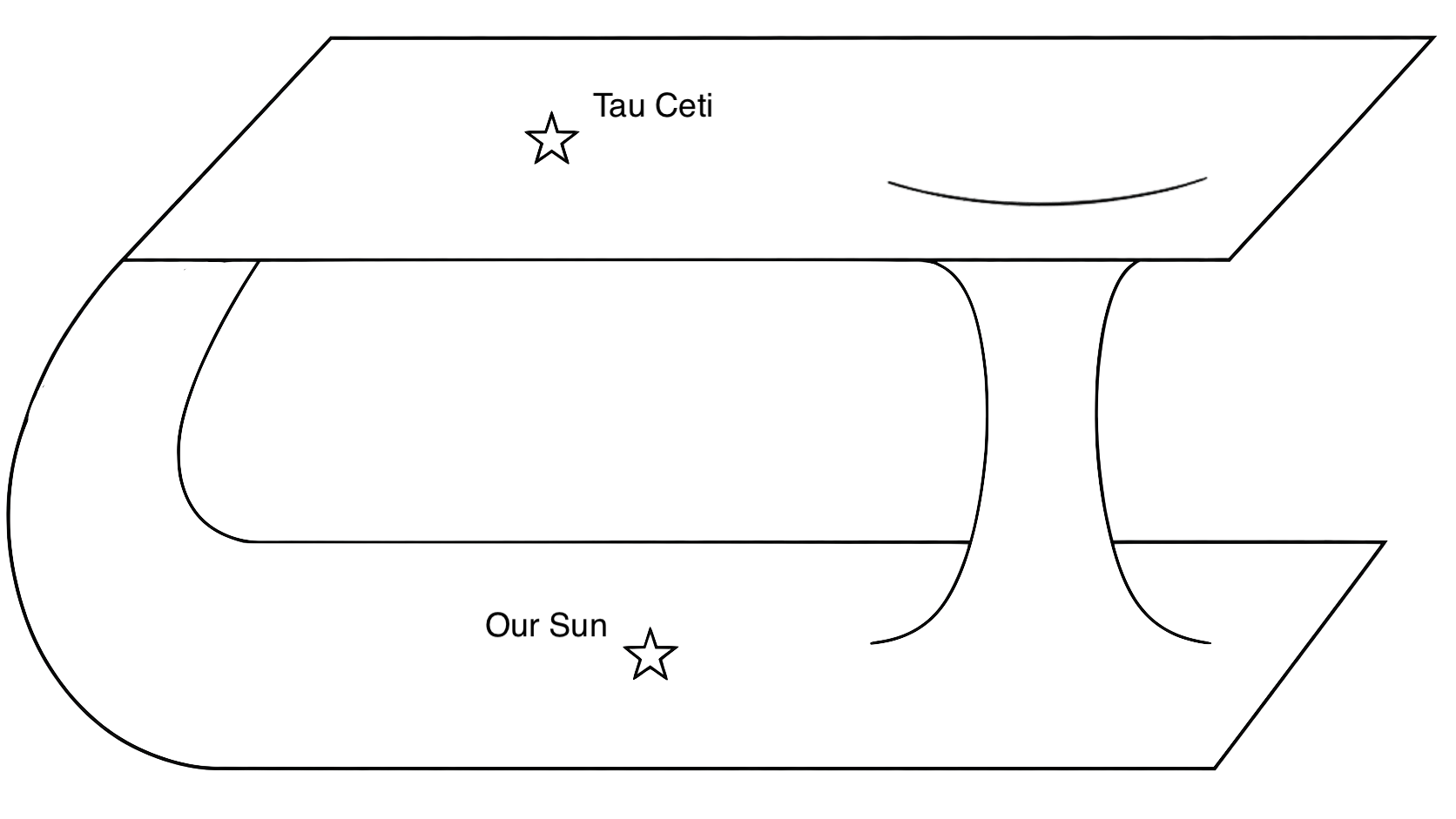}
        \caption{}
    \end{subfigure}
    \caption[Sketches of traversable wormhole geometries]{Equatorial slice of a static embedding diagram for a traversable wormhole connecting (a) two parallel universes or (b) two distant points in the same universe.}
    \label{fig:wormholes}
\end{figure}
\begin{figure}[!ht]
    \centering
    \includegraphics[width=0.45\textwidth]{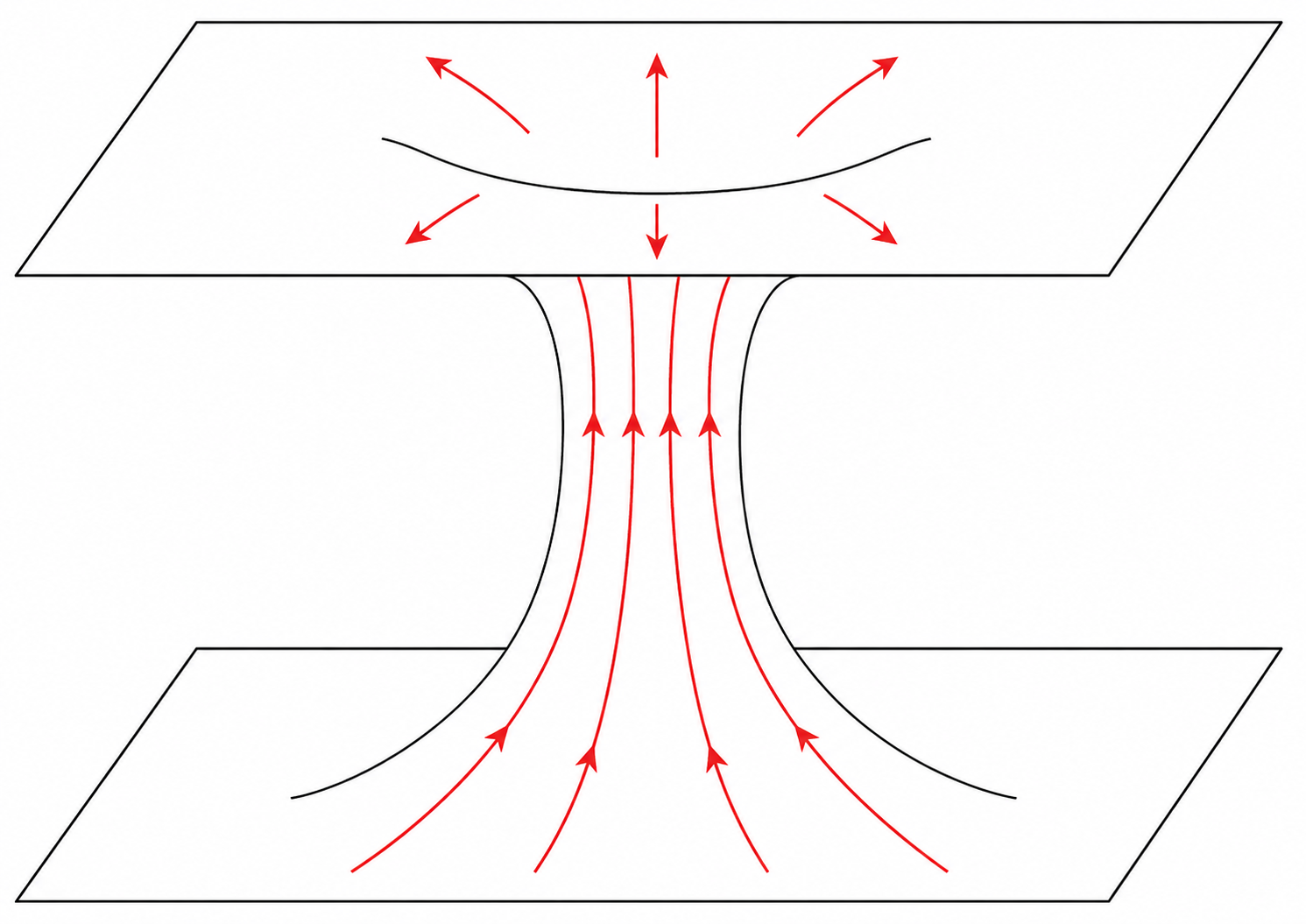}
    \caption[Radial null congruence propagating through a traversable wormhole]{A depiction of a radial congruence of null geodesics on an embedding diagram of a traversable wormhole geometry. This sketch shows an equatorial slice taken at constant time of a four-dimensional wormhole geometry.}
    \label{fig:null_cong}
\end{figure}

A major open question in the study of traversable wormholes is whether any traversable wormhole geometries are physically realistic solutions to the Einstein equations. A central obstruction to the existence of such solutions is that they require large violations of the null energy condition, which is a restriction on the stress-energy tensor $T_{ab}$ demanding that for any future-directed null vector $k^a$, we have $T_{ab}k^ak^b\geq 0$. The fact that traversable wormhole spacetimes require large violations of the null energy condition is most directly seen using the Raychaudhuri equation. Suppose we have a static, spherically symmetric, and asymptotically flat wormhole spacetime, and consider a congruence of radial null geodesics entering one end of the wormhole and exiting the other side (see Fig. \ref{fig:null_cong}). The Raychaudhuri equation for a congruence of null geodesics is \cite{Wald:1984rg}\begin{equation}
    \frac{d\theta}{d\lambda} = -\frac{1}{2}\theta^2-\sigma_{ab}\sigma^{ab}+\omega_{ab}\omega^{ab}-R_{ab}k^ak^b, \label{eq:null_raychaudhuri}
\end{equation} where $\theta,\sigma_{ab},$ and $\omega_{ab}$ are the expansion, shear, and twist of the null congruence, respectively. $\lambda$ is an affine parameter along each null geodesic, and $k^a = (\partial/\partial\lambda)^a$ is the null tangent. In the case of radial null geodesics in a spherically symmetric spacetime, we have $\sigma_{ab} = \omega_{ab} = 0$.\footnote{Spherical symmetry prevents a preferred direction from being picked out for the transverse deformation of the null congruence, forcing $\sigma_{ab}=0$. Similarly, a radial null congruence is hypersurface orthogonal, which forces $\omega_{ab} = 0$ by Frobenius' theorem.} At the narrowest point of the wormhole, referred to as the ``wormhole throat,'' the expansion $\theta$ switches from negative to positive, as the null congruence switches from converging to diverging at this point. This immediately gives $R_{ab}k^ak^b<0$, which implies via the Einstein equations that $T_{ab}k^ak^b < 0$.

A traversable wormhole requires the stronger condition that the averaged null energy condition be violated. This condition states that, along any complete, inextendible, affinely parameterized, future-directed null geodesic, the inequality \begin{equation}
    \int_{-\infty}^\infty T_{ab} k^ak^b \ d\lambda \geq 0
\end{equation} must hold, where $\lambda$ is an affine parameter and $k^a = (\partial/\partial\lambda)^a$ is the associated null tangent. By integrating eq. \eqref{eq:null_raychaudhuri} over $\lambda$ and replacing the left-hand side by its vanishing boundary values, we obtain \begin{align}
    \int_{-\infty}^\infty R_{ab}k^ak^b\ d\lambda <0 
\end{align} for a traversable wormhole geometry, which by the Einstein equations implies that the averaged null energy condition must be violated. Although we specialized to the static, spherically symmetric setting for simplicity, we note that more general topological arguments exist which show that a completely general traversable wormhole requires violations of the averaged null energy condition \cite{Galloway_1995}.

There has been a concerted effort over the last few decades to find physically realistic traversable wormhole solutions in general relativity. In a paper by  Maldacena, Milekhin, and Popov \cite{MaldacenaMilekhinPopov2023TraversableWormholes}, the authors attempt to obtain an approximate, semiclassical, traversable wormhole solution to the Einstein equations using massless, charged fermions. Their method involves dividing the spacetime geometry into regions where different analytic approximations hold, including a wormhole throat region where the semiclassical Einstein equations are solved. They then argue that these regional approximations have overlapping domains of validity, which allows the geometry to be consistently patched into a global wormhole configuration. While this approach may provide evidence for a semiclassical traversable wormhole solution, a global solution would place the analysis on firmer footing when considering genuinely global obstructions like the averaged null energy condition. A global solution would make it possible to evaluate such quantities directly, rather than inferring them from local approximations patched together across different regions, which could conceivably be insensitive to global obstructions.

Because the global semiclassical problem is exceptionally difficult, we instead investigate the simpler question of whether a classical matter model can support such a geometry. For a bosonic field, a classical solution would correspond to a semiclassical solution in a coherent state, with vacuum stress-energy neglected. However, most known classical matter satisfies the null energy condition, and by extension, the averaged null energy condition. This immediately excludes most familiar classical matter models as possible sources of traversable wormholes. Indeed, an early traversable wormhole solution due to Ellis required that the stress-energy was that of a scalar field with the wrong sign, in order to produce the required null energy condition violation \cite{Ellis1973Drainhole}. Similarly, the solutions discussed by Morris and Thorne use a general anisotropic stress-energy tensor, without imposing an equation of state that enforces any of the usual energy conditions \cite{Morris:1988tu}. These authors, and later other authors like Visser, explored solutions that constrain the presence of this hypothetical matter to a small region of the spacetime, but require hypothetical ``exotic matter'' nonetheless \cite{Visser1989TraversableWormholesSimpleExamples}.

A classical Dirac field, however, is able to violate both the null energy and averaged null energy conditions, as we will show in \S \ref{sec:Dirac_Fields_Can_Violate_ANEC}. The ability of Dirac fields to violate the null energy condition, even classically, makes them a natural matter model in which to test whether the usual obstruction can be overcome without resorting to hypothetical or nonexistent forms of matter. A paper by Bl\'azquez-Salcedo, Knoll, and Radu claimed to have found a traversable wormhole solution in general relativity, where the stress-energy tensor consisted of a charged, classical Dirac field coupled to a Maxwell field \cite{Blazquez-Salcedo:2020czn}. These authors numerically solved the Einstein-Dirac-Maxwell system from the wormhole throat out to one of the asymptotic infinities. They then reflected their solution across the wormhole throat and applied matching conditions at the junction to propose a global solution. These authors did not properly account for the behavior of the Dirac field under such a reflection, which yielded an error in their matching conditions \cite{DanielsonSatishchandranWaldWeinbaum2021BSKR}. 

Shortly thereafter, a paper by Konoplya and Zhidenko \cite{Konoplya:2021hsm} and another by Kain \cite{Kain2023EDMWormholesQFT} followed up with their own numerical solutions of the classical Einstein-Dirac-Maxwell system. Although these appear free of the issues we pointed out in \cite{DanielsonSatishchandranWaldWeinbaum2021BSKR}, in order for the classical solutions to correspond to single-particle quantum states in the static setting considered in these papers, the Dirac field must be positive-frequency with respect to Killing time, and the electromagnetic self-interaction should not be treated as an interaction term in the Dirac equation. Neither of these issues is appropriately accounted for in \cite{Konoplya:2021hsm,Kain2023EDMWormholesQFT}.

The aim of this paper is therefore to determine whether a normalized, positive-frequency classical Dirac field can source a static, spherically symmetric, asymptotically flat traversable wormhole spacetime. To start, we will dedicate Sec. \ref{ch:semiclassical_description_of_dirac_fields} to understanding which classical Dirac solutions have a physical, semiclassical interpretation. We begin with a treatment of classical Dirac fields in curved spacetime, including how the Dirac stress-energy and Dirac current are defined in a general spacetime. We then explore to what extent a classical Dirac field should be taken as an accurate model of the quantum Dirac field. We first show that single-particle states of the Dirac field should, in the absence of an external electromagnetic field, be treated as free, even when the Dirac field is electrically charged and coupled to the electromagnetic field; an appropriate definition of a single-particle state in an interacting quantum field theory already accounts for the effects of the self-interaction. We then explain that, although a single-particle state of the charged Dirac field satisfies the uncharged equation of motion, we cannot describe its stress-energy using the free Dirac stress-energy tensor, while simultaneously respecting local conservation of stress-energy. As a result, we will restrict our consideration to the neutral Einstein-Dirac system. In this setting, we will argue that positive-frequency, normalizable, classical solutions of the free Dirac equation correspond to single-particle states of the free Dirac field, and the classical, noninteracting Dirac stress-energy agrees with the expected stress-energy of a single-particle state of the free Dirac field, neglecting contributions from the vacuum.

Next, in Sec. \ref{ch:dirac_fixed_background}, we will treat Dirac solutions on a fixed traversable wormhole background geometry, neglecting the effects of backreaction. We first introduce the spacetime metric for a static, spherically symmetric traversable wormhole spacetime, and write down the explicit form of the Dirac equation, Dirac current, and Dirac stress-energy in such a spacetime. Though our discussion of the Dirac equation will be specialized to traversable wormhole geometries, these same steps can be followed exactly in a general static, spherically symmetric spacetime. Finally, we will show that Dirac solutions can violate both the pointwise and averaged null energy conditions in this ``test-field'' approximation.

In Sec. \ref{ch:static_spherically_symmetric_dirac_solutions}, we will restrict to solutions of the Dirac equation that are themselves static and spherically symmetric, while continuing to work in a fixed static and spherically symmetric background. We will argue that one needs to consider normalizable, stationary bound-state solutions of the Dirac equation to obtain time-independent observables. We then show that a mixed state of these solutions is required to obtain a fully static and spherically symmetric solution. Finally, we exhibit a stationary, normalizable, positive-frequency bound state that violates the averaged null energy condition on a fixed wormhole background, which motivates our investigation into the full Einstein-Dirac system.

 In Sec. \ref{ch:EDM_system}, we consider the fully backreacting Einstein-Dirac system. Here, we present analytic and numerical arguments that provide strong evidence against static, spherically symmetric, asymptotically flat traversable wormhole solutions to the Einstein-Dirac system. We will introduce the full equations describing the system, accounting for the considerations discussed in Sec. \ref{ch:semiclassical_description_of_dirac_fields} to ensure our solution is semiclassically meaningful. We show that there is no difficulty with obtaining a ``partial-wormhole solution,'' where the fields have the correct asymptotic behavior at one end of the wormhole and form a throat at the other end, with no horizons or singularities appearing in the solution. 
 
 We restrict consideration to the case where the Dirac solution has a definite frequency $\omega$ and a definite angular momentum quantum number $\ell$, and initially study the case where the Dirac solution is of a definite parity. We will argue in \S \ref{sec:non_existence_generic} that in this setting, a partial wormhole solution cannot be extended to a second asymptotically flat end. We present strong numerical evidence that, when evolving inward from one asymptotic end, only one of the two sphere-parities permits a wormhole throat to form while the other forbids it and, crucially, that the throat-permitting parity at one end is the opposite of the throat-permitting parity at the other. Since parity is a global label, a definite-parity solution that can form a throat from one end has the wrong parity to form one from the other end, ruling out consistent, global solutions that describe generic traversable wormholes.

In a reflection-symmetric wormhole, in which the geometry on one side of the throat is the mirror image of the other, we drop the restriction that the Dirac solutions are of definite parity, as bound-state solutions of both parities are supported at the same frequency. We treat this case separately in \S \ref{sec:non_existence_symmetric} by deriving additional constraints that a reflection-symmetric throat imposes on the Dirac spinor, and we present an extensive numerical scan showing that these constraints are never approached. Together, these results provide strong evidence that the Einstein-Dirac system, when restricted to physically realistic classical Dirac solutions, does not admit traversable wormhole solutions.

\textit{Conventions:} We will work in units with $c=\hbar = 1$ and $G = 1/8\pi$, so that the Einstein equation is simply $G_{ab} = T_{ab}$. We use the metric signature $(+---)$, as this metric signature has particular utility when working with spinor fields. Latin indices on spinor/tensor/vector fields should be interpreted as abstract indices in the sense of \cite{Wald:1984rg}, and not as components with respect to any particular basis. $\N+\frac{1}{2}$ denotes the set $\{1/2,3/2,5/2,\dots\}$ and $\Z+\frac{1}{2}$ denotes the set $\{\dots,-3/2,-1/2,1/2,3/2,\dots\}.$

\section{Semiclassical Description of Dirac Fields in Curved Spacetime}\label{ch:semiclassical_description_of_dirac_fields}

In this section, our goal is to introduce the classical Dirac field and to clarify when it admits a valid semiclassical interpretation. We will begin in \S \ref{sec:general_dirac_equation} with the classical description of the Dirac field on a curved spacetime background. This discussion fixes our spinor conventions and introduces the Dirac current and stress-energy tensor. We will also present two commonly used formalisms for writing the Dirac equation in curved spacetime: the two-spinor formalism and the four-component Dirac spinor formalism. We will summarize both techniques and present the mapping between them, as we believe such a reference will prove valuable to readers who need a dictionary to translate between the two formalisms. 

Then, in \S \ref{sec:self_interaction_effects} and \S \ref{sec:semi_classical_meaning}, we will turn to the question of when the classical Dirac field is an appropriate semiclassical model. We will show that a positive-frequency, classical Dirac field corresponds to a single-particle state in the quantum theory, and single-particle states evolve according to the free equation of motion. However, since taking the total stress-energy tensor of the theory to be the sum of the free Dirac and Maxwell stress-energy tensors is inconsistent with local stress-energy conservation, we will restrict our attention to normalizable, positive-frequency, neutral, classical Dirac fields satisfying the Einstein-Dirac equations. Such considerations were not properly accounted for in earlier works \cite{Konoplya:2021hsm,Kain2023EDMWormholesQFT}, which consider classical solutions with negative frequencies and do not treat the electromagnetic self-interaction correctly.

\subsection{\label{sec:general_dirac_equation}The Dirac Equation on a General Curved Spacetime}
We will now present the equations describing the Dirac field in curved spacetime. We will restrict our attention here to the free Dirac field, but the discussion can be generalized to include interactions without much effort. This will allow us to clearly formulate our spinor conventions and present definitions of the Dirac stress-energy tensor and Dirac current.  We will present the relevant equations using the two-spinor and four-component Dirac spinor formalisms so that this work is accessible to readers familiar with either presentation.

The two-spinor formalism is a convenient formulation of spinor fields in 4-dimensional spacetimes and is commonly used in the relativity community \cite{Wald:1984rg,PenroseRindler1984SpinorsSpacetime1}. In the two-spinor formalism, Weyl spinors, denoted by symbols carrying a raised, capitalized Latin index (e.g., $\Phi^A$), are elements of a spinor space $(W,\epsilon_{AB})$, where $W$ is a two-dimensional complex vector space, and $\epsilon_{AB}$ is a nondegenerate, antisymmetric tensor from $W\times W\to\C$. We introduce a basis $\{o^A,\iota^A\}$ for $W$, normalized by the convention $\epsilon_{AB}o^A\iota^B = 1.$ Conjugate Weyl spinors, denoted by symbols carrying a raised, capitalized, primed Latin index (e.g., $X^{A^\prime}$) belong to the conjugate spinor space $(\bar W,\bar \epsilon_{A^\prime B^\prime})$, on which we introduce a basis $\{\bar o^{A^\prime},\bar \iota^{A^\prime}\}.$ From these basis vectors, we obtain the real combinations which form a basis\footnote{Note that the combination used for ${e_2}^{AA^\prime}$ has a sign difference relative to that used in \cite{Wald:1984rg,PenroseRindler1984SpinorsSpacetime1}, but is consistent with \cite{GHP:1973} and standard particle physics literature. This will result in a mismatch of the usual 2-spinor conventions as to whether ordinary or conjugate spinors correspond to left- or right-handed Weyl spinors.} for the Hermitian elements of $W\otimes \bar W$: \begin{equation}
    \begin{aligned}
        &{e_0}^{AA^\prime} = \frac{1}{\sqrt{2}}(o^A\bar o^{A^\prime} + \iota^A\bar\iota^{A^\prime}) , 
        &&{e_1}^{AA^\prime} = \frac{1}{\sqrt{2}}(o^A\bar\iota^{A^\prime} + \iota^A\bar o^{A^\prime}),\\
        &{e_2}^{AA^\prime} =-\frac{i}{\sqrt{2}}(o^A\bar \iota^{A^\prime} - \iota^A\bar o^{A^\prime}), 
        &&{e_3}^{AA^\prime} = \frac{1}{\sqrt{2}}(o^A\bar o^{A^\prime} - \iota^A\bar \iota^{A^\prime}).
    \end{aligned}
\end{equation}

We now introduce non-coordinate vector fields $\{{e_0}^a,{e_1}^a,{e_2}^a,{e_3}^a\}$ that form a basis for the tangent space at each point in the spacetime, which are orthonormal in the sense that $g_{ab}{e_\mu}^a{e_\nu}^b = \eta_{\mu\nu}:=\text{diag}(+---)_{\mu\nu}$ everywhere; such a collection of vector fields is called a \textit{tetrad}. At each point $p$ in our spacetime, we can identify elements in the tangent space at that point with Hermitian elements of $W\otimes \bar W$ by the identification \begin{equation}
    {e_\mu}^{AA^\prime} \leftrightarrow {e_\mu}^a\vert_p.\label{spinor_identification}
\end{equation} Explicitly, a vector at a point $p$ in the spacetime can be written in terms of its components with respect to the basis given by the tetrad as $v^a = v^\mu {e_\mu}^a\vert_p$. This will be identified with a real element $v^{AA^\prime} = v^\mu {e_\mu}^{AA^\prime}$ of $W\otimes \bar W$. Using the spinor basis, this can be represented as a $2\times 2$ matrix with components \begin{equation}
    \begin{aligned}
    &\frac{1}{\sqrt{2}}\begin{pmatrix}
        v^0 + v^3 & v^1 - iv^2\\
        v^1 + iv^2 & v^0 - v^3
    \end{pmatrix} = \frac{1}{\sqrt{2}}[v^0 I_{2} + v^1\sigma^1 + v^2 \sigma^2 + v^3\sigma^3],
\end{aligned}
\end{equation} where $I_{2}$ denotes the $2\times 2$ identity matrix and $\sigma^i$ are the usual Pauli matrices \begin{align}
   & \sigma^1 = \begin{pmatrix}
        0  & 1\\
        1 & 0
    \end{pmatrix},&&\sigma^2 = \begin{pmatrix}
        0 & -i \\
        i & 0 
    \end{pmatrix}, &&& \sigma^3 = \begin{pmatrix}
        1 & 0 \\
        0 & -1
    \end{pmatrix}. 
\end{align}

A Dirac field is defined to be a pair consisting of a Weyl spinor and a dual conjugate Weyl spinor $(\Phi^A,X_{A^\prime})\in W\oplus \bar W^*$. The free Dirac equation for a Dirac field with mass $\mu$ is given by \begin{align}
    &\nabla_{AA^\prime}\Phi^A = \frac{\mu}{\sqrt{2}}X_{A^\prime} ,&&\nabla^{AA^\prime}X_{A^\prime}  =- \frac{\mu}{\sqrt{2}}\Phi^A, \label{eq:two_spinor_dirac}
\end{align}  where $\nabla_{AA^\prime}$ is the spinor covariant derivative \cite{PenroseRindler1984SpinorsSpacetime1,Wald:1984rg}. We define the spinor components of $\Phi^A$ and $X_{A^\prime}$ in terms of the aforementioned bases as \begin{align}
    &\Phi^A = \Phi_0o^A + \Phi_1\iota^A, && X_{A^\prime} = X_{1^\prime}\bar o_{A^\prime} - X_{0^\prime}\bar\iota_{A^\prime},\label{eq:two_spinor_components}
\end{align} where $\bar o_{B^\prime}:= \bar\epsilon_{A^\prime B^\prime}\bar o^{A^\prime}$ and $\bar \iota_{B^\prime}:= \bar\epsilon_{A^\prime B^\prime}\bar \iota^{A^\prime}$ form a basis of $\bar W^*$. The seemingly odd signs in the definition of the $X_{A^\prime}$ components come from the fact that, because $\bar\epsilon_{A^\prime B^\prime}$ is antisymmetric, $-\bar\iota_{A^\prime}$ is the basis vector of $\bar W^*$ dual to $\bar o^{A^\prime}$, and $\bar o_{A^\prime}$ is dual to $\bar \iota^{A^\prime}$.

We now present the same equations using the four-component Dirac spinor formalism, which is more common in the quantum field theory literature \cite{PeskinSchroeder:1995}. In the four-component Dirac spinor formalism, we again choose a tetrad $\{{e_0}^a,{e_1}^a,{e_2}^a,{e_3}^a\}$ on our spacetime as before, but we must also choose gamma matrices $\gamma^\mu$ satisfying $\{\gamma^\mu,\gamma^\nu\} = 2\eta^{\mu\nu}I_{4}$, where $I_{4}$ is the $4\times 4$ identity matrix. A convenient choice that makes the mapping between the two formalisms most transparent is to take the chiral representation \cite{PeskinSchroeder:1995} \begin{align}
    &\gamma^0 = \begin{pmatrix}
        0 & I_{2}\\
        I_{2} & 0
    \end{pmatrix}, && \gamma^i=\begin{pmatrix}
        0 & \sigma^i\\
        -\sigma^i & 0
    \end{pmatrix}.\label{eq:chiral_rep}
\end{align} To write down the Dirac equation, we define the dual basis of the cotangent bundle $\{{e^0}_a,{e^1}_a,{e^2}_a,{e^3}_a\}$ by the relation ${e^\mu}_a {e_\nu}^a = \tensor{\delta}{^\mu_\nu}$. We then solve the torsion-free condition \begin{equation}
    \text{d}e^\mu  + \tensor{\omega}{^\mu_\nu} \wedge e^\nu = 0
\end{equation} to find the connection 1-forms $(\tensor{\omega}{_\mu_\nu})_a:=\eta_{\mu\rho}(\tensor{\omega}{^\rho_\nu})_a.$ These, in turn, are used to define the spin-covariant derivative \cite{CollasKlein:2019} \begin{align}
    &D_\mu := e_\mu + \Gamma_\mu ,&&\Gamma_\mu := \frac{1}{4}(\omega_{\rho\lambda})_a\tensor{e}{_\mu^a}\gamma^\rho\gamma^\lambda,
\end{align} which we use to write the Dirac equation for a free Dirac field $\Psi$ with mass $\mu$ \begin{align}
    (i\gamma^\nu D_\nu - \mu)\Psi =0 .\label{eq:four_component_dirac}
\end{align} If $(X_{A^\prime},\Phi^A)$ satisfy eq. (\ref{eq:two_spinor_dirac}), and we define the components of $\Psi$ via \begin{equation}
    \Psi = \begin{pmatrix}
        iX_{0^\prime}\\
        iX_{1^\prime}\\
        \Phi_0\\
        \Phi_1
    \end{pmatrix},\label{eq:spinor_formalisms_identification}
\end{equation} then $\Psi$ satisfies eq. (\ref{eq:four_component_dirac}) (in the representation (\ref{eq:chiral_rep}) with the same choice of tetrad). The converse is also true, which provides a mapping between both spinor formalisms.

There is a natural inner product on Dirac solutions, and physically meaningful classical Dirac solutions must be normalizable in this Dirac inner product (see \S \ref{sec:semi_classical_meaning}). In a general spacetime, the Dirac inner product is given by \begin{align}
    &\langle \Psi_1,\Psi_2 \rangle =  \int_\Sigma j^a(\Psi_1,\Psi_2)\ n_a\ d\Sigma ,\label{eq:Dirac_Inner_Product}
\end{align} where $\Sigma$ is any Cauchy surface, $n_a$ is the normal covector to $\Sigma$, and $d\Sigma$ is the induced volume element on $\Sigma$. $j^a(\Psi_1,\Psi_2)$ is the Dirac current, which can be expressed in either the four-component spinor or two-spinor formalisms as \begin{align}
    j^a(\Psi_1,\Psi_2) := \bar\Psi_1 \gamma^a \Psi_2\label{eq:Dirac_Charge_Current}\end{align} or \begin{align} j^{AA^\prime}([X_{A^\prime},\Phi^A],[\Sigma_{A^\prime},\Xi^A]) := \sqrt{2}(\bar X^A \Sigma^{A^\prime}+\bar\Phi^{A^\prime}\Xi^A ), 
\end{align} where $\Psi_1$ and $[X_{A^\prime},\Phi^A]$ are supposed to represent the same Dirac spinor under the correspondence in eq. (\ref{eq:spinor_formalisms_identification}), and similarly for $\Psi_2$ and $[\Sigma_{A^\prime},\Xi^A]$. The raised index objects are defined by $\bar X^A := \epsilon^{AB}\bar X_B$ and $\Sigma^{A^\prime}: =\bar\epsilon^{A^\prime B^\prime}\Sigma_{B^\prime}$. $\bar \Psi:=\Psi^\dagger \gamma^0$ is the Dirac adjoint, which in the two-spinor formalism corresponds to the map $[ X_{A^\prime},\Phi^A]\to [\bar\Phi^{A^\prime},\bar X_A]$, taking a Dirac spinor to a dual Dirac spinor, also called a cospinor. This Dirac current is conserved on Dirac solutions, so eq. (\ref{eq:Dirac_Inner_Product}) is independent of the choice of Cauchy surface.

The stress-energy tensor of a Dirac field is given by \cite{Tong2006QFTSymmetriesCurrents} \begin{align}
    T_{ab} = \frac{i}{2}\bar\Psi \gamma_{(a} D_{b)} \Psi+\text{h.c.} \label{eq:dirac_stress_energy}
\end{align} where ``$+\text{h.c.}$'' denotes adding the Hermitian conjugate of the preceding terms. Equivalently, in the two-spinor formalism, the Dirac stress-energy tensor can be written \cite{PenroseRindler1984SpinorsSpacetime1}
 \begin{equation}
 \begin{aligned}
      T_{AA^\prime BB^\prime} &= -\frac{i}{2\sqrt{2}}[\Phi_A\nabla_{BB^\prime}\bar\Phi_{A^\prime} - \bar X_A\nabla_{BB^\prime} X_{A^\prime}+\Phi_B\nabla_{AA^\prime}\bar\Phi_{B^\prime} - \bar X_B\nabla_{AA^\prime} X_{B^\prime} ] +  \text{c.c.}
 \end{aligned}
    \label{eq:dirac_stress_energy_2_component}
\end{equation} where $\Phi_A:=\epsilon_{BA}\Phi^B$.

\subsection{Single-Particle States of the Dirac Field Satisfy the Free Equation of Motion}
 \label{sec:self_interaction_effects}

In previous works attempting to find wormhole solutions sourced by Dirac fields \cite{Blazquez-Salcedo:2020czn,Konoplya:2021hsm,Kain2023EDMWormholesQFT}, the authors couple the Dirac fields to electromagnetism. This is done by solving the fully coupled Einstein-Dirac-Maxwell system, in which the Dirac and Maxwell equations take the form
\begin{align}
    &(i\gamma^a(D_a -iqA_a)- \mu)\Psi = 0, && \nabla_aF^{ab} = qj^b(\Psi,\Psi).
\end{align} Here, $\gamma^a$ are the gamma matrices and $D_a$ is the spin-covariant derivative, both introduced in \S \ref{sec:general_dirac_equation}. $\Psi$ is the classical Dirac field, $\mu$ is its mass, $A_a$ is the electromagnetic 4-potential, $q$ is the electric charge of the Dirac field, $F_{ab} := \nabla_{[a}A_{b]}$ is the Maxwell field strength tensor, and $j^a$ is the Dirac current (also introduced in \S \ref{sec:general_dirac_equation}).  This system describes the behavior of an electrically charged spin-$1/2$ field, such as the electron. However, we will now argue that it is not appropriate to include the vector potential $A_a$ generated by the electron itself in the Dirac equation when considering the semiclassical model of a single particle.

The key consideration is how one defines a single-particle state in interacting quantum field theory. We will argue that, when the one-particle states are appropriately defined, interactions are already accounted for in the definitions of the free theory parameters. To see this, consider an interacting quantum field theory for a massive field in flat spacetime. In fully interacting quantum field theory, we do not expect a Fock space structure to exist. Still, from the Wightman Axioms, we do expect a Hilbert space $\mathscr{H}$ to exist that admits a unitary representation of the Poincar\'e group  \cite{Haag1996LocalQuantumPhysics}. The analysis of Wigner \cite{Wigner1939InhomogeneousLorentzGroup} shows that the irreducible representations are labeled by the values of the Casimir operators $P^2 = P_\mu P^\mu $ and $W^2= W_\mu W^\mu$ of the Poincar\'e group, where $P_\mu$ is the 4-momentum operator and $W_\mu$ is the Pauli-Lubanski operator.

Suppose that the spectrum of $P^2$ has a discrete part, and that $m^2>0$ is the smallest positive eigenvalue of $P^2$. Let $\Pi_{P^2=m^2}$ denote the orthogonal projection operator onto the eigenspace of $P^2$ with eigenvalue $m^2$. Wigner's analysis then tells us that the resulting Hilbert subspace can be expressed as a direct sum of irreducible representations labeled by their spin $s$, and so we can define a second orthogonal projection operator $\Pi_s$ onto the spin-$s$ subspace. The one-particle Hilbert space $\mathscr{H}_{1p}$ is defined by applying both of these projection operators to the full Hilbert space, i.e., $\mathscr{H}_{1p}:= \Pi_s\Pi_{P^2=m^2}\mathscr{H}$. 

$\mathscr{H}_{1p}$ carries an irreducible\footnote{In principle, $\mathscr{H}_{1p}$ can be a direct sum of multiple irreducible representations of the same mass and spin, corresponding to the case of having multiple non-interacting particles with the same mass and spin. In this case, we take one of the irreducible factors to be the one-particle Hilbert space.} representation of the Poincar\'e group, labeled by its mass $m$ and its spin $s$. Wigner's analysis shows that this representation is unitarily equivalent to that of a free particle with the same mass and spin. In particular, the Hamiltonian $H = P^0$ is represented on $\mathscr{H}_{1p}$, and will be unitarily equivalent to the Hamiltonian $H_0$ for a free particle with the same mass and spin. This implies that any observable $\mathcal{O}:\mathscr{H}_{1p}\to\mathscr{H}_{1p}$ will, under Heisenberg evolution, evolve in a manner that is unitarily equivalent to the evolution of an observable $\widetilde{\mathcal{O}}$ evolving according to the free Hamiltonian, when evaluating its matrix elements in $\mathscr{H}_{1p}$. Explicitly, if $U$ is the unitary map implementing $UHU^\dagger = H_0$, then $\widetilde{\mathcal{O}}:=U\mathcal{O}U^\dagger$ will evolve according to the free Hamiltonian $H_0$: for all $\psi,\phi\in\mathscr{H}_{1p}$ and $\mathcal{O}:\mathscr{H}_{1p}\to\mathscr{H}_{1p}$, we have 
\begin{equation}
    \begin{aligned}
    \bra{\psi}\widetilde{\mathcal{O}}(t)\ket{\phi} &= \bra{\psi} (U \mathcal{O}(t)U^\dagger)  \ket{\phi} = \bra{\psi}Ue^{iHt}\mathcal{O}(0)e^{-iHt}U^\dagger \ket{\phi}  = \bra{\psi} e^{iH_0 t}\widetilde{\mathcal{O}}(0)e^{-iH_0t}\ket{\phi}.
\end{aligned}
\end{equation} In other words, a one-particle state, defined in this way, evolves in agreement with the dynamics of a free particle of the appropriate mass and spin. In essence, the interactions are already accounted for in the values of the representation's mass and spin.

There are a few technical caveats worth mentioning. First, a general observable may not map the one-particle Hilbert space into itself, but so long as it preserves the one-particle Hilbert space to good approximation, the dynamics should be in agreement with the free theory to good approximation as well. A fully general observable that is highly disruptive to the one-particle nature of the theory (e.g., one that injects sufficient energy to induce particle production) should not be expected to evolve according to free, single-particle dynamics, but such an observable should also not be expected to admit a clean semiclassical description. Furthermore, it was essential to the argument that $P^2$ has a genuine eigenvalue at $m^2$, so that $\Pi_{P^2=m^2}$ yields a projection onto a valid closed subspace of the full Hilbert space that itself carries a representation of the Poincar\'e group. This can only be the case when considering \textit{stable} particles,\footnote{In quantum electrodynamics, a single, bare electron state does not exist as a genuine mass eigenstate of the theory. Physical electron states require a radiative dressing in the manner of Faddeev and Kulish \cite{WaldDanielsonSatishchandranInPrep}.} so one could not try to use this approach to define a single quark Hilbert space in quantum chromodynamics, for example.

This result is consistent with the familiar approach to solving the energy levels of a single electron in the hydrogen atom using the Dirac equation. In this calculation, $A_a$ is taken to be that of the Coulomb field of the proton, modeled as an external field, and the self-field generated by the electron is not included in the Dirac equation \cite{AuvilBrown1978RelativisticHydrogenAtom}. The predictions of the energy levels of hydrogen are consistent with experiment, and adding this self-force na\"ively would lead to large deviations in the results of the calculation. In a renormalizable quantum field theory, this is handled via a renormalization procedure which redefines the bare parameters in the Lagrangian \cite{PeskinSchroeder:1995,Weinberg1995BoundStatesExternalFields}.

As we have argued above, a classical Dirac field corresponding to a single-particle state should satisfy the free Dirac equation
\begin{equation}
(i\gamma^aD_a-\mu)\Psi = 0.
\end{equation} One would also expect that the electromagnetic field will satisfy Maxwell’s equation with source given by the Dirac current
\begin{equation}
\nabla_aF^{ab} = qj^b(\Psi,\Psi).
\label{eq:sourced_maxwell}
\end{equation} However, if these assumptions hold, then local conservation of stress-energy fails if we take the total stress-energy to be the sum of the free Dirac stress-energy (eq. \eqref{eq:dirac_stress_energy}) and the Maxwell stress-energy \begin{equation}
T_{ab}^\text{EM}
=
-g^{cd}
\left[
F_{ac}F_{bd}
- \frac{1}{4}g_{ab}g^{ef}F_{ce}F_{df}
\right].
\end{equation} For a Maxwell field satisfying eq. \eqref{eq:sourced_maxwell}, we have
\begin{equation}
\nabla^aT_{ab}^\text{EM}
=
qF_{ab}j^a.
\end{equation} In the fully interacting classical Dirac-Maxwell system, this nonzero divergence is precisely canceled by the divergence of the minimally coupled Dirac stress-energy tensor, as one has
\begin{equation}
\nabla^aT_{ab}^\text{Dirac,int}
=
-qF_{ab}j^a,
\end{equation}
so that we have a local conservation law for the full stress-energy tensor
\begin{equation}
\nabla^a
\left(
T_{ab}^\text{Dirac,int}
+
T_{ab}^\text{EM}
\right)
=
0.
\end{equation} Physically, this tells us that local conservation of energy follows when the Dirac and Maxwell fields exchange energy-momentum through the Lorentz force. For the semiclassical single-particle state considered here, however, the Dirac field satisfies the free equation of motion. If the Dirac stress-energy tensor were given by the free field expression given in eq. \eqref{eq:dirac_stress_energy}, then its stress-energy would be separately conserved:
\begin{equation}
\nabla^aT_{ab}^\text{Dirac,free}=0.
\end{equation}
Consequently, 
\begin{equation}
\nabla^a
\left(
T_{ab}^\text{Dirac,free}
+
T_{ab}^\text{EM}
\right)
=
qF_{ab}j^a,
\end{equation}
which is generically nonzero. As we will ultimately include gravitational backreaction via the Einstein equations, local conservation of stress-energy is required so that the stress-energy tensor is compatible with the Bianchi identity $\nabla^aG_{ab}=0$.

This argument implies that the Dirac stress-energy tensor in quantum electrodynamics, when restricted to the single-particle subspace, cannot be the same as that of a single particle in free Dirac theory. To study the full Einstein-Dirac-Maxwell system consistently, we would need to employ the full machinery of perturbative, interacting quantum electrodynamics. Since we aim to study the geometry sourced by the semiclassical Dirac field itself, our approach will be to consider an electrically neutral Dirac field and study the Einstein-Dirac system
\begin{align}
&G_{ab}=T_{ab}^\text{Dirac},
&&
(i\gamma^aD_a-\mu)\Psi=0.
\end{align}

\subsection{The Expected Stress-Energy for Single-Particle States of a Free Dirac Field}
\label{sec:semi_classical_meaning}
Though classical spinor fields are perfectly well-defined mathematical objects, their interpretation in the context of quantum field theory is not immediately clear. For bosonic fields, coherent states in the quantum theory can have large amplitude relative to their fluctuations, so these can be well-described by classical solutions. However, fermionic fields, such as the Dirac field, lack a notion of a coherent state due to the Pauli exclusion principle. Nonetheless, a physical interpretation of these solutions is possible by imposing appropriate restrictions on the frequency and the definition of the classical stress-energy, as we will now argue.

Let us review the quantum field theoretic description of a free Dirac field in a stationary spacetime. Let $\mathscr{H}$ be the Hilbert space of classical Dirac solutions, given by taking the completion of the space of classical solutions to the Dirac equation that are normalizable in the Dirac inner product (introduced in \S \ref{sec:general_dirac_equation}). In stationary spacetimes, we have a natural notion of positive and negative frequencies with respect to the timelike Killing vector field. More precisely, time evolution with respect to Killing time $t$ is implemented by a strongly continuous unitary one-parameter group $U(t)$ on $\mathscr{H}$ \cite{Wald1994QFTCurvedSpacetimeBlackHoleThermo}. By Stone's theorem, we can write $U(t) = e^{-iHt}$ where $H$ is self-adjoint on its appropriate Stone domain, and so the spectral theorem allows us to decompose $\mathscr{H}$ into $\mathscr{H}_+\oplus \mathscr{H}_-$, which are defined to be the positive and negative spectral subspaces of $H$, respectively.\footnote{By the analysis of \S \ref{sec:stationary_bound_states_exist}, the static traversable wormhole spacetimes we will consider do not admit nontrivial, normalizable, zero-energy solutions, and so we do not need to concern ourselves with the presence of so-called ``zero modes.''} $\mathscr{H}_+$ is the one-particle Hilbert space, and $\overline{\mathscr{H}_-}$ is the one-antiparticle Hilbert space. From these, we obtain the Dirac Fock space \begin{align}
    \mathscr{F}_\text{Dirac} := \mathscr{F}_A(\mathscr{H}_+)\otimes \mathscr{F}_A(\overline{\mathscr{H}_-}) 
\end{align} where  \begin{equation}
    \mathscr{F}_A(\mathscr{H}):= \C\oplus \mathscr{H} \oplus (\mathscr{H}\otimes_A\mathscr{H})\oplus\dots
\end{equation}and $\otimes_A$ denotes the antisymmetrized tensor product. 

Let $\{u_i\}$ be an orthonormal basis of $\mathscr{H}_+$ and $\{v_i\}$ be an orthonormal basis of $\mathscr{H}_-$. We can write the Dirac field operator\footnote{As we are considering classical Dirac fields in the remaining sections of the paper, $\Psi$ should not be interpreted as a quantum field theory operator outside of this section!} $\Psi$ and the conjugate Dirac field operator $\overline{\Psi}$ as \begin{align}
    & \Psi = \sum_i (u_i a_i + v_i b_i^\dagger),&& \overline{\Psi} = \sum_i(\bar v_i b_i + \bar u_i a_i^\dagger), \label{eq:def_quantum_dirac_operator}
\end{align} where $a_i,a_i^\dagger$ (resp. $b_i,b_i^\dagger$) are the annihilation and creation operators of a Dirac particle (resp. antiparticle). These satisfy the canonical anticommutation relations $\{a_i,a_j^\dagger\} = \{b_i,b_j^\dagger\}=\delta_{ij}$, and all other pairs of the creation/annihilation operators anticommute. Any normalized, positive-frequency classical solution $\psi_c$ can be expanded in terms of the basis $\{u_i\}$ as $\psi_c = \sum_i \langle u_i,\psi_c\rangle u_i,$ and we define the associated one-particle state by \begin{equation}
    \ket{\psi_c} := \sum_i \langle u_i,\psi_c\rangle a^\dagger_i\ket{\Omega},
\end{equation} where $\ket{\Omega}$ is the vacuum state annihilated by all the $a_i,b_i$. Normalized, classical, positive-frequency solutions therefore correspond to single-particle states of the Dirac field.

We now turn to the interpretation of the classical stress-energy of a classical Dirac solution. In the usual semiclassical treatment of the Einstein-Dirac system, one defines the Dirac stress-energy tensor as an operator \begin{equation}
    \hat T_{ab} =\frac{i}{2}\overline{\Psi} \gamma_{(a}D_{b)} \Psi + \text{h.c.},\label{eq:quantum_dirac_stress_energy}
\end{equation}where `$+\text{h.c.}$' refers to adding the Hermitian conjugate of the previous term. This is the quantum field theory analog of the classical Dirac stress-energy tensor, introduced in \S \ref{sec:general_dirac_equation}. Using this operator, one seeks solutions to $G_{ab} = \langle \hat T_{ab}^\text{ren}\rangle_\omega$, where $\langle\cdot\rangle_\omega$ denotes the expectation value in a chosen Hadamard state, and $\hat T_{ab}^\text{ren}$ is an appropriately renormalized stress-energy tensor. When working with classical spinor fields, we are instead solving $G_{ab} = T_{ab}^\text{classical}$, where $T_{ab}^\text{classical}$ is the stress-energy tensor of a classical Dirac solution $\psi_c$, given by eq. \eqref{eq:dirac_stress_energy}. We will now show that \begin{align}
T^\text{classical}_{ab} = \expval{\normord{\hat T_{ab}}}{\psi_c},
\end{align} where $\normord{\hat T_{ab}} \  = \hat T_{ab} - \expval{\hat T_{ab}}{\Omega}$ is the normal-ordered stress-energy tensor. 

To see this, insert the definition of the field operators (eq. (\ref{eq:def_quantum_dirac_operator})) into the Dirac bilinear appearing in the definition of the stress-energy tensor (eq. (\ref{eq:quantum_dirac_stress_energy})), which gives 
\begin{equation}
    \begin{aligned}
    &\overline{\Psi}\gamma_{(a}D_{b)}\Psi=\sum_{i,j}[(\bar v_i \gamma_{(a}D_{b)} u_j)b_ia_j + (\bar u_i\gamma_{(a}D_{b)} v_j)a_i^\dagger b_j^\dagger \\
    &\qquad \qquad\qquad\qquad  + (\bar u_i \gamma_{(a}D_{b)} u_j)a_i^\dagger a_j + (\bar v_i\gamma_{(a}D_{b)} v_j)b_ib_j^\dagger] .
\end{aligned}\label{eq:expanded_quantum_stress_energy}
\end{equation}  Taking the normal-ordered expectation value of eq. (\ref{eq:expanded_quantum_stress_energy}) in a one-particle state gives \begin{align}
    \expval{\normord{ \overline{\Psi}\gamma_{(a}D_{b)}\Psi}}{\psi_c}= \bar\psi_c\gamma_{(a}D_{b)}\psi_c.
 \end{align} It follows that, for a normalized, positive-frequency classical solution, the classical Dirac stress-energy tensor tells us what the normal-ordered expectation value would be for a fully quantum single-particle state, i.e., \begin{equation}
     \expval{\normord{\hat T_{ab}}}{\psi_c} = T_{ab}^\text{classical}(\psi_c)=\frac{i}{2}\bar\psi_c \gamma_{(a} D_{b)} \psi_c+\text{h.c.}\label{eq:classical_quantum_stress_energy_corr} 
 \end{equation} for normalized, positive-frequency classical solutions $\psi_c$. 
 
 In \S \ref{sec:spherical_sym_stress_energy}, we will need to consider incoherent superpositions of one-particle states in order to obtain a spherically symmetric Dirac stress-energy tensor. The correspondence between normalized, positive-frequency classical solutions and one-particle states of the quantum theory extends naturally to such statistical mixtures. More precisely, let $\{\psi_c^{(\alpha)}\}$ be a collection of normalized, positive-frequency classical solutions and let $p_\alpha \geq 0$ with $\sum_\alpha p_\alpha = 1$. The density matrix \begin{align}
     \hat \rho :=\sum_\alpha p_\alpha \ket{\psi_c^{(\alpha)}}\bra{\psi_c^{(\alpha)}}
 \end{align} describes an incoherent superposition of one-particle states. By linearity of the trace and eq. \eqref{eq:classical_quantum_stress_energy_corr}, its normal-ordered expected stress-energy is given by \begin{equation}
 \begin{aligned}
     \Tr\left(\hat \rho\normord{{\hat T}_{ab}}\right) &= \sum_\alpha p_\alpha \expval{\normord{\hat T_{ab}}}{\psi_c^{(\alpha)}}= \sum_\alpha p_\alpha T_{ab}^\text{classical}\left(\psi_c^{(\alpha)}\right).
 \end{aligned}
 \end{equation} Thus, the expected stress-energy of an incoherent superposition is the corresponding weighted sum of the classical stress-energy tensors of its constituent normalized, positive-frequency solutions. In particular, there are no interference terms between distinct solutions. This differs from a coherent superposition of classical solutions, whose stress-energy generally contains cross terms between the different solutions.
 
 If instead we wanted to consider a single antiparticle state, we would start with a negative-frequency, normalized, classical solution $\chi_c$, which can be expanded in terms of the $\{v_i\}$ as $\chi_c = \sum_i\langle v_i,\chi_c\rangle v_i.$ Defining a single antiparticle state by $\ket{\chi_c}:= \sum_i \langle \chi_c,v_i\rangle b_i^\dagger\ket{\Omega},$ we find \begin{align}
    \expval{\normord{ \overline{\Psi}\gamma_{(a}D_{b)}\Psi}}{\chi_c}=- \bar\chi_c\gamma_{(a}D_{b)}\chi_c,
 \end{align} where the minus sign comes from $\normord{b_ib_j^\dagger}\  = -b_j^\dagger b_i$. It follows that \begin{equation}
     \expval{\normord{\hat T_{ab}}}{\chi_c} = -T_{ab}^\text{classical}(\chi_c)=-\frac{i}{2}\bar\chi_c \gamma_{(a} D_{b)} \chi_c+\text{c.c.} 
 \end{equation} for normalized, negative-frequency classical solutions $\chi_c$.

Note that, by the definition of normal-ordering, the vacuum expectation value of the Dirac stress-energy is neglected when we use the classical stress-energy. In other words, our approach does not capture contributions to the stress-energy coming from the vacuum itself, but only contributions to the quantum stress-tensor arising from single-particle excitations above the vacuum. We do not argue that the vacuum contributions to the stress-energy tensor should be unimportant; in fact, such contributions are essential in the semiclassical wormhole solutions appearing in \cite{MaldacenaMilekhinPopov2023TraversableWormholes}. In this work, however, we are isolating whether the appropriate contributions to the stress-energy could have arisen from the single-particle sector itself.
 
 In summary, classical Dirac fields have the following physical interpretation. For a normalized, positive-frequency classical solution, the classical stress-energy of this solution agrees with the normal-ordered expectation value of the stress-energy in the associated single-particle state. Similarly, for a normalized, negative-frequency classical solution, \textit{minus} the classical stress-energy of this solution agrees with the normal-ordered expectation value of the stress-energy in the associated single-antiparticle state. The wormhole solutions given in \cite{Konoplya:2021hsm,Kain2023EDMWormholesQFT} are all negative-frequency solutions with the classical stress-energy tensor (eq. (\ref{eq:dirac_stress_energy})), and so a physical interpretation of these classical Dirac solutions does not exist. In this paper, we will restrict to normalizable,\footnote{Physical solutions are normalized to 1. On a fixed background, the Dirac equation is linear, so a normalizable solution can be freely rescaled to have norm 1. Even in the fully backreacting Einstein-Dirac case, where the equations become nonlinear, we will show in \S \ref{sec:asym_initial_data} that an appropriate rescaling of the coordinates, metric, and Dirac solution can be performed to rescale the Dirac norm of a normalizable solution to 1. For this reason, we will only require that our Dirac solutions be normalizable, but will not enforce a unit-norm requirement unless explicitly stated.} positive-frequency Dirac solutions with the classical stress-energy given by eq. (\ref{eq:dirac_stress_energy}), but all the calculations and results readily extend to negative-frequency Dirac solutions with a stress-energy given by minus the stress-energy tensor in eq. (\ref{eq:dirac_stress_energy}).

\section{Dirac Fields on a Fixed Wormhole Background}\label{ch:dirac_fixed_background} 
Our aim in this section is to cleanly formulate the behavior of a Dirac field in a static, spherically symmetric, asymptotically flat traversable wormhole spacetime, and to show that it need not satisfy the pointwise or averaged null energy conditions. We begin by describing, in detail, the geometry of a static and spherically symmetric traversable wormhole spacetime in \S \ref{sec:traversable_wormhole_spacetime_geometry}. Then, in \S \ref{sec:dirac_equation}, we apply the machinery of \S \ref{sec:general_dirac_equation} to write the Dirac equation for this spacetime, making use of the spacetime symmetries to reduce the content of the Dirac equation to a system of first-order ordinary differential equations. We also present explicit expressions for the Dirac current and stress-energy tensor. Finally, in \S \ref{sec:Dirac_Fields_Can_Violate_ANEC}, we explain how one can obtain counterexamples to the null energy condition for the Dirac field on a fixed traversable wormhole background, as well as how to construct traversable wormhole geometries on which the averaged null energy condition is violated by Dirac solutions.

\subsection{\label{sec:traversable_wormhole_spacetime_geometry}Traversable Wormhole Spacetimes}

Here, we will describe the geometry of a traversable wormhole spacetime. To make the analysis tractable, we restrict our consideration to static, spherically symmetric, and asymptotically flat traversable wormhole geometries. It is convenient to use a time coordinate $-\infty<t<\infty$ associated with the static Killing field $\xi^a$ by $\xi^a = (\partial/\partial t)^a$, and to denote its norm by $f$. Spherical symmetry allows us to foliate our spacetime by two-spheres, on which we introduce the usual polar and azimuthal angular coordinates $0\leq\theta<\pi$ and $0\leq \varphi<2\pi$, respectively. In spherically symmetric spacetimes, one usually chooses the last coordinate to be the Schwarzschild radial coordinate $r$ describing the size of each two-sphere, but in wormhole geometries $r$ must be nonmonotonic, and it will not be a convenient coordinate to use. Instead, we introduce a well-behaved radial coordinate $-\infty<x<\infty$. The most general static, spherically symmetric line element is then given by \begin{equation}
    ds^2 =  f(x)^2 \ \dt^2 - h(x)^2\ \dx^2 - r(x)^2[\dtheta^2+\sin^2\theta \ \dphi^2].
\end{equation} 

It will prove convenient to choose the coordinate $x$ to be an affine parameter along radial null geodesics, in which case we have $h(x) = 1/f(x)$. To see why, let $k^a = \alpha (\partial/\partial t)^a + \beta (\partial/\partial x)^a$ be tangent to an affinely parameterized radial null geodesic. Since $k^a$ satisfies the geodesic equation $k^a\nabla_ak^b = 0$, then along the geodesics generated by $k^a$ the inner product of $k^a$ and the timelike Killing field $\xi^a$ is constant, i.e., $k^a\nabla_a(k^b\xi_b) = 0$. This allows us to freely rescale $k^a$ such that $k^a\xi_a = 1$, implying that $\alpha = 1/f^2$. Next, if $x$ is the affine parameter associated with $k^a$, then $\beta = k^a\nabla_a x = 1$. Finally, the requirement that $k^a$ is null implies $\alpha^2 f^2 - \beta^2 h^2 = 0$, which, together with $\alpha = 1/f^2$ and $\beta = 1$, implies that $h = 1/f$. In conclusion, we have that the metric of an asymptotically flat traversable wormhole spacetime $(M,g)$ can be expressed in terms of the line element \begin{equation}
    ds^2 =  f(x)^2 \ \dt^2 - \frac{\dx^2}{f(x)^2} - r(x)^2[\dtheta^2+\sin^2\theta \ \dphi^2].\label{eq:metric}
\end{equation} 

The function $r(x)$ has a clear geometric interpretation, as it describes the radius of a 2-sphere that is invariant under the flow of the rotational Killing fields. To describe a wormhole, $r(x)$ must be nonmonotonic: as $x$ ranges from $-\infty$ to $+\infty$, $r(x)$ must decrease from $+\infty$ down to some minimal value $r_t>0$, and then increase again to $+\infty$ \cite{Morris:1988tu}. We call $r_t$ the ``throat radius,'' as it describes the size of the narrowest point of the wormhole. We will not assume that $r$ is monotonic on either side of its minimum unless otherwise stated, so it is possible to have several local minima of $r$. For the metric (\ref{eq:metric}) to describe a \textit{traversable} wormhole, there must be no horizons, which is enforced by requiring $f > 0$ everywhere. For the wormhole to be asymptotically flat, we require that \begin{align}
    r(x) \sim \pm \frac{x}{f_0^\pm} + \mathcal{O}(1), && f(x) \sim f_0^\pm + \mathcal{O}(1/x)
\end{align} as $x\to\pm \infty$. Because the wormhole has two asymptotically flat ends, we cannot, in general, simultaneously normalize $f_0^+$ and $f_0^-$ to 1.

\subsection{Dirac Fields on a Static and Spherically Symmetric Spacetime} \label{sec:dirac_equation} 

\subsubsection{The Dirac Equation} 
\label{sec:explicit_dirac_equation}

Let us now specialize the discussion of \S \ref{sec:general_dirac_equation} to the wormhole spacetime metric presented in eq. (\ref{eq:metric}). Our goal in this section is to present the explicit form of the Dirac equation that we will analyze in this paper, using the static and spherical symmetry of our spacetime to reduce the Dirac equation to a system of coupled first-order ordinary differential equations. We will also introduce a set of spinor components characterized by their behavior under the discrete symmetries of parity and time-reversal admitted by the spacetimes under consideration, which will play an essential role in the analysis for \S \ref{ch:static_spherically_symmetric_dirac_solutions} onward. Analogous simplifications of the Dirac equation exist in any static, spherically symmetric spacetime.

A natural choice of tetrad is obtained by normalizing the coordinate vector fields introduced in \S \ref{sec:traversable_wormhole_spacetime_geometry}\begin{equation}
    \begin{aligned}
    &{e_{0}}^a = f(x)^{-1}\left(\frac{\partial}{\partial t}\right)^a, &&{e_{1}}^a = r(x)^{-1}\left(\frac{\partial}{\partial \theta}\right)^a,   \\
    & {e_{2}}^a = r(x)^{-1}\csc\theta \left(\frac{\partial}{\partial\varphi}\right)^a,&&{e_{3}}^a = f(x)\left(\frac{\partial}{\partial x}\right)^a.
\end{aligned} \label{eq:non_coordinate_basis_def}
\end{equation} Using the components introduced in eq. (\ref{eq:two_spinor_components}) or eq. (\ref{eq:spinor_formalisms_identification}), the Dirac equation can be written as \begin{equation}
   \begin{aligned}
   \frac{\mu}{\sqrt{2}}(r f^{1/2} X_{0^\prime}) &= l^a\nabla_a(rf^{1/2}\Phi_0) + \edth^\prime(rf^{1/2}\Phi_1),\\
   \frac{\mu}{\sqrt{2}}(rf^{1/2} X_{1^\prime})&= n^a\nabla_a (rf^{1/2}\Phi_1)+\edth(rf^{1/2}\Phi_0),\\
   -\frac{\mu}{\sqrt{2}}(rf^{1/2}\Phi_1) &= l^a\nabla_a(rf^{1/2} X_{1^\prime}) - \edth(rf^{1/2} X_{0^\prime}), \\
    -\frac{\mu}{\sqrt{2}}(rf^{1/2}\Phi_0) &= n^a\nabla_a  (rf^{1/2} X_{0^\prime}) - \edth^\prime(rf^{1/2} X_{1^\prime}),\label{eq:full_dirac_equation_in_components}
\end{aligned}
\end{equation} where \begin{align}
    &l^a = \frac{1}{\sqrt{2}}\left[{e_0}^a + {e_3}^a\right], && n^a = \frac{1}{\sqrt{2}}\left[{e_0}^a - {e_3}^a\right]\label{eq:radial_null_vectors}
\end{align} are the radial null directions,\footnote{These correspond to future-directed null directions threading the wormhole from opposite directions, in such a way that they are orthogonal to the spatial two-spheres at each point.} and the angular derivative operators $\edth$ and $\edth^\prime$, defined by \cite{GHP:1973} \begin{equation}
    \begin{aligned}
    &\edth\eta  = \frac{1}{\sqrt{2}r(x)}\left(\frac{\partial}{\partial\theta} +i\csc\theta\frac{\partial}{\partial\varphi}+\frac{1}{2}\cot\theta\right)\eta ,\\
    &\edth^\prime \eta  = \frac{1}{\sqrt{2}r(x)}\left(\frac{\partial}{\partial\theta} -i\csc\theta\frac{\partial}{\partial\varphi}+\frac{1}{2}\cot\theta\right)\eta ,
\end{aligned}
\end{equation} are introduced in appendix \ref{app:spin_weighted_Spherical_Harmonics}. If we define rescaled spinor components, distinguished by a ``hat'' \begin{align}
    &\hat\Phi_i:= rf^{1/2}\Phi_i,&& \hat X_i:= rf^{1/2}X_i,
\end{align} then eq. (\ref{eq:full_dirac_equation_in_components}) simplifies to \begin{equation}
   \begin{aligned}
   &\frac{\mu}{\sqrt{2}}\hat  X_{0^\prime} = l^a\nabla_a\hat\Phi_0 + \edth^\prime\hat\Phi_1,&&\frac{\mu}{\sqrt{2}}\hat X_{1^\prime}= n^a\nabla_a \hat \Phi_1+\edth\hat\Phi_0,\\
   -&\frac{\mu}{\sqrt{2}}\hat \Phi_1 = l^a\nabla_a\hat X_{1^\prime} - \edth\hat X_{0^\prime},&&-\frac{\mu}{\sqrt{2}}\hat\Phi_0= n^a\nabla_a  \hat X_{0^\prime} - \edth^\prime\hat X_{1^\prime}.
\end{aligned}\label{eq:rescaled_dirac_equation_in_components}
\end{equation}

We also define the rescaled 4-component Dirac spinor $\hat\Psi$ to have components given by eq. (\ref{eq:spinor_formalisms_identification}), but using rescaled two-spinor components $\hat X_{i}$ and $\hat \Phi_i$.

We will now make use of the time-translation and spherical symmetries of our spacetime to obtain a further simplified Dirac equation. In stationary spacetimes, it is convenient to take the Fourier transform of each spinor component with respect to Killing time $t$, e.g., \begin{align}
    \hat \Phi_0^{(\omega)}(x,\theta,\varphi) = \int_{-\infty}^\infty dt\ e^{i\omega t}\hat \Phi_0(t,x,\theta,\varphi), 
\end{align} and similarly for the other components. In a static spacetime, the various frequencies $\omega$ decouple from one another in the Dirac equation, allowing us to consider each fixed $\omega$ independently.

Now, let us make use of spherical symmetry by expanding the components in terms of spherical harmonics. The components $(\hat\Phi_0^{(\omega)},\hat\Phi_1^{(\omega)},\hat X_{0^\prime}^{(\omega)},\hat X_{1^\prime}^{(\omega)})$ carry nontrivial spin weights, so we must expand their angular dependence in terms of spin-weight $\pm1/2$ spherical harmonics\footnote{The spin-weighted spherical harmonics are a generalization of the usual scalar/vector/tensor spherical harmonics that account for the fact that these components do not come back to themselves under a $2\pi$-rotation, but instead pick up an overall sign. See appendix \ref{app:spin_weighted_Spherical_Harmonics} for details.} as follows: \begin{equation}
    \begin{aligned}
        &\hat\Phi_0^{(\omega)}(x,\theta,\varphi)=\sum_{\ell,m}\phi_0^{(\omega,\ell,m)}(x)\ _{-1/2}Y_{\ell,m}(\theta,\varphi),\\
        &\hat\Phi_1^{(\omega)}(x,\theta,\varphi)=\sum_{\ell,m}\phi_1^{(\omega,\ell,m)}(x)\ _{+1/2}Y_{\ell,m}(\theta,\varphi),\\
        &\hat X_{0^\prime}^{(\omega)}(x,\theta,\varphi)=\sum_{\ell,m}\chi_{0^\prime}^{(\omega,\ell,m)}(x)\ _{-1/2}Y_{\ell,m}(\theta,\varphi),\\
        &\hat X_{1^\prime}^{(\omega)}(x,\theta,\varphi)=\sum_{\ell,m}\chi_{1^\prime}^{(\omega,\ell,m)}(x)\ _{+1/2}Y_{\ell,m}(\theta,\varphi),
    \end{aligned} \label{eq:spinor_component_decomposition}
\end{equation} where the sum over $(\ell,m)$ is taken over $\ell,m\in\Z+\frac{1}{2}$ satisfying $\ell\geq1/2$ and $|m|\leq \ell.$ In the resulting radial Dirac equation, the distinct $(\ell,m)$ decouple from one another, and the equations become ordinary differential equations in $x$ which depend only on $\omega$ and $\ell$:
 \begin{equation}
    \begin{aligned}
        \mu  \chi _{0^\prime}^{(\omega,\ell,m)}(x)&=f(x) \frac{d}{dx}\phi _0^{(\omega,\ell,m)}(x)-\frac{i \omega  }{f(x)}\phi _0^{(\omega,\ell,m)}(x)+\frac{(\ell+1/2)}{r(x)}\phi _1^{(\omega,\ell,m)}(x),\\
        -\mu\chi_{1^\prime}^{(\omega,\ell,m)}(x)&=f(x)\frac{d}{dx} \phi _1^{(\omega,\ell,m)}(x)+\frac{i \omega }{f(x)} \phi _1^{(\omega,\ell,m)}(x)+\frac{(\ell+1/2)}{r(x)}\phi
   _0^{(\omega,\ell,m)}(x),\\
   \mu\phi_0^{(\omega,\ell,m)}(x)&=f(x) \frac{d}{dx}\chi _{0^\prime}^{(\omega,\ell,m)}(x)+\frac{i \omega }{f(x)}\chi _{0^\prime}^{(\omega,\ell,m)}(x)+\frac{(\ell+1/2) }{r(x)}\chi _{1^\prime}^{(\omega,\ell,m)}(x),\\
   -\mu  \phi _1^{(\omega,\ell,m)}(x)&=f(x)\frac{d}{dx} \chi
   _{1^\prime}^{(\omega,\ell,m)}(x)-\frac{i \omega }{f(x)} \chi _{1^\prime}^{(\omega,\ell,m)}(x)+\frac{(\ell+1/2) }{r(x)}\chi _{0^\prime}^{(\omega,\ell,m)}(x).
    \end{aligned} \label{eq:reduced_dirac_equation}
\end{equation} 

We will formulate the remainder of this section in terms of the radial spinor components used in eq. \eqref{eq:reduced_dirac_equation}, but from Sec. \ref{ch:static_spherically_symmetric_dirac_solutions} onwards it will prove more convenient to use a different set of radial spinor components which take advantage of the fact that static, spherically symmetric spacetimes admit two discrete symmetries. Working in a basis of Dirac components with a nice behavior under these discrete symmetries further simplifies the analysis and form of the Dirac equation. In spherically symmetric spacetimes, the metric admits a discrete symmetry under the antipodal map on each two-sphere $\Upsilon$ (discussed in detail in appendix \ref{app:antipodal}).  This has a corresponding action on solutions $\Psi$ to the Dirac equation, taking them to a solution $\Upsilon^*\Psi\vert_{(t,x,\theta,\varphi)}:= i\gamma^0\gamma^2\gamma^3\Psi\vert_{\Upsilon(t,x,\theta,\varphi)}$. Importantly, $(\Upsilon^*)^2\Psi = -\Psi$, since it maps the spinor back to itself at the point $\varphi+2\pi$, which induces a sign flip for spin $1/2$ fields.\footnote{In standard quantum field theory conventions of, e.g. \cite{PeskinSchroeder:1995}, the parity map squares to 1. The complex Dirac field also has a $U(1)$ freedom associated with multiplication by a constant phase $\eta$. Replacing $\Upsilon^* \mapsto \eta\Upsilon^*$ changes its square by $\eta^2$ while leaving its action on the spacetime frame and on all Dirac bilinears (and hence on all physical quantities) unchanged. The choice $\eta = \pm i$ recovers the convention of \cite{PeskinSchroeder:1995}. We use the phase that arises from the geometric construction of appendix \ref{app:antipodal}, though the sign of $(\Upsilon^*)^2$ has no physical consequence in this work. See appendix \ref{app:antipodal} for further discussion. } Since the Dirac equation is linear, any linear combination of $\Psi$ and $\Upsilon^*\Psi$ also solves the Dirac equation. In particular, \begin{align}
    \Psi_\pm := \frac{1}{\sqrt{2}}(\Psi\mp i\Upsilon^*\Psi)\label{eq:def_of_sphere_parity_even_odd_solutions}
\end{align} are both solutions and eigenstates of $\Upsilon^*$ with eigenvalues $\pm i$, respectively. Since any $\Psi$ can be written as $\Psi = \frac{1}{\sqrt{2}}(\Psi_++\Psi_-)$, we call a solution  \textit{sphere-parity\footnote{Suppose we have a wormhole geometry where the throat is at $x=0$. In some cases, we will consider reflection-symmetric wormholes, where the metric is invariant under $x\to-x$. In such cases, it is unclear whether the word ``parity'' should refer to the antipodal map $\Upsilon$ or the reflection map $x\to-x$, so we refer to the former as ``sphere-parity'' and the latter as ``reflection'' to avoid potential confusion.} even} (resp. \textit{odd}) if $\Psi_-=0$ (resp. $\Psi_+=0$). 

Given a solution $\Psi$ on a fixed background geometry, we can construct the associated linear combination $\Psi_\pm$ and consider this new solution instead, allowing us to restrict to either sphere-parity even or sphere-parity odd solutions. If we work instead with components of $\Psi_\pm$, then (after taking the Fourier transform in $t$, expanding in spin-weighted spherical harmonics, and rescaling by $rf^{1/2}$ as we did for the components $(\hat\Phi_0,\hat\Phi_1,\hat X_{0^\prime},\hat X_{1^\prime})$) we obtain a Dirac equation in which the components decouple into sphere-parity even and odd parts:
\begin{equation}
    \begin{aligned}
    &0= \left[f\frac{d}{dx} + \frac{i\omega}{f}\sigma^3 + \frac{(\ell+1/2)}{r}\sigma^1 + \kappa_\ell \mu \sigma^2\right]\begin{pmatrix}
        \tilde\psi_0^{(\omega,\ell,m)}\\
        \tilde\psi_1^{(\omega,\ell,m)}
    \end{pmatrix},\\
    &0= \left[f\frac{d}{dx} + \frac{i\omega}{f}\sigma^3 + \frac{(\ell+1/2)}{r}\sigma^1 -\kappa_\ell \mu \sigma^2\right]\begin{pmatrix}
        \tilde\psi_2^{(\omega,\ell,m)}\\
        \tilde\psi_3^{(\omega,\ell,m)}
    \end{pmatrix},
\end{aligned}\label{eq:intermediate_dirac_equation}
\end{equation} where $\kappa_\ell :=(-1)^{\ell-1/2}$ and the new radial functions are defined by \begin{equation}\begin{aligned}
    &\tilde \psi_0^{(\omega,\ell,m)}(x):= \frac{1}{\sqrt{2}}(\phi_1^{(\omega,\ell,m)}+ i(-1)^{\ell-1/2} \chi_{0^\prime}^{(\omega,\ell,m)}), \\
    & \tilde\psi_1^{(\omega,\ell,m)}(x):= \frac{1}{\sqrt{2}}(\phi_0^{(\omega,\ell,m)}+i(-1)^{\ell-1/2} \chi_{1^\prime}^{(\omega,\ell,m)}),\\
    &\tilde \psi_2^{(\omega,\ell,m)}(x):= \frac{1}{\sqrt{2}}(\phi_1^{(\omega,\ell,m)}- i(-1)^{\ell-1/2} \chi_{0^\prime}^{(\omega,\ell,m)}) ,\\
    & \tilde\psi_3^{(\omega,\ell,m)}(x):=  \frac{1}{\sqrt{2}}(\phi_0^{(\omega,\ell,m)}-i(-1)^{\ell-1/2} \chi_{1^\prime}^{(\omega,\ell,m)}).
\end{aligned}\label{eq:pre_sphere_parity_components}
\end{equation} On sphere-parity even (resp. odd) solutions we have $\tilde\psi_2^{(\omega,\ell,m)}=\tilde\psi_3 ^{(\omega,\ell,m)}= 0$ (resp. $\tilde\psi_0^{(\omega,\ell,m)}=\tilde\psi_1^{(\omega,\ell,m)} = 0$).

Finally, because our spacetime is static, we have a second discrete symmetry given by time-reversal, which we can use to obtain a \textit{real}, radial Dirac equation. Static spacetimes admit a time-reversal symmetry $t\mapsto -t$, whose action on Dirac solutions is given by an antilinear map $T^*$ (discussed in detail in appendix \ref{app:time_reversal}). Time-reversal interchanges the radial components we introduced within each sphere-parity sector, up to complex conjugation and an overall phase. This motivates defining the following linear combinations of $(\tilde\psi_0^{(\omega,\ell,m)},\tilde\psi_1^{(\omega,\ell,m)})$ and of $(\tilde\psi_2^{(\omega,\ell,m)},\tilde\psi_3^{(\omega,\ell,m)})$\begin{equation}
    \begin{aligned}
    &\psi_0^{(\omega,\ell,m)}:= \frac{1}{\sqrt{2}}(\tilde\psi_0^{(\omega,\ell,m)} -i \tilde\psi_1^{(\omega,\ell,m)}), 
    &&\psi_1^{(\omega,\ell,m)}:= \frac{1}{\sqrt{2}}(i\tilde\psi_0^{(\omega,\ell,m)} - \tilde\psi_1^{(\omega,\ell,m)}),\\
    &\psi_2^{(\omega,\ell,m)}:= \frac{1}{\sqrt{2}}(-\tilde\psi_2^{(\omega,\ell,m)}-i\tilde\psi_3^{(\omega,\ell,m)}), && \psi_3^{(\omega,\ell,m)} := \frac{1}{\sqrt{2}}(i\tilde\psi_2^{(\omega,\ell,m)}+\tilde\psi_3^{(\omega,\ell,m)}),
\end{aligned} \label{eq:real_spinor_component_definition}
\end{equation} which we use to obtain the following radial equation with only real coefficients
    \begin{equation}
    \begin{aligned}
    &0= \left[f\frac{d}{dx} + \frac{\omega}{f}(i\sigma^2) - \frac{(\ell+1/2)}{r}\sigma^1 + \kappa_\ell \mu  \sigma^3\right]\begin{pmatrix}
        \psi_0^{(\omega,\ell,m)}\\
        \psi_1^{(\omega,\ell,m)}
    \end{pmatrix},\\
    &0= \left[f\frac{d}{dx} - \frac{\omega}{f}(i\sigma^2) - \frac{(\ell+1/2)}{r}\sigma^1 + \kappa_\ell\mu \sigma^3\right]\begin{pmatrix}
        \psi_2^{(\omega,\ell,m)}\\
        \psi_3^{(\omega,\ell,m)}
    \end{pmatrix},
\end{aligned}\label{eq:real_dirac}
\end{equation} 
where again $\kappa_\ell:=(-1)^{\ell-1/2}$. The matrices $\sigma^1,i\sigma^2,\sigma^3$ have all real entries and their coefficients are real.  Thus, using the radial components $({\psi_0}^{(\omega,\ell,m)},{\psi_1}^{(\omega,\ell,m)},{\psi_2}^{(\omega,\ell,m)},{\psi_3}^{(\omega,\ell,m)})$ results in a set of ordinary differential equations in $x$, in which $({\psi_0}^{(\omega,\ell,m)},{\psi_1}^{(\omega,\ell,m)})$ is decoupled from $({\psi_2}^{(\omega,\ell,m)},{\psi_3}^{(\omega,\ell,m)})$. Moreover, all the coefficients appearing in the differential equation are real, and thus the real and imaginary parts of solutions are decoupled.

\subsubsection{The Dirac Inner Product and Stress-Energy Tensor}
We turn now to the explicit form of the Dirac inner product. In spacetimes with a timelike Killing field, it is convenient to take our Cauchy surface in eq. (\ref{eq:Dirac_Inner_Product}) to be a surface of constant Killing time. For our static, spherically symmetric wormhole geometry (eq. (\ref{eq:metric})), this gives \begin{equation}
    \begin{aligned}
    &\langle \Psi_1,\Psi_2\rangle = \langle [X_{A^\prime},\Phi^A],[\Sigma_{A^\prime},\Xi^A]\rangle \\
    & = \int_{t=\text{const}} [(\bar\Phi_0 \Xi_0 + \bar\Phi_1 \Xi_1 +\bar X_{0^\prime}\Sigma_{0^\prime} + \bar X_{1^\prime}\Sigma_{1^\prime})f^{-1}r^2\sin\theta\ dxd\theta d\varphi],
\end{aligned}
\end{equation} where the bars on the components denote complex conjugation. Replacing these spinor components by their rescaled counterparts $\hat\Phi_i := rf^{1/2}\Phi_i$ and $\hat X_i:= rf^{1/2}X_i$ used in eq. (\ref{eq:rescaled_dirac_equation_in_components}), we find that the Dirac norm is given by \begin{equation}
\begin{aligned}
  \|\Psi\|^2&= \|[ X_{A^\prime}, \Phi^A]\|^2 \\
  &= \int_{t=\text{const}} \biggr[\left(|\hat \Phi_0|^2 +|\hat \Phi_1|^2 + |\hat X_{0^\prime}|^2 + |\hat X_{1^\prime}|^2\right)f^{-2}\sin\theta\ dxd\theta d\varphi \biggr] .\label{eq:dirac_norm}
\end{aligned}      
\end{equation}

Next, we consider the Dirac stress-energy tensor. For the analysis in this section, we will not require explicit expressions for all the stress-energy tensor components; however, as we will present violations of the pointwise and averaged null energy conditions in \S \ref{sec:Dirac_Fields_Can_Violate_ANEC}, we will need the stress-energy tensor components that appear in the formulation of these conditions. If we consider the affinely parameterized radial null geodesics threading the wormhole in different directions, their future-directed null tangent vectors are given by  \begin{align}
    &\left(\frac{\partial}{\partial \lambda}\right)^a = f(x)^{-2}\left(\frac{\partial}{\partial t}\right)^a + \left(\frac{\partial}{\partial x}\right)^a,&&\left( \frac{\partial}{\partial\tilde\lambda}\right)^a= f(x)^{-2}\left(\frac{\partial}{\partial t}\right)^a-\left(\frac{\partial}{\partial x}\right)^a. \label{eq:affine_param_null_geos}
\end{align} In terms of the rescaled spinor components of \S \ref{sec:explicit_dirac_equation}, we have \begin{equation}
    \begin{aligned}
    &T_{ab}\left(\frac{\partial}{\partial \lambda}\right)^a\left(\frac{\partial}{\partial \lambda}\right)^b=\frac{i}{\sqrt{2}f^3r^2}\left[\bar {\hat X}_{0^\prime} l^a\nabla_a\hat X_{0^\prime} - \hat \Phi_1l^a\nabla_a\bar{\hat \Phi}_1\right]+\text{c.c.}\ ,\\
    &T_{ab}\left(\frac{\partial}{\partial \tilde \lambda}\right)^a\left(\frac{\partial}{\partial \tilde \lambda}\right)^b=\frac{i}{\sqrt{2}f^3r^2}\left[\bar {\hat X}_{1^\prime} n^a\nabla_a\hat X_{1^\prime} - \hat \Phi_0n^a\nabla_a\bar{\hat \Phi}_0\right]+\text{c.c.}\ ,
\end{aligned}\label{eq:general_null_null_components}
\end{equation} where the radial null vectors $l^a$ and $n^a$ were introduced in eq. \eqref{eq:radial_null_vectors}. If the components have a definite frequency $\omega$ and spherical harmonic dependence $(\ell,m)$, these expressions become \begin{equation}
    \begin{aligned}
    T_{ab}\left(\frac{\partial}{\partial \lambda}\right)^a\left(\frac{\partial}{\partial \lambda}\right)^b &=\frac{1}{2f^2r^2}\left[\frac{\omega}{f^2}|\chi_{0^\prime}|^2+i\bar\chi_{0^\prime}\frac{\partial}{\partial x}\chi_{0^\prime}\right]\left|_{-1/2}Y_{\ell,m}\right|^2 \\
    &\qquad \qquad\qquad +\frac{1}{2f^2r^2}\left[\frac{\omega}{f^2}|\phi_1|^2-i\phi_1\frac{\partial}{\partial x}\bar \phi_1\right]\left|_{1/2}Y_{\ell,m}\right|^2+\text{c.c.}\ , \\
    T_{ab}\left(\frac{\partial}{\partial \tilde \lambda}\right)^a\left(\frac{\partial}{\partial \tilde \lambda}\right)^b&=\frac{1}{2f^2r^2}\left[\frac{\omega}{f^2}|\chi_{1^\prime}|^2-i\bar\chi_{1^\prime}\frac{\partial}{\partial x}\chi_{1^\prime}\right]\left|_{1/2}Y_{\ell,m}\right|^2 \\
    &\qquad\qquad \qquad +\frac{1}{2f^2r^2}\left[\frac{\omega}{f^2}|\phi_0|^2+i\phi_0\frac{\partial}{\partial x}\bar \phi_0\right]\left|_{-1/2}Y_{\ell,m}\right|^2+\text{c.c.}\ ,
\end{aligned}\label{eq:explicit_null_null_components}
\end{equation}where we suppressed the $(\omega,\ell,m)$ superscripts on the radial components $(\chi_{0^\prime},\chi_{1^\prime},\phi_0,\phi_1)$ to write this more compactly. Finally, we can eliminate the $x$ derivatives using eq. (\ref{eq:reduced_dirac_equation}) to obtain 
\begin{widetext}
    \begin{equation}
    \begin{aligned}
    T_{ab}\left(\frac{\partial}{\partial \lambda}\right)^a\left(\frac{\partial}{\partial \lambda}\right)^b&=\frac{1}{2f^2r^2}\left[\frac{2\omega}{f^2}|\chi_{0^\prime}|^2+\frac{i\mu\phi_0\bar\chi_{0^\prime}}{f} - \frac{i(\ell+1/2)\chi_{1^\prime}\bar\chi_{0^\prime}}{fr}\right]\left|_{-1/2}Y_{\ell,m}\right|^2 \\
    &\qquad +\frac{1}{2f^2r^2}\left[\frac{2\omega}{f^2}|\phi_1|^2-\frac{i\mu\chi_{1^\prime}\bar\phi_1}{f} - \frac{i(\ell+1/2)\phi_0\bar\phi_1}{fr}\right]\left|_{1/2}Y_{\ell,m}\right|^2+\text{c.c.}\ ,\\
    T_{ab}\left(\frac{\partial}{\partial \tilde \lambda}\right)^a\left(\frac{\partial}{\partial \tilde \lambda}\right)^b&=\frac{1}{2f^2r^2}\left[\frac{2\omega}{f^2}|\chi_{1^\prime}|^2+\frac{i\mu\phi_1\bar\chi_{1^\prime}}{f}+\frac{i(\ell+1/2)\chi_{0^\prime}\bar\chi_{1^\prime}}{fr}\right]\left|_{1/2}Y_{\ell,m}\right|^2 \\
    &\qquad +\frac{1}{2f^2r^2}\left[\frac{2\omega}{f^2}|\phi_0|^2-\frac{i\mu\chi_{0^\prime}\bar\phi_0}{f}+\frac{i(\ell+1/2)\phi_1\bar\phi_0}{fr}\right]\left|_{-1/2}Y_{\ell,m}\right|^2+\text{c.c.}\ 
\end{aligned}\label{eq:pointwise_NEC}
\end{equation} 
\end{widetext}

\subsection{Dirac Fields Can Violate the Pointwise and Averaged Null Energy Conditions}\label{sec:Dirac_Fields_Can_Violate_ANEC}

As discussed in the introduction, traversable wormhole solutions in general relativity require violations of the pointwise and averaged null energy conditions. If one could prove that Dirac fields on a fixed wormhole background satisfy either the pointwise or averaged null energy condition, then immediately one could conclude that Dirac fields are not a possible source for the stress-energy of a traversable wormhole spacetime. 

Dirac fields can easily violate the null energy condition at a point, which we will now show by an explicit example. Note that, because the Dirac equation is a first-order ordinary differential equation, the radial spinor components $(\chi_{0^\prime},\chi_{1^\prime},\phi_0,\phi_1)$ at a point are freely specifiable data. For a given background metric, one can choose values for the spinors and the parameters $\omega,\mu,$ and $\ell$ to make eq. (\ref{eq:pointwise_NEC}) negative. For example, suppose that at a given point $f=r=1$; setting $\omega = 1$, $\mu =1$, $\ell=m=1/2$, and \begin{align}
    &\chi_{0^\prime}(0) = 3, &&\chi_{1^\prime}(0) = 0, &&&\phi_0(0) = i,&&& \phi_1(0) = 0,
\end{align} results in a value for \begin{align}
    T_{ab}\left(\frac{\partial}{\partial \lambda}\right)^a\left(\frac{\partial}{\partial \lambda}\right)^b\biggr\vert_{\theta = 0}= -\frac{1}{2\pi}<0,
\end{align}  so the null energy condition is clearly violated. Examples of null energy condition violations for the Dirac field at a point are easy to obtain in a wide range of background geometries, so this example is not a special case.

For this solution to have a legitimate semiclassical interpretation, it must be normalizable in the Dirac norm of eq. (\ref{eq:dirac_norm}), which is not immediately clear from prescribing its initial data at a point. As $x \to \pm \infty$, asymptotic flatness gives the leading-order behavior of $f(x)\sim f_0^\pm+\mathcal{O}(1/x)$ and $r(x)\sim \mathcal{O}(|x|)$ for the metric functions.\footnote{In an asymptotically flat spacetime, one usually chooses to normalize the timelike Killing field by having its norm $f$ go to 1 at infinity. However, as we now have two asymptotic infinities (one at each end of the wormhole), it is not generally possible to simultaneously normalize $f$ at both ends.} Replacing the metric functions by their asymptotic values in eq. (\ref{eq:reduced_dirac_equation}) tells us that the leading-order behavior of each spinor component is \begin{align}
    \phi_0 \sim A^\pm e^{x\sqrt{\mu^2-\omega^2/{f_0^\pm}^2}}+B^\pm e^{-x\sqrt{\mu^2-\omega^2/{f_0^\pm}^2}}
\label{eq:asymptotic_dirac_solution_behavior}\end{align} as $x\to\pm\infty$, where $A^\pm,B^\pm\in\C$, and similarly for the other components.\footnote{A more precise asymptotic falloff on the metric and spinor components is reserved for \S \ref{sec:asym_initial_data}, but this crude estimate suffices for the purposes of the discussion here.} This means that, for any solution with $\mu < \omega/f_0^\pm$, the behavior of the components $(\hat\Phi_0,\hat\Phi_1,\hat X_{0^\prime},\hat X_{1^\prime})$ appearing in eq. (\ref{eq:dirac_norm}) will be oscillatory. Thus, any Dirac solution with $\mu  < \omega/f_0^\pm$ will not be normalizable if we consider stationary solutions with a single frequency $\omega$.

Nevertheless, we will now show that normalizable solutions exist which violate the null energy condition. In the null-energy-condition-violating example given above, let us choose the background geometry so that the solution is asymptotically oscillatory by choosing sufficiently small asymptotic values for $f$. Then we can construct a normalizable solution violating the null energy condition by forming a wavepacket, i.e., by coherently superposing a range of frequencies to obtain a normalizable solution, similar to how superposing non-normalizable plane wave solutions in flat spacetime can result in normalizable solutions to the wave equation. More explicitly, let $\Psi_\omega$ denote a smooth family of positive-frequency, stationary (but not necessarily normalizable) Dirac solutions labeled by their frequency $\omega$. We choose these solutions to have the standard delta-function normalization in the Dirac inner product (eq. \eqref{eq:Dirac_Inner_Product}) 
\begin{equation}
    \langle\Psi_\omega,\Psi_{\omega^\prime}\rangle = \delta(\omega-\omega^\prime).
\end{equation} For any smooth, compactly supported, positive-frequency profile $a(\omega)$, the wavepacket \begin{equation}
    \Psi[a] := \int_{0}^\infty a(\omega)\Psi_\omega\,d\omega
\end{equation} will have norm \begin{align}
    \langle \Psi[a],\Psi[a]\rangle = \int_{0}^\infty |a(\omega)|^2\, d\omega < \infty.
\end{align}
 
Define the null-energy kernel by \begin{equation}
\mathcal{A}(\omega,\omega^\prime) :=   T_{ab}[\Psi_\omega,\Psi_{\omega^\prime}]\left(\frac{\partial}{\partial\lambda}\right)^a\left(\frac{\partial}{\partial\lambda}\right)^b \end{equation} where $T_{ab}[\Psi,\Psi^\prime]$ denotes the polarized Dirac stress-tensor, treated as a Hermitian sesquilinear form. We must have \begin{equation}
    \overline{\mathcal{A}(\omega,\omega^\prime)} = \mathcal{A}(\omega^\prime,\omega)
\end{equation} so that $\mathcal{A}(\omega,\omega)$ is real. Suppose that, for some $\omega_0>0$, the corresponding stationary Dirac solution $\Psi_{\omega_0}$ strictly violates the null energy condition: \begin{equation}
\mathcal A(\omega_0,\omega_0)<0 .
\end{equation} By continuity, for all $\epsilon>0$ there exists some $\delta_0 > 0$ for which \begin{equation}
    |\mathcal{A}(\omega,\omega^\prime) - \mathcal{A}(\omega_0,\omega_0)|< \epsilon
\end{equation}whenever \begin{align}
    &|\omega-\omega_0|<\delta_0,&& |\omega^\prime-\omega_0|<\delta_0.
\end{align} Fix $0<\epsilon < |\mathcal{A}(\omega_0,\omega_0)|/2$, and let $h\not \equiv 0$ be a smooth, non-negative frequency profile with compact support on the interval $[-1,1]$. For $0<\delta<\min\{\delta_0,\omega_0\}$, we define
\begin{equation}
a_\delta(\omega)= \delta^{-1/2}
h\left(\frac{\omega-\omega_0}{\delta}\right),
\end{equation} which has support contained in an interval of radius $\delta$ about $\omega_0$, and hence is supported only on positive frequencies. We can use this to form a normalizable, positive-frequency Dirac solution
\begin{equation}
\Psi_\delta :=\Psi[a_\delta] = \int_{0}^\infty a_\delta(\omega)\Psi_\omega\,d\omega 
\end{equation} whose null energy is \begin{equation}
    \begin{aligned}
    \text{NEC}[\Psi_\delta]  &=\iint \overline{a_\delta(\omega)}a_\delta(\omega^\prime) \mathcal{A}(\omega,\omega^\prime) \, d\omega d\omega^\prime=  \iint a_\delta(\omega)a_\delta(\omega^\prime) \Re \mathcal{A}(\omega,\omega^\prime)\,d\omega d\omega^\prime.
\end{aligned}
\end{equation} Over the support of $a_\delta$, we have by continuity \begin{equation}
    \Re\mathcal{A}(\omega,\omega^\prime) < \mathcal{A}(\omega_0,\omega_0) + \epsilon,
\end{equation}from which it follows 
\begin{align}
\text{NEC}[\Psi_\delta] & 
< (\mathcal{A}(\omega_0,\omega_0) +\epsilon)\left|\int d\omega\, a_\delta(\omega)\right|^2 < \frac{1}{2}\mathcal{A}(\omega_0,\omega_0) \left|\int d\omega\, a_\delta(\omega)\right|^2 <0.
\end{align} Thus, we have shown that there exist normalizable solutions which violate the null energy condition.

A more nontrivial question is whether the Dirac field can violate the averaged null energy condition. The averaged null energy condition is given by \begin{align}
    \int_{-\infty}^\infty T_{ab}\left(\frac{\partial}{\partial\lambda}\right)^a\left(\frac{\partial}{\partial\lambda}\right)^b \ d\lambda,
\end{align} where $\lambda$ is an affine parameter along a complete, inextendible, future-directed radial null geodesic.\footnote{The radial null geodesics we are considering here are \textit{achronal}, meaning that for any two points $p$ and $q$ on the geodesic, there does not exist another causal path connecting $p$ and $q$. For consistent, semiclassical solutions to the Einstein equations, the averaged null energy condition is believed to hold along achronal null geodesics \cite{KontouSanders2020EnergyConditions}. Indeed, counterexamples are known along chronal null geodesics \cite{GrahamOlum2007AchronalANEC,Visser1996BoulwareVacuumEnergyConditions}, or when neglecting backreaction \cite{WaldYurtsever1991ANEC,Visser1995ScaleAnomaliesANEC}.} Indeed, without further restriction on the character of the solutions, we will now argue that violations of the averaged null energy condition are possible for the Dirac field. 

Consider first the case of a massless Dirac field ($\mu = 0$), and suppose we have a solution $\hat \Psi$ to the Dirac equation on a wormhole spacetime $(M,g)$. The massless Dirac equation is conformally invariant, i.e., if we choose a non-vanishing conformal factor $\Omega$ on $M$ and replace the spacetime metric $g\mapsto \Omega^2 g$, then $\hat \Psi$ will be a solution to the Dirac equation on the conformally rescaled spacetime $(M,\Omega^2 g)$.\footnote{The conformal rescaling of a massless Dirac field in a four-dimensional spacetime is $\Psi\to\Omega^{-3/2}\Psi$, but our definition of the rescaled Dirac spinor $\hat\Psi$ already contains a factor of $rf^{1/2}$ in its definition (see \S \ref{sec:explicit_dirac_equation}) to simplify the form of the Dirac equation, which cancels the factor $\Omega^{-3/2}$ by which the original spinor $\Psi$ rescales under a conformal rescaling of the metric. Thus, the rescaled spinor $\hat\Psi$ we are using in this paper is invariant under a conformal rescaling of the metric.} 

In the conformally rescaled spacetime, one also needs to rescale the affine parameter $d\lambda \mapsto \Omega^{2}\ d\lambda$ and the null tangent vectors $(\partial/\partial\lambda)^a\mapsto \Omega^{-2}(\partial/\partial\lambda)^a$ so that the geodesic is still affinely parameterized. The stress-energy tensor will also rescale by a factor $T_{ab}\to \Omega^{-2}T_{ab}$, as is required to preserve $\nabla_aT^{ab}=0$ on the conformally rescaled spacetime \cite{Wald:1984rg}, or as can be verified directly by rescaling the metric functions. Putting this all together, the averaged null energy integrand rescales by a nontrivial conformal factor \begin{equation}
    \begin{aligned}
     &\int_{-\infty}^\infty T_{ab}\left(\frac{\partial}{\partial\lambda}\right)^a\left(\frac{\partial}{\partial\lambda}\right)^b \ d\lambda\mapsto \int_{-\infty}^\infty \Omega^{-4}T_{ab}\left(\frac{\partial}{\partial\lambda}\right)^a\left(\frac{\partial}{\partial\lambda}\right)^b\ d\lambda.
\end{aligned}\label{eq:conformal_ANEC_rescaling}
\end{equation} We now exploit the freedom afforded by the arbitrary conformal factor. Take any normalizable Dirac solution that violates the null energy condition at some point; by continuity, this means the null energy condition is violated over some open interval. By choosing an appropriate conformal factor $\Omega(x)$, we can make $\Omega(x)^{-4}$ sufficiently large on this interval and sufficiently small away from this interval, while still having $\Omega(x)$ tend to a constant as $x\to\pm\infty$ (so that the resulting spacetime is still asymptotically flat). It is clear that, by judicious choice of $\Omega$, one can make the right-hand side of eq. (\ref{eq:conformal_ANEC_rescaling}) as negative as one would like, and thereby we can find a normalizable Dirac solution violating the averaged null energy condition on a wormhole spacetime. 

In the massive case $(\mu\not = 0)$, a similar argument holds. Though we lose exact conformal invariance when we consider a massive Dirac field, we expect solutions with $\mu \ll \omega/f$ everywhere to behave as massless fields to a good approximation. Therefore, massive fields can violate the averaged null energy condition by considering sufficiently high-frequency solutions.  Thus, there exist normalizable, positive-frequency Dirac solutions that violate the averaged null energy condition.

\section{Stationary, Spherically Symmetric Bound-State Dirac Solutions on Traversable Wormhole Backgrounds} \label{ch:static_spherically_symmetric_dirac_solutions}

The examples that violate the averaged null energy condition presented in the previous section are unsatisfactory in the context of static, spherically symmetric spacetimes for several reasons. First, in order to obtain normalizable counterexamples, we needed to construct a wavepacket over a range of frequencies; however, such solutions no longer give rise to time-independent observables (e.g., the stress-energy tensor) as would be required to source a static spacetime. Moreover, these solutions are not spherically symmetric, which can be seen from the appearance of the spherical harmonics in eq. (\ref{eq:explicit_null_null_components}).

In \S \ref{sec:stationary_bound_states} and \S \ref{sec:spherical_sym_stress_energy}, we will show how static, spherically symmetric Dirac solutions can be obtained. By considering an incoherent superposition of Dirac solutions with different angular dependence, the result will have both a spherically symmetric stress-energy tensor and Dirac current. Moreover, if each Dirac solution in the mixed state is itself a stationary bound state, the resulting stress-energy tensor and Dirac current will also be static with appropriate falloff to be compatible with asymptotic flatness.

Then, in \S \ref{sec:stationary_bound_states_exist}, we find numerical examples of static, spherically symmetric Dirac solutions and characterize under what conditions they exist. We will find that generically, sphere-parity even and odd bound states cannot simultaneously exist at a given frequency, unless the spacetime possesses additional symmetries. Finally, we present an example of a static, spherically symmetric bound-state Dirac solution that violates the pointwise and averaged null energy conditions.

\subsection{Static and Normalizable Dirac Solutions Require Stationary Bound States}\label{sec:stationary_bound_states}

Let us first address the question of finding normalizable solutions that yield time-independent observables. In order for the stress-energy tensor or Dirac current (which are constructed from the spinor field and its complex conjugate) to be time-independent, the spinor components must have stationary time dependence;\footnote{A further generalization of this could be considered using an incoherent superposition of stationary solutions, but we restrict our attention in this paper to solutions with a fixed frequency dependence.} each spinor component must be of the form \begin{align}
    \hat \Phi_0(t,x,\theta,\varphi) = e^{-i\omega t}\hat\Phi_0(x,\theta,\varphi)
\end{align} for a fixed $\omega\in\R$, and similarly for the other spinor components. Clearly, any observable constructed from quadratic combinations of these components multiplied by their conjugates will be time-independent, as the time-dependent phases will cancel. Expanding the angular dependence in terms of spin-weighted spherical harmonics as we did in \S \ref{sec:explicit_dirac_equation} results in the following decomposition for these solutions \begin{align}
    \hat \Phi_0(t,x,\theta,\varphi) = e^{-i\omega t}\sum_{\ell,m}{\phi_0}^{(\ell,m)}(x)\ _{-1/2}Y_{\ell,m}(\theta,\varphi),
\end{align} and similarly for the other components, as in eq. (\ref{eq:spinor_component_decomposition}). Plugging this decomposition into the Dirac norm (eq. (\ref{eq:dirac_norm})) gives \begin{equation}
    \begin{aligned}
    &\|[X_{A^\prime},\Phi^A]\|^2 \\
    &= \sum_{\ell,m} \int_{-\infty}^\infty\biggr( \frac{|{\phi_0}^{(\ell,m)}|^2}{f^2}+\frac{|{\phi_1}^{(\ell,m)}|^2}{f^2}+\frac{|{\chi_{0^\prime}}^{(\ell,m)}|^2}{f^2}+\frac{|{\chi_{1^\prime}}^{(\ell,m)}|^2}{f^2}\biggr) dx,
\end{aligned}
\end{equation} where we used the orthonormality of the spin-weighted spherical harmonics to evaluate the angular integral (see appendix \ref{app:spin_weighted_Spherical_Harmonics}).

For a Dirac solution with stationary time-dependence to be normalizable, each term in the series must be finite, so the spinor components $({\phi_0}^{(\ell,m)},{\phi_1}^{(\ell,m)},{\chi_{0^\prime}}^{(\ell,m)},{\chi_{1^\prime}}^{(\ell,m)})$ must have the appropriate falloff as $x\to\pm \infty$. In an asymptotically flat spacetime, each spinor component has a falloff that is described at leading order by eq. (\ref{eq:asymptotic_dirac_solution_behavior}). As was pointed out in \S \ref{sec:Dirac_Fields_Can_Violate_ANEC}, a stationary Dirac solution cannot be normalizable if $\omega > \min(\mu f_0^\pm)$, where $f_0^\pm := \lim_{x\to\pm \infty}f(x)$, so we will restrict to $\omega < \min(\mu f_0^\pm)$, and continue to allow only $\omega > 0$ so that the solutions have a valid semiclassical interpretation (see \S \ref{sec:semi_classical_meaning}). In the interval $0<\omega < \min(\mu f_0^\pm)$, the solutions can have exponentially growing or decaying behavior as $x\to\pm \infty$; exponential growth will also prevent the solution from being normalizable, so we need to find solutions that are exponentially decaying at both asymptotic infinities. This is only possible for particular, discrete values of $\omega$ in our interval, and such solutions are the stationary bound states.

When considering stationary bound states, it becomes particularly useful to work instead with the spinor components\footnote{As we are now treating $\omega$ as fixed, we suppressed the $\omega$ that appeared in the superscript when these components were initially introduced.}  $({\psi_0}^{(\ell,m)},{\psi_1}^{(\ell,m)},{\psi_2}^{(\ell,m)},{\psi_3}^{(\ell,m)})$ introduced in eq. (\ref{eq:real_spinor_component_definition}). These components were motivated by their behavior under sphere-parity and time-reversal and have the property that for sphere-parity even (resp. odd) solutions of the Dirac equation, we have $\psi_2^{(\ell,m)}=\psi_3 ^{(\ell,m)}= 0$ (resp. $\psi_0^{(\ell,m)}=\psi_1^{(\ell,m)} = 0$). Using the linearity of the Dirac equation on a fixed background, we can always restrict ourselves to a sphere-parity-definite solution, in which one pair of the components vanishes.

We will now show that, for stationary bound states, these spinor components can be taken to be real, making them a convenient choice for our analysis moving forward. For stationary bound states, we remarked that each of the spinor components must vanish as $x\to\pm\infty$, and a direct application of eq. (\ref{eq:real_dirac}) can be used to show that \begin{align}
    \frac{d}{dx}\Im[{\psi_0}^{(\ell,m)}{\bar\psi_1}^{(\ell,m)}]=\frac{d}{dx}\Im[ {\bar\psi_2}^{(\ell,m)}{\psi_3}^{(\ell,m)}]=0,\label{eq:conserved_Dirac_current}
\end{align} where `Im' denotes the imaginary part of a complex quantity. Since each component must vanish as $x\to\pm \infty$, we conclude that \begin{align}
    \Im[{\psi_0}^{(\ell,m)}{\bar\psi_1}^{(\ell,m)}]=\Im[ {\bar\psi_2}^{(\ell,m)}{\psi_3}^{(\ell,m)}]=0 \label{eq:component_constraint}
\end{align} everywhere. We can use this to show that, without loss of generality, we can take the components $({\psi_0}^{(\ell,m)},{\psi_1}^{(\ell,m)},{\psi_2}^{(\ell,m)},{\psi_3}^{(\ell,m)})$ to be real everywhere. Consider the pair $({\psi_0}^{(\ell,m)},{\psi_1}^{(\ell,m)})$: the condition $ \Im[{\psi_0}^{(\ell,m)}\bar {\psi_1}^{(\ell,m)}]=0$ implies that  ${\psi_0}^{(\ell,m)}(x)=e^{i\alpha(x)}R_0(x)$ and ${\psi_1}^{(\ell,m)}(x)=e^{i\alpha(x)}R_1(x)$ with $R_0,R_1$ being real-valued functions, using the same real phase $\alpha(x)$. Plugging this into eq. (\ref{eq:real_dirac}) gives
\begin{equation}
    \begin{aligned}
    0&= \left[f\frac{d}{dx} + \frac{\omega}{f}(i\sigma^2) - \frac{(\ell+1/2)}{r}\sigma^1 + \kappa_\ell \mu  \sigma^3\right]e^{i\alpha}\begin{pmatrix}
        R_0\\
        R_1
    \end{pmatrix}\\
    &=e^{i\alpha} \biggr[if\frac{d\alpha}{dx} + f\frac{d}{dx} + \frac{\omega}{f}(i\sigma^2)  - \frac{(\ell+1/2)}{r}\sigma^1 + \kappa_\ell \mu  \sigma^3\biggr]\begin{pmatrix}
        R_0\\
        R_1
    \end{pmatrix}.\end{aligned}
\end{equation} After multiplying the equation by $e^{-i\alpha(x)}$, the imaginary part of the resulting equation tells us that $d\alpha/dx = 0$, so the two components can be taken to be real functions times a constant phase. Since the Dirac equation is linear, we can freely multiply the components $({\psi_0}^{(\ell,m)},{\psi_1}^{(\ell,m)})$ by a constant, complex phase to make them real everywhere. Note that multiplication by a constant phase also leaves all Dirac bilinears (and hence all observable quantities) unchanged. An identical story holds for ${\psi_2}^{(\ell,m)}$ and ${\psi_3}^{(\ell,m)}.$ 

In summary, for a Dirac solution to be normalizable and give rise to a time-independent stress-energy and current, the solution must be one of these stationary bound states, characterized by exponential decay as $x\to\pm \infty$ and stationary time dependence. When working with stationary bound states, we will opt to use the real components $(\psi_0,\psi_1,\psi_2,\psi_3)$ for convenience.

\subsection{Constructing Spherically Symmetric Dirac Solutions}
\label{sec:spherical_sym_stress_energy}
We now turn to the issue of obtaining a Dirac solution whose observables are spherically symmetric. There are well-known difficulties with obtaining solutions to the Dirac equation with spherically symmetric stress-energy tensors or Dirac currents \cite{Alcubierre2025SphericallySymmetricSolutions}. The resolution is to take an incoherent superposition of stationary solutions with different angular dependence, in a manner that gives rise to static, spherically symmetric observables.

To obtain a spherically symmetric solution, we start with a Dirac solution that has fixed spherical harmonic dependence $(\ell,m)$. As we pointed out in the introduction to this section, observables such as the stress-energy tensor for this solution will not be spherically symmetric. Notice, however, that the equations for the radial Dirac components, presented in eq. (\ref{eq:real_dirac}), depend on $\ell$ but not on $m$. If we consider solutions for which the functions $({\psi_0}^{(\ell,m)},{\psi_1}^{(\ell,m)},{\psi_2}^{(\ell,m)},{\psi_3}^{(\ell,m)})$ are the same for all $m$, we can take an incoherent superposition of solutions by summing over all allowed values of $m$ to construct spherically symmetric observables. Considering first the Dirac stress-energy tensor (introduced in \S \ref{sec:general_dirac_equation}), let ${T_{ab}}^{(\ell,m)}$ be the stress-energy for a solution with a given spherical harmonic dependence. We define $T_{ab}^{(\ell)}$ to be the stress-energy tensor for a mixed state obtained by considering the following incoherent superposition\footnote{See \S \ref{sec:semi_classical_meaning} for a discussion as to why the sum of classical stress-energy tensors corresponds to a statistical mixture of one-particle states in the quantum theory.} over $m$: \begin{align}
    {T_{ab}}^{(\ell)}:= \frac{4\pi}{(2\ell+1)}\sum_{m=-\ell}^{\ell} {T_{ab}}^{(\ell,m)}\label{eq:spherically_symmetric_stress_tensor_def}
\end{align} (with the normalization factor of $4\pi$ chosen for convenience), and the result will be spherically symmetric on Dirac solutions. To see why this is the case, consider the component $T_{\hat t\hat t}=T_{ab}{e_{0}}^a{e_{0}}^b$ (with ${e_0}^a$ introduced in eq. (\ref{eq:non_coordinate_basis_def})), which gives the energy density of a Dirac field as measured by an observer following the timelike Killing field. For a given solution $({\psi_0}^{(\ell,m)},{\psi_1}^{(\ell,m)},{\psi_2}^{(\ell,m)},{\psi_3}^{(\ell,m)})$, we have \begin{widetext}
    \begin{equation}\begin{aligned}
   { T_{\hat t\hat t}}^{(\ell,m)} &=\frac{\omega}{2f^2r^2}\biggr[\left(|{\psi_0}^{(\ell,m)}-{\psi_2}^{(\ell,m)}|^2 + |{\psi_1}^{(\ell,m)}+{\psi_3}^{(\ell,m)}|^2\right)\left|\ _{1/2}Y_{\ell,m}\right|^2 \\
   &\qquad\qquad +\left(|{\psi_0}^{(\ell,m)}+{\psi_2}^{(\ell,m)}|^2 + |{\psi_1}^{(\ell,m)}-{\psi_3}^{(\ell,m)}|^2\right)\left|\ _{-1/2}Y_{\ell,m}\right|^2\biggr].
\end{aligned}\end{equation}
\end{widetext}If we now take an incoherent superposition over $m$ as in eq. \eqref{eq:spherically_symmetric_stress_tensor_def}, assuming that the spinor components $({\psi_0}^{(\ell,m)},{\psi_1}^{(\ell,m)},{\psi_2}^{(\ell,m)},{\psi_3}^{(\ell,m)})$ are the same for all $m$, we find that ${T_{\hat t\hat t}}^{(\ell)}$ no longer has any angular dependence \begin{align}
    {T_{\hat t\hat  t}}^{(\ell)} = \frac{\omega}{f^2r^2}\left(|{\psi_0}^{(\ell)}|^2+|{\psi_1}^{(\ell)}|^2+|{\psi_2}^{(\ell)}|^2+|{\psi_3}^{(\ell)}|^2\right),
\end{align} using the identity \begin{equation}
    \sum_{m=-\ell}^\ell \ _sY_{\ell,m}\ \overline{_{s^\prime}Y_{\ell,m}}= \frac{2\ell+1}{4\pi}\delta_{ss^\prime}\label{eq:spin_weighted_harmonics_orthogonality}
\end{equation} proven in appendix \ref{app:spin_weighted_Spherical_Harmonics}.

Henceforth, we drop the index $(\ell,m)$ on all quantities since we treat $\omega$ and $\ell$ as fixed, and assume that the functions $({\psi_0}^{(\ell,m)},{\psi_1}^{(\ell,m)},{\psi_2}^{(\ell,m)},{\psi_3}^{(\ell,m)})$ are the same for all $m$. Other stress-tensor components that need to vanish in a spherically symmetric stress-energy tensor either vanish identically, or vanish because the spin-indices $s,s^\prime$ on the spin-weighted spherical harmonics do not match, and are annihilated in the sum over $m$ via eq. (\ref{eq:spin_weighted_harmonics_orthogonality}). 

For the stress-energy tensor to be considered static, the mixed time-space components of the stress-energy must vanish. The time-angle components vanish under the incoherent superposition taken in eq. (\ref{eq:spherically_symmetric_stress_tensor_def}), but we also need $T_{\hat t\hat x}=T_{ab}{e_{0}}^a{e_{3}}^b$ to vanish. Indeed, $T_{\hat t\hat x}=0$ on Dirac solutions so long as \begin{align}
    \Im[\psi_0(x)\bar\psi_1(x) - \bar\psi_2(x)\psi_3(x)]=0,
\end{align} but this holds for stationary bound-state solutions (see eq. (\ref{eq:component_constraint})).

After taking the incoherent superposition necessary to produce a spherically symmetric stress-energy tensor, and restricting ourselves to stationary bound-state solutions, the non-vanishing components of the stress-energy tensor are \begin{equation}
\begin{aligned}
&T_{\hat t\hat t} := T_{ab}{e_0}^a{e_0}^b= \frac{\omega}{f^2 r^2} S,
\\
&T_{\hat x\hat x} := T_{ab}{e_3}^a{e_3}^b= \frac{\omega}{f^2 r^2} S
+ \frac{(\ell + 1/2)}{f r^3} D
+ \frac{\mu}{f r^2} P, \\
&T_{\hat\theta\hat\theta}:= T_{ab}{e_1}^a{e_1}^b 
= -\frac{(\ell + 1/2)}{2 f r^3} D,\\
&T_{\hat\varphi\hat\varphi}:=T_{ab}{e_2}^a{e_2}^b=-\frac{(\ell + 1/2)}{2 f r^3} D ,
\end{aligned} \label{eq:Dirac_stress_energy_tensor_components}
\end{equation} on Dirac solutions, expressed in the tetrad basis $\{{e_{0}}^a,{e_1}^a,{e_2}^a,{e_3}^a\}$ defined in eq. (\ref{eq:non_coordinate_basis_def}), and where we used eq. (\ref{eq:real_dirac}) to eliminate all derivatives of $(\psi_0,\psi_1,\psi_2,\psi_3)$. The quantities $S,D,P$ are Dirac bilinears, introduced for convenience, which are given by \begin{equation}
\begin{gathered}
S :=\psi_0^2+\psi_1^2+\psi_2^2+\psi_3^2,\qquad   D:= \psi_0^2-\psi_1^2-\psi_2^2+\psi_3^2, \\[1.2ex]
P:= 2(-1)^{\ell-1/2}(\psi_0\psi_1 -\psi_2\psi_3).
\end{gathered} \label{eq:Dirac_bilinear_def}
\end{equation}   We note that, for positive-frequency solutions, the stress-energy tensor component $T_{\hat t\hat t} > 0$; we will use this in \S \ref{sec:ED_equations} to prove a positive mass theorem for Einstein-Dirac solutions describing static, asymptotically flat, traversable wormhole spacetimes. The stress-energy tensor components that appear in the pointwise or averaged null energy condition are \begin{equation}
    \begin{aligned}
    &T_{ab}\left(\frac{\partial}{\partial\lambda}\right)^a\left(\frac{\partial}{\partial\lambda}\right)^b=T_{ab}\left(\frac{\partial}{\partial\tilde\lambda}\right)^a\left(\frac{\partial}{\partial\tilde\lambda}\right)^b = \frac{2\omega}{f^4 r^2} S
+ \frac{(\ell + 1/2)}{f^3 r^3} D
+ \frac{\mu}{f^3 r^2} P, \label{eq:bilinear_NEC}
\end{aligned}
\end{equation} where $(\partial/\partial \lambda)^a$ and $(\partial/\partial\tilde \lambda)^a$ are tangent to affinely parameterized radial null geodesics traversing the wormhole in opposite directions (see eq. \eqref{eq:affine_param_null_geos}).

The same procedure can be used to obtain a static and spherically symmetric Dirac current. Explicitly, if $j^a_{(\ell,m)}$ is the Dirac current (given by eq. (\ref{eq:Dirac_Charge_Current})) for a stationary bound-state solution with a particular spherical harmonic dependence $(\ell,m)$, the incoherent superposition \begin{equation}
    j^a_{(\ell)} = \frac{4\pi}{2\ell+1}\sum_{m=-\ell}^\ell j^a_{(\ell,m)} 
\end{equation} results in a static and spherically symmetric Dirac current, with the only nonvanishing component\footnote{The component $j_{(\ell)}^{\hat x}$ vanishes by eq. \eqref{eq:component_constraint}. Moreover, eq. \eqref{eq:conserved_Dirac_current}, which was used to derive eq. \eqref{eq:component_constraint}, implies that the Dirac current is conserved.} being \begin{align}
    j^{\hat t}_{(\ell)} := j_{(\ell)}^a{e^0}_a= \frac{|{\psi_0}^{(\ell)}|^2+|{\psi_1}^{(\ell)}|^2+|{\psi_2}^{(\ell)}|^2+|{\psi_3}^{(\ell)}|^2}{fr^2}.
\end{align}

We have now fully developed the machinery necessary to examine normalizable, static, and spherically symmetric Dirac solutions. The conclusion is that we must consider incoherent superpositions (in the sense of eq. (\ref{eq:spherically_symmetric_stress_tensor_def})) of stationary bound-state solutions. The resulting solutions are described by a set of real components $(\psi_0,\psi_1,\psi_2,\psi_3)$ for each fixed $\omega$ and $\ell$, satisfying eq. (\ref{eq:real_dirac}).

\subsection{Existence and Properties of Stationary Bound-State Solutions of the Dirac Equation} \label{sec:stationary_bound_states_exist}

We have established that stationary bound-state solutions are essential for normalizable, static, and spherically symmetric Dirac solutions. We will show here that traversable wormhole geometries can support stationary bound-state Dirac solutions by presenting a numerical example. We elucidate necessary conditions for bound state solutions to exist, and discuss when it is possible for bound state solutions of indefinite sphere-parity to exist at a single frequency. Finally, we present an example of a bound-state solution that violates the averaged null energy condition.

For this analysis, it is convenient to recast eq. (\ref{eq:real_dirac}) into the form of an eigenvalue equation \begin{equation}
        \begin{aligned}
    \omega \begin{pmatrix}
        {\psi_0}\\
        {\psi_1}
    \end{pmatrix}&=
         \left(i\sigma^2 f^2\frac{d}{dx} - (\ell+1/2)\frac{f}{r} \sigma^3 - \kappa_\ell\mu f\sigma^1\right)\begin{pmatrix}
        {\psi_0}\\
        {\psi_1}
    \end{pmatrix}, \\
    -\omega \begin{pmatrix}
        {\psi_2}\\
        {\psi_3}
    \end{pmatrix} &= \left(i\sigma^2 f^2\frac{d}{dx} - (\ell+1/2)\frac{f}{r} \sigma^3 - \kappa_\ell\mu f\sigma^1\right)\begin{pmatrix}
        {\psi_2}\\
        {\psi_3}
    \end{pmatrix},
    \end{aligned}\label{eq:discrete_dirac}
    \end{equation} where $\kappa_\ell:=(-1)^{\ell-1/2}$. Finding stationary bound-state solutions now corresponds to solving for eigenvalues and eigenfunctions of the operator \begin{equation}
        M = \left(i\sigma^2 f(x)^2\frac{d}{dx} - (\ell+1/2)\frac{f(x)}{r(x)} \sigma^3 - \kappa_\ell\mu f(x)\sigma^1\right).\label{eq:Dirac_operator}
    \end{equation} Sphere-parity even solutions will correspond to solutions with positive eigenvalue and sphere-parity odd solutions to those with negative eigenvalue.  If we consider stationary bound-state solutions at a fixed frequency $\omega$ in a generic wormhole spacetime, the solution will have definite sphere-parity. This is because a generic spacetime has no symmetry which enforces that both $\omega$ and $-\omega$ will be in the discrete spectrum of $M$, so the solution at that frequency will be either sphere-parity even or sphere-parity odd, but not a linear combination of the two.

    There do exist non-generic exceptions for which $\omega$ and $-\omega$ are simultaneously in the discrete spectrum of the Dirac operator. For one, this occurs when the wormhole geometry is reflection symmetric about the throat. If the metric functions are even functions about the throat, reflection symmetry forces the Dirac operator to have the property that, for all $\omega\in \R$, $\omega \in\text{Spec}(M)\Rightarrow -\omega\in\text{Spec}(M)$, where $\text{Spec}(M)$ denotes the discrete spectrum of $M$. Suppose that the wormhole throat is at $x=0$ (by translating the coordinate $x$ if necessary), so that reflection is implemented by the map $x\mapsto -x$. This map has an associated action on the radial spinor components \begin{equation}
        \begin{aligned}
        &\psi_0(x)\mapsto (-1)^{\ell-1/2}\psi_3(-x),
        &&\psi_1(x)\mapsto -(-1)^{\ell-1/2}\psi_2(-x), \\
        &\psi_2(x) \mapsto -(-1)^{\ell-1/2}\psi_1(-x),
        &&\psi_3(x)\mapsto (-1)^{\ell-1/2}\psi_0(-x),
    \end{aligned}
    \end{equation} which follows from the discussion of reflection symmetry in appendix \ref{app:reflection}. Thus, for every sphere-parity even solution on a reflection-symmetric spacetime, we have a sphere-parity odd solution at the same frequency (or vice versa).\footnote{Any static and spherically symmetric spacetime also admits the discrete symmetries given by sphere-parity and time-reversal, both introduced in \S \ref{sec:dirac_equation} and discussed in detail in appendix \ref{app:discrete_symmetries}. However, these solutions take bound-state solutions of a particular sphere-parity into themselves, and do not give rise to additional degeneracies.} This can also occur when two unrelated discrete eigenvalues accidentally satisfy $\omega_i=-\omega_j$. Such a pairing is not forbidden, but should be regarded as exceptional in the space of wormhole backgrounds.  In either case, a bound-state solution will be some linear combination of a sphere-parity even and sphere-parity odd solution.

   We will now show that stationary bound-state Dirac solutions exist by presenting a numerical example. Since $M$ is a linear ordinary differential operator on $-\infty<x<\infty$, we first pass $M$ to a differential operator on the finite interval $(-1,1)$ by introducing a compactified coordinate $X$ via $x= X/(1-X^2)$, and we will treat $f,r,$ and the spinor components as functions of $X$. Then, by replacing the continuous variable $X$ with a discrete grid of $N$ points $(X_0=-1,X_1,\dots,X_{N-2},X_{N-1}=1)$, the operator $M$ becomes a $2N\times 2N$-dimensional matrix 
   \begin{equation}
       \begin{aligned}
        M_D &= \text{Diag}\left(f^2\frac{dX}{dx}\right)D_1\otimes (i\sigma^2) - (\ell+1/2)\text{Diag}\left(\frac{f}{r}\right)\otimes \sigma^3 - (-1)^{\ell-1/2} \mu\  \text{Diag}(f)\otimes \sigma^1 \label{eq:discrete_matrix}
    \end{aligned}
   \end{equation} where `$\text{Diag}(F)$' is the diagonal matrix whose entries are $(F(X_0),F(X_1),\dots,F(X_{N-1}))$, $D_1$ is a differentiation matrix representing the linear operator $d/dX$ on $(X_0,\dots,X_{N-1})$, and $\otimes$ denotes the Kronecker product of two matrices. The eigenvalues of $M_D$ will give the frequency, and its eigenvectors are the corresponding values of the solution on the grid points, e.g., \begin{equation}(\psi_0(X_0),\psi_1(X_0),\dots,\psi_0(X_{N-1}),\psi_1(X_{N-1})).
    \end{equation} We choose our grid $\{X_k\}$ to be the Chebyshev-Gauss-Lobatto points $X_k:= \cos\left(\frac{\pi k}{N-1}\right)$, and the derivative matrix $D_1$ is the pseudospectral derivative matrix \cite{Trefethen2000SpectralMethodsMATLAB,Boyd2001ChebyshevFourierSpectralMethods}. Solutions were validated using convergence tests on the obtained eigenvalues and eigenvectors, and by reproducing the solutions using a finite-difference scheme. For $\mu = 1$ and $\ell = 1/2$, we explored the family of traversable wormhole metrics \begin{align}
        &r(x) = \sqrt{x^2+1}, && f(x) = 1- \frac{a}{r(x)}
    \end{align} for $10^{-3} < a < 1-10^{-3}$ (in 100 linearly spaced increments) with $N=500$ points and found that bound-state solutions exist in each case.

     When stationary bound-state solutions do exist, there appear to be infinitely many of them in the examples we have obtained. The upper frequency bound $\mu \min(f_0^\pm)$ acts as an accumulation point for the spectrum. This shows that the gravitational potential generated by an asymptotically flat wormhole falls off slowly enough to support an infinite number of bound states, as one would expect from standard theorems about potentials with slow falloff  \cite{ReedSimon1978MMMP4}.

    However, not all wormhole geometries admit stationary bound-state solutions; whether or not such solutions are permitted is largely controlled by the behavior of the metric function $f$. To see why, define the sphere-parity even Dirac bilinears by taking the sphere-parity even part of eq. (\ref{eq:Dirac_bilinear_def}): \begin{align}
         S_\text{even} = \psi_0^2+\psi_1^2, \  D_\text{even} = \psi_0^2-\psi_1^2, \ P_\text{even} = 2\kappa_\ell\psi_0\psi_1,\label{eq:Dirac_parity_even_bilinears}
     \end{align} where $\kappa_\ell:=(-1)^{\ell-1/2}$. We will proceed with an analysis of sphere-parity even Dirac solutions for definiteness, though identical results would hold if we analyzed the sphere-parity odd case. Using the Dirac equation (eq. (\ref{eq:real_dirac})), $D_\text{even}$ and $P_\text{even}$ satisfy \begin{equation}
         \begin{aligned}
         -\frac{\kappa_\ell}{2}\frac{dD_\text{even}}{dx} &=  \frac{\mu}{f}S_\text{even} + \frac{\omega}{f^2}P_\text{even} ,
         \\
         \frac{\kappa_\ell}{2}\frac{dP_\text{even}}{dx} &= \frac{\omega}{f^2}D_\text{even} + \frac{(\ell+1/2)}{fr} S_\text{even}.
     \end{aligned}
     \end{equation} For bound-state solutions, $D_\text{even},P_\text{even}\to 0$ as $|x|\to\infty$, so by integrating over all $x\in\R$, we can replace the left-hand side by its (vanishing) boundary contribution, and obtain \begin{equation}
         \begin{aligned}
         0 & = \int_{-\infty}^\infty\left[ \frac{\mu}{f}S_\text{even} + \frac{\omega}{f^2}P_\text{even}\right]dx, \\
         0 &= \int_{-\infty}^\infty \left[\frac{\omega}{f^2}D_\text{even} + \frac{(\ell+1/2)}{fr}S_\text{even}\right]dx.\label{eq:dirac_bilinear_equalities}
     \end{aligned}
     \end{equation} $D_\text{even}$ and $P_\text{even}$ satisfy the algebraic inequalities $|D_\text{even}|\leq S_\text{even}$ and $|P_\text{even}|\leq S_\text{even}$.  Assuming the solution is normalized in the Dirac norm of eq. \eqref{eq:dirac_norm}, these expressions imply \begin{equation}
         \begin{aligned}
         &\omega \geq \max\left\{\int_{-\infty}^\infty \frac{\mu }{f}S_\text{even},\int_{-\infty}^\infty \frac{\ell+1/2}{fr}S_\text{even}\right\},\label{eq:lower_bound_omega}
     \end{aligned}
     \end{equation}which rules out zero-frequency stationary bound-state solutions, since the right-hand side is positive for nontrivial solutions. By manipulating the first of the two equations in eq. \eqref{eq:dirac_bilinear_equalities}, we find \begin{equation}
         \begin{aligned}
         0 &\geq \int_{-\infty}^{\infty} \ \frac{S_\text{even}}{f}\left[\mu-\frac{\omega}{f}\right]dx\\
         &=\int_{-\infty}^\infty \frac{S_\text{even}}{f}\left[\mu - \frac{\omega}{f_0^\pm} +\frac{\omega}{f_0^\pm} - \frac{\omega}{f}\right]dx\\
         &\geq \int_{-\infty}^\infty\frac{\omega S_\text{even}}{f^2}\left[\frac{f}{f_0^\pm} - 1\right]\ dx,
     \end{aligned}
     \end{equation} where $f_0^\pm:=\lim_{x\to\pm\infty}f(x),$ and we used that bound-state solutions must satisfy $\mu > \omega/f_0^\pm$, as discussed in \S \ref{sec:stationary_bound_states}. The result is that we cannot have bound-state solutions if $f \geq f_0^\pm$ everywhere; in other words, $f$ must dip below its asymptotic values, supporting the usual interpretation of $f$ as a gravitational potential. In other words, the wormhole must generate an attractive gravitational potential to support any bound states.

     We searched for examples of stationary bound states that violate the null energy condition and the averaged null energy condition. For the explicit metric profile \begin{equation}
         \begin{aligned}
         &f(x) = 300\left[1-0.8 e^{-(x/300)^2} \right]\\
         & r(x) = \left[\frac{1}{\sqrt{(x/300)^2+1}}+f(x)e^{-(3x)^2}\right]^{-1},\label{eq:anec_metric_profile}
     \end{aligned}
     \end{equation} plotted in Fig. \ref{fig:anec_violating_profile}, we obtained a sphere-parity even bound-state Dirac solution with $\mu = 1$, $\ell=1/2$, and frequency $\omega \approx 111.88$ that produced a robust violation of the averaged null energy condition \begin{equation}
         \int \frac{2\omega}{f^4r^2}S + \frac{(\ell+1/2)}{f^3r^3}D + \frac{\mu}{f^3r^2}P \ dx \approx -3.59<0. \label{eq:ANEC_integral}
     \end{equation}  The factor of 300 in the metric profile was used to stretch the profile along the $x$-axis to make the sharp peak in the $r$ profile easier to resolve numerically. This solution was validated by convergence testing to ensure it was not a numerical artifact.
     
     To obtain this counterexample, we made use of the fact that we found bound-state solutions with the second term in eq. \eqref{eq:ANEC_integral} negative at the throat, and designed a localized perturbation to make $r$ smaller while keeping $D<0$, until this term was able to dominate the other two terms in the averaged null energy integral. 

     \begin{figure}[htbp!]
    \centering
    
    \hfill
        \centering
        \begin{subfigure}[t]{0.48\textwidth}
        \centering
            \centering
            \includegraphics[width=\linewidth]{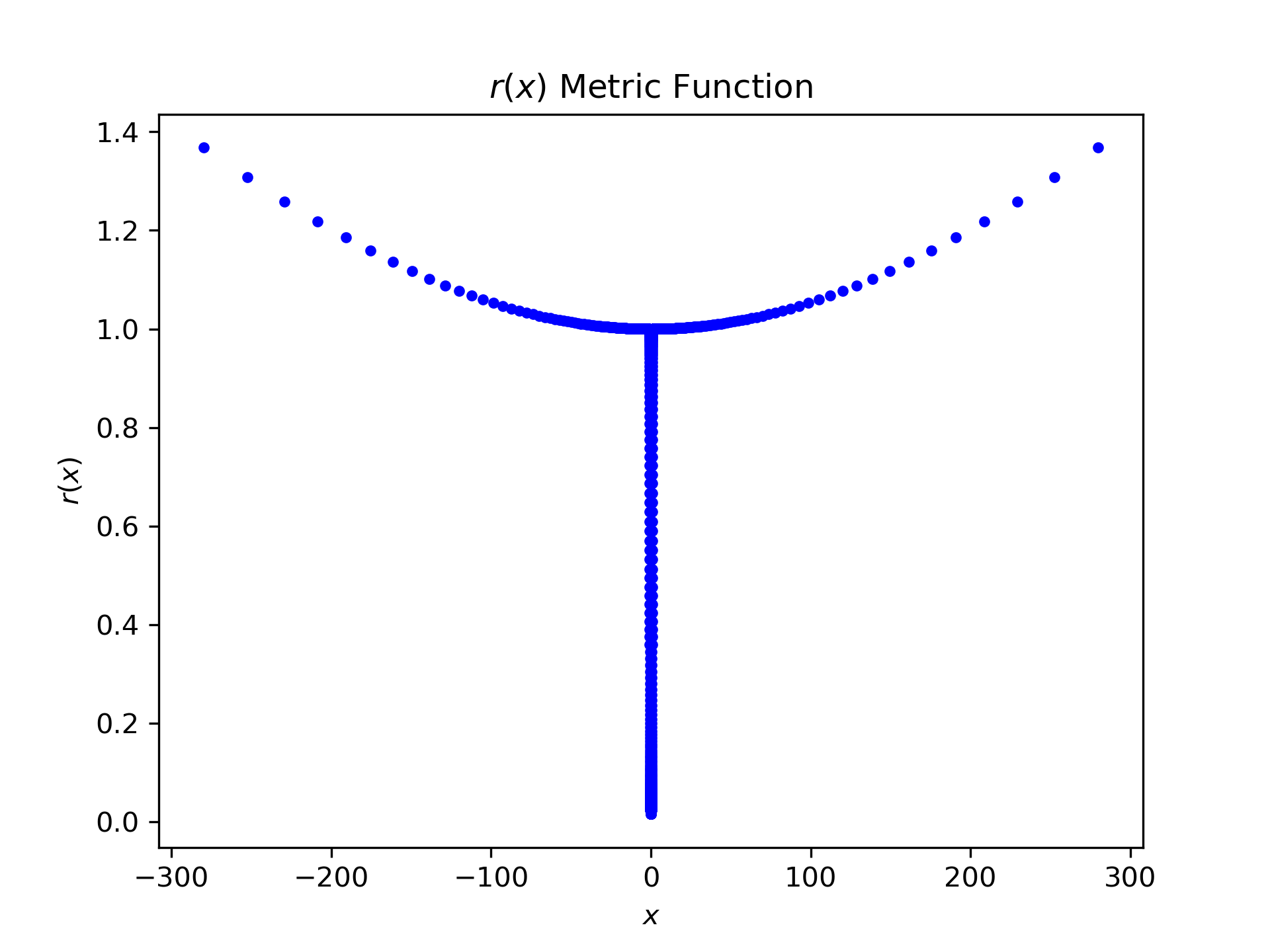}
        \caption{Plot of $r$.}
    \end{subfigure}
    \begin{subfigure}[t]{0.48\textwidth}
            \centering
            \includegraphics[width=\linewidth]{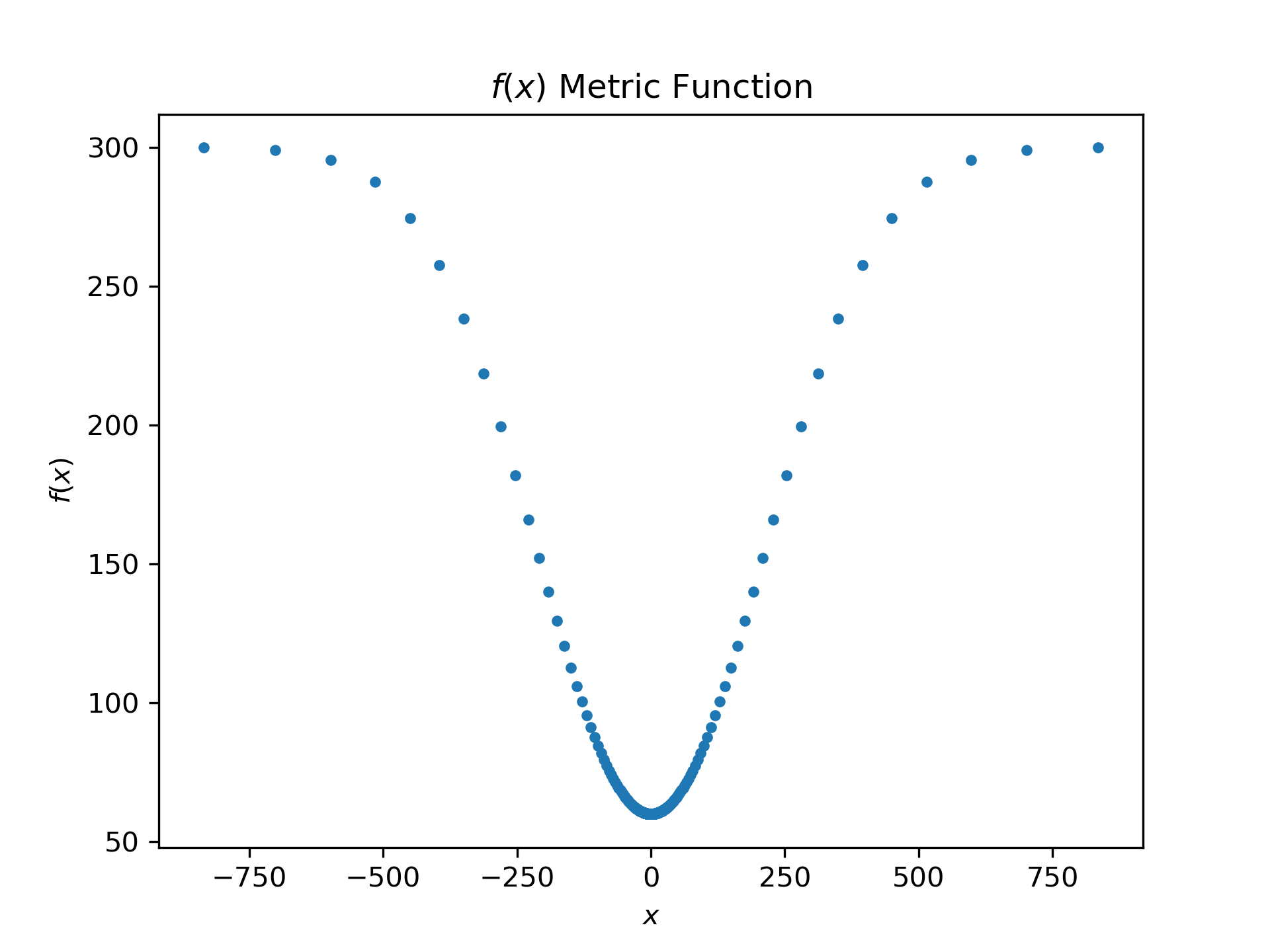}
        \caption{Plot of $f$.}
    \end{subfigure}
    \caption{Metric profile for eq. \eqref{eq:anec_metric_profile}, for which violations of the averaged null energy condition were explicitly obtained.}
    \label{fig:anec_violating_profile}
    \end{figure}

     This counterexample is important because it shows that the restrictions imposed by staticity and spherical symmetry do not themselves restore the averaged null energy condition. In particular, a normalizable, positive-frequency stationary bound state can provide the negative averaged null energy required of matter sourcing a traversable wormhole. Thus, any obstruction to static, spherically symmetric Einstein-Dirac wormholes cannot arise simply because the physically admissible stationary Dirac states considered here necessarily satisfy the averaged null energy condition; any obstruction must arise from the requirement of self-consistency between the Dirac field and the geometry it sources. We therefore turn in the next section to the fully backreacting Einstein-Dirac system.

\section{Obstructions to Traversable Wormholes in Einstein-Dirac Theory}\label{ch:EDM_system}

We now turn to the treatment of the full Einstein-Dirac system. In static, spherically symmetric spacetimes, the equations take the form of a system of ordinary differential equations, which we write explicitly in \S \ref{sec:ED_equations}, and prove a positive mass theorem for wormhole spacetimes. The remainder of this section is dedicated to numerical investigations of the Einstein-Dirac system. It will prove most convenient to provide initial data for the Einstein-Dirac system asymptotically, and we will formulate the parameter space of asymptotic initial data in \S \ref{sec:asym_initial_data}. We will then present our numerical strategy for evolving this asymptotic initial data to produce an Einstein-Dirac solution in \S \ref{sec:numerical_strategy}.

We will then investigate whether any full traversable wormhole solutions to the Einstein-Dirac equations exist. In \S \ref{sec:non_existence_generic}, we first demonstrate that there is no issue with obtaining a ``partial-wormhole solution'' of the Einstein-Dirac system, which shows that there are no inconsistencies with the Einstein-Dirac equations admitting a local solution in the neighborhood of a wormhole throat, or with extending these solutions to one of the asymptotic infinities. However, we will argue that solutions with definite $\omega$, $\ell$, and sphere-parity cannot be extended to a second asymptotically flat end. We then drop the definite sphere-parity assumption in \S \ref{sec:non_existence_symmetric} in the case of reflection-symmetric wormholes, where we use an extensive numerical search to argue that the conditions required for a reflection-symmetric wormhole throat cannot be satisfied. These facts together support the conclusion that traversable wormhole solutions do not exist as valid solutions to the Einstein-Dirac equations.

\subsection{The Einstein-Dirac System}
\label{sec:ED_equations}

We now incorporate backreaction, coupling the Dirac field to the spacetime geometry via Einstein's equations. Having already presented the Dirac equation in \S \ref{sec:dirac_equation} and the Dirac stress-energy tensor components in \S \ref{sec:spherical_sym_stress_energy}, we need only compute the Einstein tensor for the spacetime metric of eq. (\ref{eq:metric}). Its components, in terms of the tetrad of eq. \eqref{eq:non_coordinate_basis_def}, are given by 
\begin{equation}
    \begin{aligned}
        G_{\hat t\hat t} &:=G_{ab}{e_0}^a{e_0}^b=   -\frac{2 f^2 {r^{\prime\prime}}}{r}-\frac{f^2 {r^\prime}^2}{r^2}-\frac{2 ff^\prime r^\prime}{r}+\frac{1}{r^2} ,\\
        G_{\hat x\hat x} &:= G_{ab}{e_3}^a{e_3}^b=  \frac{f^2 {r^\prime}^2}{r^2}+\frac{2 f f^\prime r^\prime}{r}-\frac{1}{r^2}, \\
        G_{\hat\theta\hat\theta} &:=G_{ab}{e_1}^a{e_1}^b= \frac{f^2 {r^{\prime\prime}}}{r}+f f^{\prime\prime}+\frac{2 f f^\prime r^\prime}{r}+{f^\prime}^2,\\
        G_{\hat\varphi\hat\varphi} &:=G_{ab}{e_2}^a{e_2}^b= \frac{f^2 {r^{\prime\prime}}}{r}+f f^{\prime\prime}+\frac{2 f f^\prime r^\prime}{r}+{f^\prime}^2,    \end{aligned}\label{eq:Einstein_tensor_components}
\end{equation} where primes denote derivatives with respect to the $x$-coordinate. Equating the Einstein tensor to the Dirac stress-energy tensor gives us the Einstein field equations \begin{align}
    &-\frac{2 f^2 {r^{\prime\prime}}}{r}-\frac{f^2 {r^\prime}^2}{r^2}-\frac{2 ff^\prime r^\prime}{r}+\frac{1}{r^2} = \frac{\omega}{f^2r^2}S ,\label{eq:tt_eqn}\\
    &\frac{f^2 {r^\prime}^2}{r^2}+\frac{2 f f^\prime r^\prime}{r}-\frac{1}{r^2} = \frac{\omega}{f^2 r^2} S
+ \frac{(\ell + 1/2)}{f r^3} D
+ \frac{\mu}{f r^2} P,\label{eq:xx_eqn}\\
    &\frac{f^2 {r^{\prime\prime}}}{r}+f f^{\prime\prime}+\frac{2 f f^\prime r^\prime}{r}+{f^\prime}^2=-\frac{(\ell + 1/2)}{2 f r^3} D ,\label{eq:angleangle_eqn}
\end{align} where $S,D,$ and $P$ are the Dirac bilinears introduced in eq. \eqref{eq:Dirac_bilinear_def}.

Using these equations, we will now prove that positivity of $T_{\hat t\hat t}$ (which we pointed out in \S \ref{sec:spherical_sym_stress_energy}) forces the spacetime to have positive ADM mass. The usual positive mass theorem does not apply here, as it relies on the dominant energy condition \cite{SchoenYau1979PositiveMass,SchoenYau1981PositiveMassII,Witten1981PositiveEnergy} which is not satisfied by Dirac fields. Beyond being an interesting result in its own right, this will be a crucial ingredient in our numerical investigations, as one of the free parameters will be the spacetime mass.

\begin{theorem}
    All static, spherically symmetric, asymptotically flat traversable wormhole solutions to the Einstein-Dirac equations have ADM mass bounded below by half the (nearest) throat radius; in particular, the ADM mass is positive at each asymptotic end.\label{thm:spacetime_mass}
\end{theorem} 

\begin{proof}
    In a spherically symmetric spacetime, it is convenient to introduce the Misner-Sharp mass \cite{Hayward1996GravitationalEnergySphericalSymmetry} \begin{equation}
        M(x):= \frac{r(x)}{2}(1-r^\prime(x)^2f(x)^2).
    \end{equation}
 A direct application of the time-time component of the Einstein equations (eq. (\ref{eq:tt_eqn})) gives \begin{equation}
    \frac{dM}{dx} = \frac{r^\prime}{2}(r^2T_{\hat t\hat t} )= \frac{r^\prime}{2}\left(\frac{\omega}{f^2}S\right).
\end{equation}  If $x_t^+$ is the $x$ coordinate of the rightmost (i.e., largest $x$ coordinate) wormhole throat, then $r^\prime(x)\geq 0$ for $x>x_t^+$. If not, there would exist another minimum of $r$ at some larger $x$ for $r$ to have the correct asymptotic behavior, which contradicts our assumption that $x_t^+$ was the position of the rightmost wormhole throat. Similarly, if the leftmost wormhole throat (i.e., the smallest $x$-coordinate) is at $x_t^-$, we have $r^\prime(x) \leq 0$ for $x<x_t^-$. Since $\omega S/f^2>0$, this implies that $M$ is non-decreasing as we move to the right of the rightmost throat or to the left of the leftmost throat. At either throat (which is by definition a minimum of $r$), we have $r^\prime(x_t^\pm) = 0$, so $M(x_t^\pm) = r(x_t^\pm)/2>0$. Combining these observations, we learn that \begin{equation}
    \begin{aligned}
    &2M_+ = r(x_t^+)  +\omega  \int_{x_t^+}^\infty \frac{ r^\prime(x^\prime)}{f(x^\prime)^2}S(x^\prime) \ dx^\prime > 0, \\
    &2M_- = r(x_t^-)  - \omega\int_{-\infty}^{x_t^-} \frac{ r^\prime(x^\prime)}{f(x^\prime)^2}S(x^\prime) \ dx^\prime > 0 ,
\end{aligned}\label{eq:asym_spacetime_mass_equation}
\end{equation} where $M_\pm := \lim_{x\to\pm \infty}M(x)$ is equal to the ADM mass at each asymptotic end \cite{Hayward1996GravitationalEnergySphericalSymmetry}. In particular, eq. \eqref{eq:asym_spacetime_mass_equation} gives that \begin{equation}
    M_\pm > \frac{r(x_t^\pm)}{2} > 0.
\end{equation} In other words, the ADM mass at each asymptotic end is bounded below by half the throat radius of the wormhole throat nearest to that asymptotic end, and is therefore positive.\end{proof}

We also point out two additional consequences of the positivity of $T_{\hat t\hat t}$. As was shown in the proof, we have positivity of the Misner-Sharp mass on $(x_t^+,\infty)$ and $(-\infty,x_t^-)$. This, in turn, implies that $|fr^\prime|\leq 1$ on both intervals.\footnote{When this condition is met, we can take a static slice of the wormhole and subsequently take a cross-section at a particular angle $\theta$, and the resulting 2-dimensional surface can be embedded into $\R^3$ as we did in \S \ref{sec:traversable_wormhole_spacetime_geometry}. This means that the spacetime from the rightmost throat to $x\to \infty$ is embeddable in this way, and from the leftmost throat to $x\to -\infty$; this also implies that, for an Einstein-Dirac wormhole with only one throat (where $x_t^+=x_t^-$), the resulting 2-dimensional surface is always embeddable in $\R^3$.} The second consequence is that the region where the null energy condition must be violated cannot be made too small (relative to the throat radius). Some solutions have been proposed in which the authors minimize the size of the region over which the null energy condition is violated, e.g., \cite{Visser1989TraversableWormholesSimpleExamples}. Morris and Thorne pointed out that when $T_{\hat t\hat t}>0$, the null energy condition must be violated over a macroscopically large region of space \cite{Morris:1988tu}.

\subsection{Asymptotic Initial Data for the Einstein-Dirac System}\label{sec:asym_initial_data}
As the Einstein-Dirac system forms a set of ordinary differential equations, we need to characterize the parameters and initial data that determine a particular solution, in order to set up a well-posed numerical problem. In general, there is no clear restriction on the initial data ensuring that we have $f$ and $r$ everywhere nonvanishing, that the Dirac field has the correct asymptotic falloff necessary to describe a stationary bound state, and that the metric functions describe an asymptotically flat geometry. To account for the latter two issues, we will prescribe initial data asymptotically for the Dirac field and spacetime metric that forces the metric and Dirac components to have the correct asymptotic behavior at one of the asymptotic ends of the spacetime.

Let us now describe the precise asymptotic behavior that is required of a solution to the Einstein-Dirac system. Taking linear combinations of the Einstein equations [eqs. (\ref{eq:tt_eqn})-(\ref{eq:angleangle_eqn})], we obtain second-order evolution equations for $r$ and $f$
    \begin{align}
        -\frac{2f^2r^{\prime\prime}}{r} &= \frac{2\omega}{f^2r^2}S + \frac{(\ell+1/2)}{fr^3}D + \frac{\mu}{fr^2}P,\label{eq:EFE_r_ev}\\
        ff^{\prime\prime} &= \frac{\omega}{f^2r^2}S +\frac{\mu}{2fr^2}P - {f^\prime}^2 - \frac{2ff^\prime r^\prime}{r} \label{eq:EFE_f_ev},
    \end{align}
 as well as a constraint equation relating their first derivatives \begin{equation}
     \frac{f^2 {r^\prime}^2}{r^2}+\frac{2 f f^\prime r^\prime}{r}-\frac{1}{r^2} = \frac{\omega}{f^2 r^2} S
+ \frac{(\ell + 1/2)}{f r^3} D
+ \frac{\mu}{f r^2} P. \label{eq:EFE_constraint}
\end{equation}  It is straightforward to check that if eq. (\ref{eq:EFE_constraint}) is satisfied at a point, it is preserved under evolution according to eqs. (\ref{eq:EFE_r_ev}), (\ref{eq:EFE_f_ev}), and the Dirac equation (eq. (\ref{eq:real_dirac})).

As pointed out in \S \ref{ch:static_spherically_symmetric_dirac_solutions}, the spinor field will need to decay exponentially, meaning that the metric functions $r$ and $f$ for very large $|x|$ will satisfy the vacuum Einstein equations to very good approximation; thus, the solution will be well-described by the Schwarzschild metric asymptotically. In our coordinates, the solution is (up to exponentially small terms) given by\begin{equation}
    \begin{aligned}
    &r(x) \sim \pm \frac{x}{f_0^\pm} + d^\pm, 
    && f(x) \sim f_0^\pm \left(1 - \frac{2 M_\pm}{(\pm x/f_0^\pm + d^\pm)}\right)^{1/2}, \label{eq:metric_asymptotic_solution}
\end{aligned}
\end{equation} as $x\to\pm\infty$, where $M_\pm$ is the ADM mass, which we showed in theorem \ref{thm:spacetime_mass} is positive. We will use the $\pm$ superscripts and subscripts throughout to keep track of whether a particular quantity is defined at $x\to +\infty$ or $x\to-\infty$. Thus, the metric initial data is encoded in only three free parameters at either asymptotic end ($f_0^\pm$, $d^\pm$, $M_\pm$), with the restrictions $f_0^\pm >0$ and $M_\pm > 0$.

Some of these asymptotic parameters are coordinate artifacts which can be normalized or set to zero by appropriately rescaling the coordinates, metric functions, and Dirac solution, or by translating the coordinates. Fix a spherical harmonic index $\ell\in \N+\frac{1}{2}$, and suppose $(f(x),r(x),\psi_i(x),\omega,\mu)$ are the metric functions, spinor components, frequency, and Dirac mass for an Einstein-Dirac solution using the coordinates $t,x$ introduced in \S \ref{sec:traversable_wormhole_spacetime_geometry}. The coordinate $x$ was defined as an affine parameter along null geodesics, but this means any coordinate affinely related to $x$ (e.g. $\tilde x =  \alpha x+\beta $ with $\alpha,\beta\in\R$) would be just as good. By rescaling the coordinate $x\mapsto \tilde x := ax$ (with $a>0)$ and simultaneously rescaling the metric and spinor functions/parameters as follows, we can map one solution of the Einstein-Dirac system to another:\begin{equation}
    \begin{aligned}
   & (\tilde f(\tilde x),\tilde r(\tilde x),\tilde \psi_i(\tilde x),\tilde x,\tilde t,\tilde \omega,\tilde \mu) = (f(x),ar(x),a^{1/2}\psi_i(x),ax,t,a^{-1}\omega,a^{-1}\mu).\label{eq:rescaling_1}
\end{aligned} 
\end{equation}  This preserves asymptotic flatness and maps solutions of the Einstein-Dirac system into themselves. Similarly, if we rescale the $t$ coordinate\footnote{Because we are in a static spacetime, translating the $t$ coordinate is also a symmetry, but this is not useful for simplifying the analysis.} $t\mapsto\tilde t:=  bt$ (with $b>0$), then we can rescale everything else by \begin{equation}
    \begin{aligned}
    &(\tilde f(\tilde x),\tilde r(\tilde x), \tilde\psi_i(\tilde x),\tilde x , \tilde t, \tilde\omega,\tilde \mu) = (b^{-1}f(x),r(x),b^{-1/2}\psi_i(x),b^{-1}x,bt,b^{-1}\omega,\mu),\label{eq:rescaling_2}
\end{aligned}
\end{equation}and again find that asymptotically flat Einstein-Dirac solutions map into other asymptotically flat solutions.

These rescaling freedoms allow us to make three crucial simplifications. First, using the freedom to translate our coordinate $x$ allows us to set either $d^-$ or $d^+$ to zero.\footnote{Note that this prevents us from assuming that the wormhole throat is at $x=0$ when we impose such an assumption. Unless the wormhole is reflection symmetric, we cannot fix \textit{both} $d^\pm$ simultaneously.} Second, the rescaling freedom in eq. (\ref{eq:rescaling_2}) allows us to normalize either $f_0^+$ or $f_0^-$ to 1. Finally, the rescaling in eq. (\ref{eq:rescaling_1}) allows us to set $\mu =1$. Note that we could have used eq. \eqref{eq:rescaling_1} to rescale the norm of any normalizable Dirac solution to the Einstein-Dirac system: 
\begin{equation}
    \begin{aligned}
&\int_{-\infty}^\infty \frac{|\tilde \psi_0(\tilde x)|^2+|\tilde \psi_1(\tilde x)|^2+|\tilde \psi_2(\tilde x)|^2+|\tilde \psi_3(\tilde x)|^2}{\tilde f(\tilde x)^2}\ d\tilde x \\
& = a^2\int_{-\infty}^\infty \frac{|\psi_0(x)|^2+|\psi_1(x)|^2+|\psi_2(x)|^2+|\psi_3(x)|^2}{f(x)^2} \ dx.\end{aligned}
\end{equation}
A physical Dirac solution must have norm 1, but this freedom would allow us to take any Einstein-Dirac solution with finite Dirac norm and rescale its value to 1. Thus, finding a normalizable solution with $\mu = 1$ is equivalent to finding a solution with norm 1 and some other mass.

With the asymptotic behavior of the metric functions taken care of, we turn to the asymptotic behavior of the spinor fields themselves. To determine the appropriate asymptotic initial data for the spinor fields, we need to strip off the universal asymptotic behavior of a bound state as $x\to\pm\infty$ by defining asymptotic spinor components $\Xi_i^\pm$: \begin{align}
    &\Xi_i^\pm(x) := e^{\pm \mu x\sqrt{1-\Omega_\pm^2}}(\pm x)^{A_\pm} \psi_i(x) ,\label{eq:def_asymptotic_variables}\end{align}where \begin{align}
    A_\pm := \frac{1-2\Omega_\pm^2}{\sqrt{1-\Omega_\pm^2}}\tilde M_\pm. 
\end{align} We also introduce here the dimensionless parameters \begin{align}
    &\Omega_\pm := \frac{\omega}{\mu f_0^\pm}\in (0,1), && \tilde M_\pm :=  \mu M_\pm > 0.
\end{align} The Dirac equations imply that each of the functions $\Xi_i^\pm$ admits an asymptotic expansion of the form \begin{align}
    \Xi_i^\pm \sim \Xi_{i,0}^\pm + \Xi_{i,1}^\pm/x + \dots \qquad \text{ as }x\to\pm\infty, \label{eq:spinor_asym_expansion}
\end{align} where the coefficients $\Xi_{i,j}^\pm$ for $j \geq 1$ can be determined from the lowest-order coefficient $\Xi_{i,0}^\pm$ by solving the Dirac equation order-by-order in $1/x.$ Moreover, the lowest-order equation relates the leading coefficient for each pair $(\Xi_0^\pm,\Xi_1^\pm)$ and $(\Xi_2^\pm,\Xi_3^\pm)$: \begin{align}
    &\Xi_{1,0}^\pm = B_\pm \Xi_{0,0}^\pm, && \Xi_{3,0}^\pm = -B_\pm \Xi_{2,0}^\pm,\label{eq:spinor_asym_reln}\end{align} where \begin{align} B_\pm := \frac{(-1)^{\ell+1/2} \pm \sqrt{1-\Omega_\pm^2}}{\Omega_\pm}.
\end{align}  This means that the asymptotic spinor initial data (at fixed $\omega$, $\mu$, and $\ell$) is encoded in our choice of just two parameters by specifying one of $(\Xi_{0,0}^\pm ,\Xi_{1,0}^\pm)$ and one of $(\Xi_{2,0}^\pm ,\Xi_{3,0}^\pm)$ at either asymptotic end.

\begin{table*}
    \centering
    \hspace*{-1cm}
    \renewcommand{\arraystretch}{1.5} 
    \begin{tabular}{c|c|c}
    Parameter & Definition & Meaning \\
    \hline $\Omega_\pm$ & $\omega/(\mu f_0^\pm)$ & Dimensionless frequency (between 0 and 1)\\
    $\tilde M_\pm$ & $\mu M_\pm$ & Dimensionless spacetime mass (positive)\\
    $\Xi_{i,0}^\pm$ & eqs. (\ref{eq:def_asymptotic_variables}) and (\ref{eq:spinor_asym_expansion}) & Leading-order asymptotic spinor component for $\Xi_i^\pm$
\end{tabular}
\caption[Review of definitions of asymptotic parameters]{Review of definitions for asymptotic parameters.}
\end{table*}

Henceforth, we will restrict our attention to evolving initial data specified at the left asymptotic end ($x\to-\infty$) into the interior of the spacetime. Given a solution to the Einstein-Dirac system whose asymptotic parameters are given in terms of $(\Omega_-,\tilde M_-,\Xi_{0,0}^-,\Xi_{3,0}^-)$ at $x\to-\infty$, there is also the reflected solution, obtained by taking $x\mapsto -x$ and appropriately interchanging the spinor components (see appendix \ref{app:reflection}). The right asymptotic parameters for the reflected solution are related to the original left asymptotic parameters by \begin{equation}
    (\Omega_+ ,\tilde M_+,\Xi_{0,0}^+,\Xi_{3,0}^+) = (\Omega_-,\tilde M_-, \Xi_{3,0}^-, \Xi_{0,0}^-),\label{eq:reflection_asymptotic_data}
\end{equation} since the sphere-parity even and odd spinor components are interchanged. Thus, we will drop the $\pm$ subscripts and superscripts on asymptotic quantities moving forward; unless otherwise specified, our asymptotic parameters are taken to be those at $x\to-\infty$, but by applying eq. \eqref{eq:reflection_asymptotic_data} one can obtain the analogous result for asymptotic data at $x\to+\infty$.

Note that the Einstein-Dirac system is invariant under independently multiplying both sphere-parity even components $(\Xi_0,\Xi_1)$ or both sphere-parity odd components $(\Xi_2,\Xi_3)$ by $-1$. Indeed, the Dirac equation is linear, and the two parity sectors are decoupled, while the Dirac stress-energy tensor is quadratic within each sector. These sign changes, therefore, produce the same stress-energy tensor and hence the same geometry. This means that, whichever one of $(\Xi_{0,0} ,\Xi_{1,0})$ and of $(\Xi_{2,0} ,\Xi_{3,0})$ we choose to specify, we can take both to be non-negative without loss of generality. We choose to prescribe initial data in terms of $\Xi_{0,0}$ and $\Xi_{2,0}$, and take both to be non-negative.

In summary, a stationary bound-state solution to the Einstein-Dirac system (at fixed $\ell \in \N+\frac{1}{2}$) is picked out uniquely by specifying four parameters \begin{align}
    \Omega \in (0,1), \ \tilde M >0,\ \Xi_{0,0} \geq 0, \ \Xi_{2,0}  \geq 0.\label{eq:asym_initial_data}
\end{align} We normalize $\mu=1$, $d = 0,$ and $f_0 = 1$ using eqs. \eqref{eq:rescaling_1} and \eqref{eq:rescaling_2}. This data determines a local asymptotic solution near $x\to-\infty$, and a global traversable wormhole solution is one for which the corresponding evolution remains regular, develops a throat, and decays appropriately at the opposite asymptotic end.

\subsection{Strategy for Evolving Asymptotic Initial Data}
\label{sec:numerical_strategy}

Recall that, per the discussion of \S \ref{sec:asym_initial_data}, the initial data for the Einstein-Dirac system can be given in terms of four freely specified parameters $(\Omega,\tilde M,\Xi_{0,0},\Xi_{2,0})$, within the ranges specified in eq. \eqref{eq:asym_initial_data}. Rather than specifying $\Xi_{0,0}$ and $\Xi_{2,0}$ directly, we parameterize our choice in terms of a modulus and angle via \begin{align}
(\Xi_{0,0},\Xi_{2,0})
=
(|\Xi|\cos\rho,|\Xi|\sin\rho).
\label{eq:spinor_angle}
\end{align}
In lieu of $|\Xi|$, we introduce a more physical quantity $I$, proportional to the actual spinor magnitude, defined by 
\begin{equation}
    \begin{aligned}
    I&:=\sqrt{(\Xi_{0,0})^2+(\Xi_{1,0})^2+(\Xi_{2,0})^2+(\Xi_{3,0})^2}= |\Xi|\sqrt{1+B_-^2},
\end{aligned}
\end{equation} where $B_-$ was defined in eq. \eqref{eq:spinor_asym_reln}. We refer to $\rho$ as the parity angle: $\rho = 0$ gives a sphere-parity even Dirac solution, $\rho = \pi/2$ gives a sphere-parity odd solution, while more general $\rho$ correspond to linear combinations of the two. The restriction that $\Xi_{0,0},\Xi_{2,0}\geq0$ restricts $\rho$ to the interval $[0,\pi/2]$.  We will encode our asymptotic spinor initial data in terms of $(I,\rho)$. In summary, at a particular $\ell\in\N+\frac{1}{2}$, we may freely specify \begin{align}
    &\Omega \in (0,1), &&\tilde M > 0 ,&&& I >0, &&&&\rho\in[0,\pi/2]
\end{align} as independent asymptotic parameters that determine a particular Einstein-Dirac solution.

Let us now turn to the numerical strategy for propagating asymptotic initial data into the interior of the spacetime. We define a new metric variable \begin{equation}
    p(x):= \frac{r(x)}{\sqrt{x^2+1}}
\end{equation} so that, in asymptotically flat wormhole spacetimes, $p$ is everywhere finite and nonzero, approaching a constant as $x\to\pm \infty.$ We also introduce a compactified coordinate \begin{equation}
    x = \frac{X}{1-X^2}
\end{equation} so that $X$ ranges over $(-1,1).$

Now fix asymptotic initial data $(\Omega,\tilde M,I,\rho)$ at $x\to-\infty$, and recall we have set the metric asymptotic coefficients $f_0 = 1$, $d = 0$, and the Dirac mass $\mu = 1$ by coordinate redefinitions and rescaling as in eqs. (\ref{eq:rescaling_1}) and (\ref{eq:rescaling_2}). We first propagate our initial data from $X\to-1$ to $X= -0.999$ ($x\approx -500$) using the asymptotic series to $O(1/x^3)$ for the asymptotic spinor components and using the asymptotic form of the metric (eq. \eqref{eq:metric_asymptotic_solution}) for the metric functions.\footnote{\label{ft:asymptotic_scale}The choice $X= -0.999$ could be concerning if we wanted to consider solutions where the length scale associated with the exponential falloff (eq. \eqref{eq:exp_length_scale}) were very large, in which case the asymptotic series solution would not be valid even at $X=-0.999$. We will find in \S \ref{sec:non_existence_generic} that no solutions can be obtained where this length scale exceeds $|X|\approx 0.66$ ($|x| \approx 1.2$), so the choice $X = -0.999$ is exceedingly safe.}  Next, we evolve the asymptotic spinor components $\Xi_i$ and the metric functions $p,f$ directly using a fifth-order Runge-Kutta method with a fourth-order error estimator, setting $\text{atol} = 10^{-6}$ and using the default $\text{rtol} = 10^{-3}$ \cite{SciPyRK45Docs,DormandPrince1980EmbeddedRungeKutta}. We use the two evolution equations [eqs. (\ref{eq:EFE_r_ev}) and (\ref{eq:EFE_f_ev})] to evolve the metric functions $p(x)$ and $f(x)$, and we reserve eq. (\ref{eq:EFE_constraint}) as a means of verifying that solutions are well behaved.  Once we have evolved to a point where the exponential decay does not make the spinor components $\psi_i$ too small for numerical accuracy, we switch from evolving the asymptotic spinor components $\Xi_i$ to the original spinor components $\psi_i$. More precisely, we define the exponential length scale as the one that appears in eq. \eqref{eq:def_asymptotic_variables} (after rescaling $\mu = 1$), \begin{equation}
    x_\text{e} := 1/\sqrt{1-\Omega^2}, \label{eq:exp_length_scale}
\end{equation} and evolve in terms of $\Xi_i$ until we reach $x= -3x_\text{e}$. We then evolve in terms of $\psi_i$ until we reach the wormhole throat or the equations become singular.

Our numerical strategy is summarized concisely in the following five steps: \begin{enumerate}
    \item For a given $\ell\in\N+\frac{1}{2},$ prescribe asymptotic initial data at $x\to-\infty$ in the form of $(\Omega,\tilde M,I,\rho)$.
    \item Using the asymptotic form of the metric and the asymptotic series solution to the Dirac equation, propagate the initial data to finite $X = -0.999$.
    \item Use an adaptive Runge-Kutta method to evolve the Einstein-Dirac system in terms of the asymptotic spinor components $\Xi_i$ until we have reached a point on the order of the exponential length scale.
    \item Switch from the asymptotic spinor components $\Xi_i$ to the ordinary spinor components $\psi_i$, and continue evolving using an adaptive Runge-Kutta method.
    \item Stop evolution at a predetermined point beyond the wormhole throat, or when the evolution breaks down due to singular behavior.
\end{enumerate}

\subsection{Obstructions to Traversable Wormhole Solutions with Definite Sphere-Parity}
\label{sec:non_existence_generic}

Recall from \S \ref{sec:stationary_bound_states_exist} that in a generic wormhole geometry (without additional symmetries beyond being static and spherically symmetric), a bound-state solution of the Dirac equation will have a definite sphere-parity. It is sometimes possible that the spectrum of the Dirac operator admits an accidental degeneracy, allowing for bound-state solutions that are a linear combination of both sphere-parities; however, this is highly exceptional in the space of possible traversable wormhole geometries. We will ultimately consider the case of reflection-symmetric wormholes in \S \ref{sec:non_existence_symmetric}, where we will address more general bound-state solutions that are linear combinations of sphere-parity even and odd. In this section, however, we will restrict to the generic case with fixed frequency $\omega$ and angular momentum quantum number $\ell$ and consider Dirac solutions of a definite sphere-parity. We will present strong numerical evidence that there are no traversable wormhole solutions to the Einstein-Dirac system in such cases. 

We will first show that there is no issue in producing a solution describing a traversable wormhole throat connected to one asymptotic end. Using the numerical scheme laid out in \S \ref{sec:numerical_strategy}, we present in Fig. \ref{fig:half_wormhole_solution} such a solution, where the metric and spinor functions have the correct asymptotic behavior and the evolution produces a wormhole throat (i.e., a local minimum of $r$). We validated the legitimacy of this solution by verifying that the evolution and constraint equations were satisfied with an absolute error less than $10^{-6}$. The existence of these partial solutions allows us to draw important conclusions. First, the Einstein-Dirac equations admit local solutions in a neighborhood of the wormhole throat. Indeed, we can analytically find initial data for the Einstein-Dirac system compatible with $r>0$, $f>0$, $dr/dx=0$, and $d^2r/dx^2>0$, and the existence of a local solution is then guaranteed by standard existence theorems for ordinary differential equations. Since a local inconsistency in the Einstein-Dirac system for traversable wormholes is unlikely to arise anywhere other than at the throat, this suggests that only a nonlocal result could provide a valid nonexistence argument.

\begin{figure}[htbp!]
    \centering
    
    \hfill
        \centering
        \begin{subfigure}[t]{0.48\textwidth}
        \centering
            \centering
            \includegraphics[width=\linewidth]{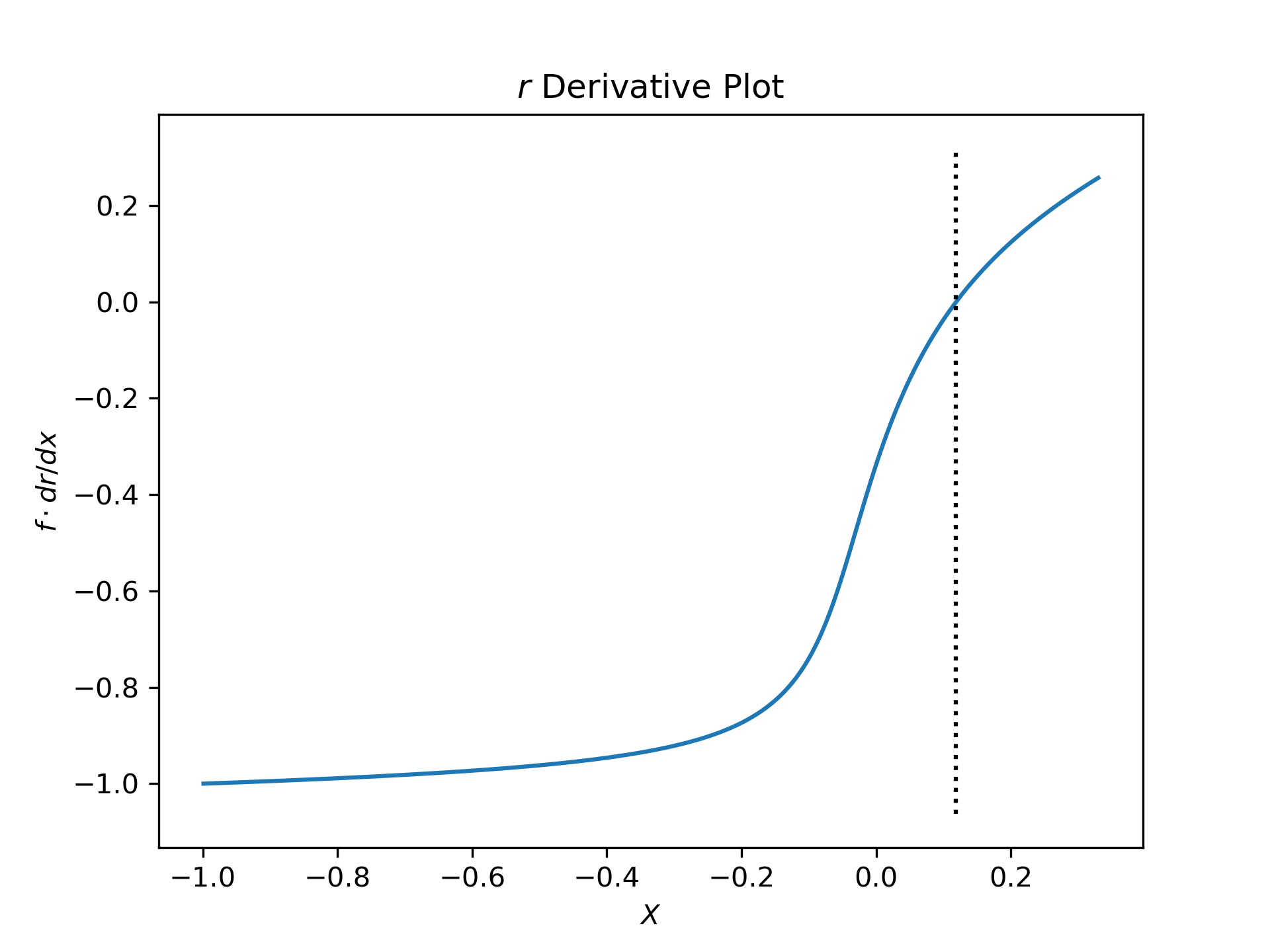}
        \caption{Plot of $f\cdot dr/dx$, showing clearly that a throat was obtained.}
    \end{subfigure}
    \begin{subfigure}[t]{0.48\textwidth}
            \centering
            \includegraphics[width=\linewidth]{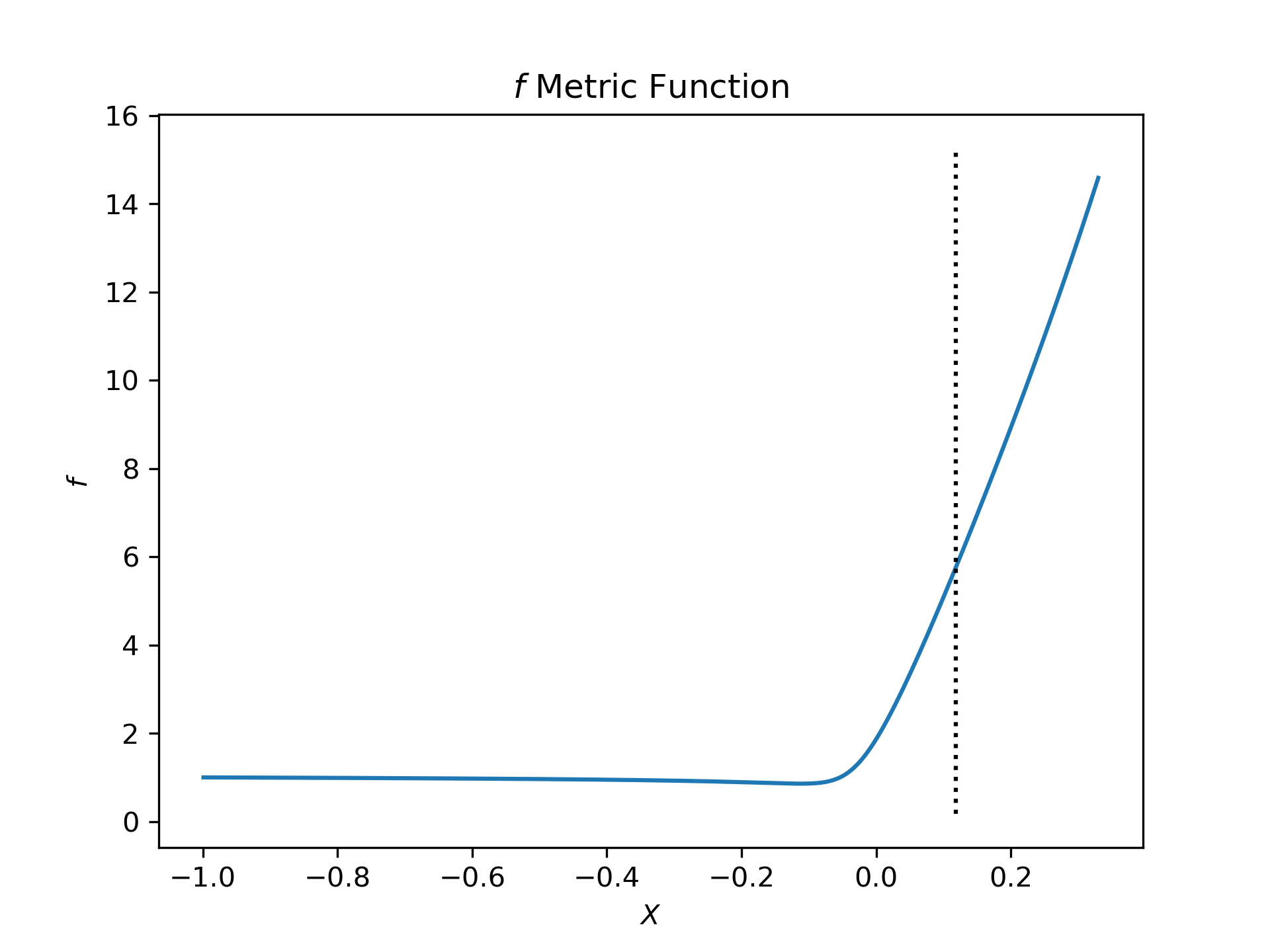}
        \caption{Plot of $f$.}
    \end{subfigure}
    \caption[An example of a partial-wormhole solution with one asymptotic end]{A partial-wormhole solution is presented, connecting a throat to an asymptotically flat region. The dashed black line denotes the $X$-position of the throat in all plots. This example was generated using a sphere-parity even Dirac solution with $(\Omega,\tilde M,I) = (0.1,0.01,0.1)$.}
    \label{fig:half_wormhole_solution}
    \end{figure}

We will now characterize which asymptotic parameter values give rise to a wormhole throat by conducting a numerical scan of the asymptotic parameter space and recording when a local minimum of $r$ is produced before the evolution breaks down. For $\ell \in \{1/2,3/2\}$, we performed a scan over ten logarithmically spaced values of $\Omega$ ranging from $10^{-6}$ to $0.5$, and six logarithmically spaced values of $\tilde M$ from $10^{-3}$ to $10^3$. For each of these values of $\ell$ and $\tilde M$ (except the largest $\tilde M$ value, for which we encountered an overflow error), we also included a search at $\Omega = 0.8$ to include dimensionless frequencies closer to the maximum value $\Omega = 1$. For each value of $(\Omega,\tilde M,\ell)$, we scanned 100 logarithmically spaced asymptotic spinor magnitudes $I$ from $10^{-10}$ to $10^{10}$ to determine whether any wormhole throats were produced using sphere-parity even initial data, sphere-parity odd initial data, neither, or both. The resulting scan is presented in Fig. \ref{fig:omega_m_plane}.

\begin{figure}[!ht]
    \centering
    \begin{subfigure}[t]{0.49\textwidth}
        \centering        \includegraphics[width=\linewidth]{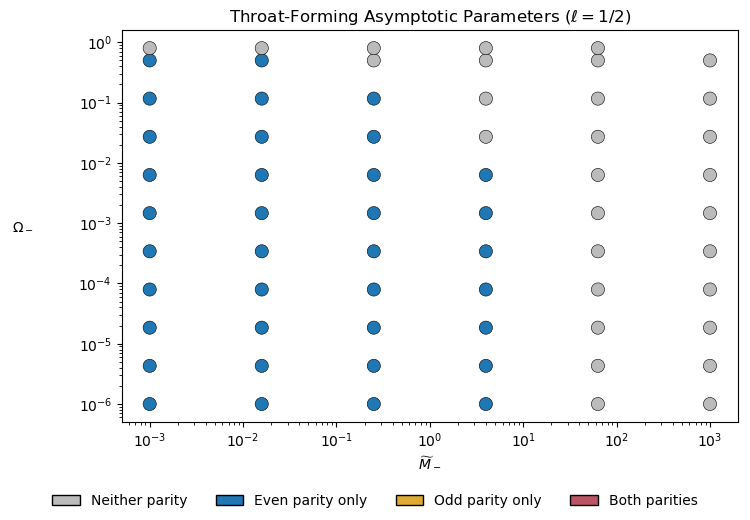}
    \end{subfigure}
    \hfill
    \begin{subfigure}[t]{0.49\textwidth}
        \centering
            \centering
            \includegraphics[width=\linewidth]{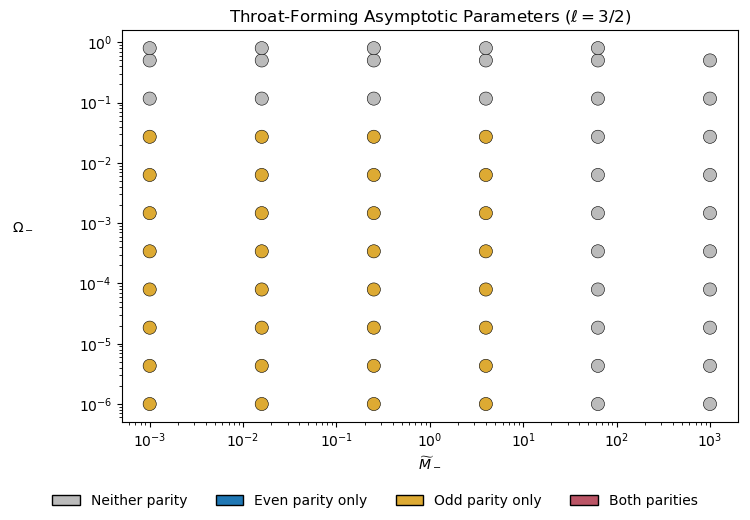}
    \end{subfigure}\\
    \caption[Scan of throat-forming asymptotic parameters]{Scan depicting, for each $\Omega$ and $\tilde M$, whether a wormhole throat could be obtained using asymptotic initial data of either sphere-parity.}
    \label{fig:omega_m_plane}
\end{figure}
\begin{figure}[H]
    \centering
    \includegraphics[width=\linewidth]{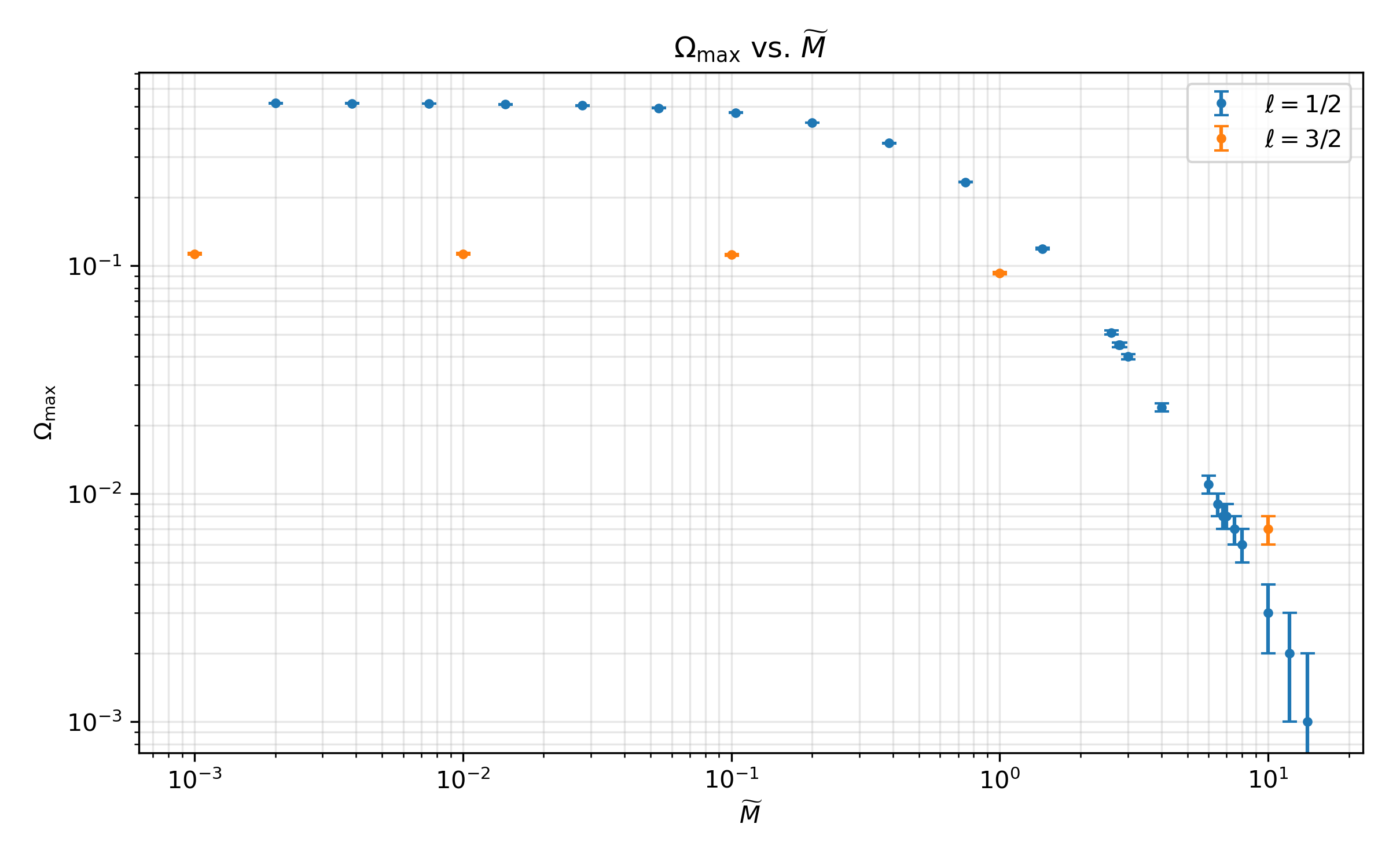}
    \caption[Maximum dimensionless frequency that allows for wormhole throat formation]{Maximum allowed $\Omega$ as a function of $\tilde M$,  computed to an accuracy of $10^{-3}$.}
    \label{fig:max_Omega}
\end{figure}

Figure \ref{fig:omega_m_plane} reveals a few essential features of the throat-forming asymptotic parameter space. First, we observe that for each $\tilde M$, there is a maximum $\Omega$ above which we no longer find any asymptotic initial data that evolves to a wormhole throat. To understand this, we observed in all cases that throat-forming asymptotic spinor magnitudes live within a bounded interval. This behavior is consistent with eq. \eqref{eq:asym_spacetime_mass_equation}, which relates the throat radius and spinor amplitude to the spacetime mass at each end. For a fixed spacetime mass, one can make the spinor amplitude larger at the cost of making the wormhole throat smaller, up to a point. Similarly, one can make the spinor amplitudes smaller and make the wormhole throat larger to compensate, but a spinor amplitude that is too small will result in a horizon; the solution approaches Schwarzschild, and the Dirac field becomes singular at the horizon. 

As $\Omega$ increases, the widths of the allowed bands of asymptotic spinor magnitudes shrink. We numerically tracked the location of the throat-forming spinor magnitude interval as we slowly increased $\Omega$, using increasingly refined searches in $I$ until no throat-forming parameters were observed; the results at various values of $\tilde M$ are presented in Fig. \ref{fig:max_Omega}. Note that wormhole throats were not found with $\Omega \gtrsim 0.52$, which means we do not expect to find solutions with very large exponential length scales (eq. \eqref{eq:exp_length_scale}). This protects our assumption that the solutions are well-described by the asymptotic forms presented in \S \ref{sec:asym_initial_data} for the larger values of $X$ we used in our numerics (see footnote \ref{ft:asymptotic_scale}).

Another striking observation of Fig. \ref{fig:omega_m_plane} is that, at a fixed $\ell$, at most one sphere-parity is capable of forming a wormhole throat in all examples searched. To ensure that this is not an artifact of undersampling in the spinor amplitude $I$, we re-ran the scan for the non-throat-forming parity in the cases $\Omega,\tilde M \in\{10^{-1},10^{-2},10^{-3}\}$ for $\ell \in \{1/2,3/2\}$ at a density of 1000 logarithmically spaced values of $I$, still ranging from $10^{-10}$ to $10^{10}$, and found no exceptions to this claim. One plausible mechanism follows from sign differences in the contributions of the two parities to the Dirac bilinears $S,D$, and $P$. More precisely, define the sphere-parity even and odd parts of the Dirac bilinears as  \begin{equation}
    \begin{aligned}
    &S_\text{even}={\psi_0}^2+{\psi_1}^2,&&S_\text{odd} ={\psi_2}^2 + {\psi_3}^2,\\
    &D_\text{even} = {\psi_0}^2-{\psi_1}^2,&&D_\text{odd} = {\psi_3}^2-{\psi_2}^2,\\
    &P_\text{even} = 2\kappa_\ell \psi_0\psi_1,&& P_\text{odd} = -2\kappa_\ell\psi_2\psi_3,
\end{aligned}\label{eq:parity_sector_bilinears}
\end{equation} where $\kappa_\ell :=(-1)^{\ell-1/2}$, so that the Dirac bilinears introduced in eq. (\ref{eq:Dirac_bilinear_def}) are \begin{equation}
    \begin{aligned}
    &S = S_\text{even} + S_\text{odd},
    &&D = D_\text{odd} + D_\text{even},
    &&&P = P_\text{even} + P_\text{odd}.\label{eq:bilinear_decomp}
\end{aligned}
\end{equation} Though each of these bilinears decays to zero as $|x|\to\infty$ for bound states, each has a definite sign for very large $|x|$. $S_\text{even}$ and $S_\text{odd}$ are manifestly positive everywhere, but for the sign-indefinite bilinears $D,P$, we can use the asymptotic expansion for the spinor coefficients in eq. \eqref{eq:spinor_asym_expansion} and the relations in eq. \eqref{eq:spinor_asym_reln} to obtain  \begin{equation}
    \begin{aligned}
    &\frac{D_\text{even}}{S_\text{even}} \to \pm(-1)^{\ell-1/2}\sqrt{1-\Omega_\pm^2}, 
    && \frac{D_\text{odd}}{S_\text{odd}} \to \mp (-1)^{\ell-1/2}\sqrt{1-\Omega_\pm^2}, 
    &&& \frac{P_\text{even}}{S_\text{even}},\frac{P_\text{odd}}{S_\text{odd}}\to-\Omega_\pm, \label{eq:asymptotic_bilinears}
\end{aligned}
\end{equation} as $x\to\pm \infty$. For a partial-wormhole solution to exist, the partial\footnote{The usual averaged null energy condition involves an integral over a complete, inextendible null geodesic, so we add the ``partial'' modifier to distinguish this condition from the true averaged null energy condition.} averaged null energy condition must be violated, i.e., \begin{equation}
    \int_{-\infty}^{x_t} \left[\frac{2\omega}{f^4r^2}S + \frac{(\ell+1/2)}{f^3r^3}D + \frac{\mu}{f^3r^2}P\right]\ dx\leq 0
\end{equation} when integrating from $x\to-\infty$ to a wormhole throat, which can be verified directly by integrating eq. \eqref{eq:EFE_r_ev}. The first term is manifestly positive, so the case where $D,P<0$ most effectively drives violations of the partial averaged null energy condition. 

In each case examined in our scan, the throat-forming sphere-parity is the one whose asymptotic behavior gives rise to the more favorable sign structure ($D,P<0$) for violating the partial averaged null energy condition. As we are evolving from the left asymptotic end ($x\to-\infty$), we consider the lower sign in eq. \eqref{eq:asymptotic_bilinears}, and see that for $\ell = 1/2$ the sphere-parity even solutions are favored, while for $\ell=3/2$ the sphere-parity odd solutions are favored, and so forth. This matches the pattern observed in Fig. \ref{fig:omega_m_plane}; additionally, a coarse scan consisting of three logarithmically spaced values of $\Omega$ between $10^{-6}$ and $0.5$ and three logarithmically spaced values of $\tilde M$ between $10^{-3}$ and $10^3$ was run at $\ell=5/2$ to validate that the pattern continues to hold, and indeed only sphere-parity even solutions formed wormhole throats.

The results we have presented thus far are for initial data at the left asymptotic end ($x\to-\infty$) evolving into the interior of the spacetime; from the discussion surrounding eq. \eqref{eq:reflection_asymptotic_data}, we can deduce that these results also hold for evolution from the right asymptotic end ($x\to+\infty$), with ``sphere-parity even'' and ``sphere-parity odd'' interchanged. For example, when $\ell = 1/2$, we find only sphere-parity odd Dirac solutions that evolve to a wormhole throat from the right asymptotic end, and only sphere-parity even solutions form throats for $\ell=3/2$, and so forth. Crucially, in each case, the throat-forming sphere-parity from the right asymptotic end is \textit{opposite} that of the left asymptotic end.

This is strong numerical evidence against generic, static, spherically symmetric traversable wormhole solutions. If such a solution did exist, then generically the Dirac solution will have a definite sphere-parity (as discussed in \S \ref{sec:stationary_bound_states_exist}). However, we have seen that the sphere-parity which allows a wormhole throat to form when evolving inward from $x\to -\infty$ is the opposite of the sphere-parity which allows a wormhole throat to form when evolving inward from $x\to +\infty$. Since sphere-parity is a global label of the Dirac solution, a single solution cannot have one sphere-parity on one asymptotic end and the opposite sphere-parity on the other. Thus, if the solution has the parity required to form a throat from the left end, it has the wrong parity to form a throat when evolving in from the right end, and conversely. This suggests that a partial-wormhole solution with definite $\omega,\ell$ and sphere-parity cannot be extended to a second asymptotically flat end.

\subsection{Obstructions to Reflection-Symmetric Traversable Wormhole Solutions}
\label{sec:non_existence_symmetric}

In this section, we will address the symmetry-protected case in which both sphere-parities can be present at the same frequency. As discussed in \S \ref{sec:stationary_bound_states_exist}, reflection symmetry guarantees that sphere-parity even and odd solutions can exist at the same frequency, so a general bound-state solution will be some linear combination of the two. To obtain a numerical obstruction to reflection-symmetric wormholes, we first derive conditions that must be satisfied by the spinor field and metric functions at a reflection-symmetric throat, and introduce an objective function that measures how well these conditions are satisfied. Our approach will be to undertake an extensive scan of the throat-forming region of asymptotic parameter space and show that the reflection-symmetric throat conditions never approach satisfaction over the scanned parameter range.

We begin by establishing the necessary conditions at the throat of a reflection-symmetric wormhole. We will assume, for this discussion, that the wormhole throat is at $x=0$ by translating the $x$-coordinate if necessary. This means that reflection symmetry will be implemented by the map \begin{equation}
    Z(t,x,\theta,\varphi) = (t,-x,\theta,\varphi).
\end{equation} The conditions we derive at the throat will still hold even if the throat is translated away from $x=0$. By definition of a reflection-symmetric geometry, the metric functions must be symmetric about the wormhole throat, which implies that their derivatives must vanish there. We already have $dr/dx = 0$ at the throat, but we must also have $df/dx = 0$. 

Imposing conditions on the Dirac spinor at the throat is a more delicate matter. We do not require each of the spinor components to be symmetric about the wormhole throat; rather, the physically motivated restriction is that all the observables (e.g., Dirac stress-energy and Dirac current) are reflection symmetric. In fact, the condition that $S,D$, and $P$ are reflection-symmetric functions of $x$ about the throat is both a necessary and sufficient condition for a reflection-symmetric Einstein-Dirac solution: from the Einstein equations, we can solve for $S,D$, and $P$ in eqs. \eqref{eq:tt_eqn}-\eqref{eq:angleangle_eqn} and observe that the corresponding geometric expressions are reflection symmetric about the wormhole throat.

We will prove below that the only Dirac solutions that have $S$, $D$, and $P$ reflection symmetric on a reflection-symmetric, asymptotically flat, static, spherically symmetric wormhole background are eigenstates of the reflection operator $Z^*$, which implements $x\mapsto -x$ on spinor fields (see the discussion of reflection symmetry in appendix \ref{app:reflection}). If $\Psi$ solves the Dirac equation on a reflection-symmetric background, then applying the reflection operator $Z^*\Psi$ also gives a solution. If our Dirac solution is an eigenstate of the reflection operator $Z^*$ with eigenvalue $\pm 1$, then its components must satisfy \begin{align}
    & \psi_0(x) = \pm \psi_3(-x), && \psi_1(x) = \mp \psi_2(-x),\label{eq:ref_sym_spinor_cond}
\end{align} and these give rise to a local relation between the sphere-parity even and sphere-parity odd spinor components at the throat, $x=0$.

We now proceed with the proof: consider a Dirac solution on a fixed, reflection-symmetric wormhole background, at a fixed frequency $\omega<\mu f_0^\pm$ and $\ell\in\N+\frac{1}{2}$, where we have taken an incoherent superposition of Dirac solutions over the spherical harmonic index $m$ to produce a spherically symmetric Dirac solution, as discussed in \S \ref{sec:spherical_sym_stress_energy}. This leaves us with only the radial parts of the Dirac spinor, satisfying eq. \eqref{eq:real_dirac}. As eq. \eqref{eq:real_dirac} consists of two pairs of coupled first-order ordinary differential equations, the initial spinor data at the throat is four-dimensional, consisting of a two-dimensional subspace of sphere-parity even solutions and a two-dimensional subspace of sphere-parity odd solutions. Within each parity sector, generic initial data gives rise to a sum of an exponentially growing mode and an exponentially decaying mode asymptotically; when restricting to bound states, the condition that the solutions decay exponentially as $x\to\pm \infty$ imposes an additional constraint which, for a generic $\omega$, does not admit any solutions. However, if the frequency $\omega$ is such that there exists a bound state, then restricting to the exponentially decaying subspace reduces the dimension count within each parity sector by 1, resulting in a two-dimensional solution space consisting of one sphere-parity even bound state and one sphere-parity odd bound state (up to rescaling). The reflection operator $Z^*$ leaves the full two-dimensional space of bound-state solutions invariant, squares to 1, and exchanges the sphere-parity even and sphere-parity odd bound states.

If we choose any sphere-parity even bound state \begin{align}
    \Psi_\text{even} = \begin{pmatrix}
        \psi_0(x)\\
        \psi_1(x)\\
        0 \\
        0 
    \end{pmatrix},
\end{align}then  $\Psi_\text{odd}:=Z^*\Psi_\text{even}$ will be a sphere-parity odd bound state. A general bound-state solution is given by the linear combination \begin{equation}
    \Psi = a\Psi_\text{even} + b\Psi_\text{odd}
\end{equation} with $a,b\in\R$. The Dirac bilinears $S,D$ and $P$ can be decomposed as in eq. \eqref{eq:bilinear_decomp}, for example, \begin{equation}
    S(x)= a^2S_\text{even}(x)+b^2S_\text{odd}(x).
\end{equation} Since the even and odd parts of each bilinear are interchanged under a reflection, this implies $S_\text{even}(x)= S_\text{odd}(-x)$, and similarly for the other bilinears. Because each of these Dirac bilinears is reflection-symmetric about the wormhole throat, it follows that either the even/odd parts of each bilinear are themselves reflection-symmetric functions of $x$, or $a^2=b^2$. From eq. \eqref{eq:asymptotic_bilinears}, the former is not possible, since $D_\text{even}$ must have different signs as it approaches each asymptotic end, and similarly for $D_\text{odd}$. Thus, the only reflection-symmetric solutions have $a^2=b^2\Rightarrow a = \pm b$, which are eigenstates of $Z^*$ with eigenvalue $\pm1$, corresponding exactly to the conditions in eq. \eqref{eq:ref_sym_spinor_cond}.

In summary, we would like to check whether any asymptotic Einstein-Dirac initial data evolves to produce a wormhole throat where \begin{align}
    &df/dx= 0, &&\psi_0 = \pm\psi_3, &&&\psi_1=\mp\psi_2\label{eq:throat_condition}
\end{align} are all satisfied at the throat. It turns out that an obstruction already arises when trying to find throat data meeting the necessary conditions on the spinor components, without restricting to $df/dx = 0$. To assess if the spinor condition is met, it is convenient to consider the ratio \begin{align}
    Q = \frac{2\psi_0\psi_3 - 2\psi_1\psi_2}{{\psi_0}^2+{\psi_1}^2+{\psi_2}^2+{\psi_3}^2}.
\end{align} It is straightforward to check that $|Q|\leq  1$, and that the definite-parity cases $\rho =0$ or $\rho = \pi/2$ give $Q = 0$ everywhere, since one of the two sphere-parities vanishes identically. When the throat condition in eq. \eqref{eq:throat_condition} is met, we find that the inequality is saturated, with $Q = + 1$ in the case $(\psi_0,\psi_1) = (\psi_3,-\psi_2)$ and $Q=-1$ in the case $(\psi_0,\psi_1) = (-\psi_3,\psi_2)$. Defining $Q_t$ to be the value of $Q$ at the throat, this gives a straightforward objective function for the spinor components, by seeking initial data producing a wormhole throat where the spinor data satisfies $Q_t^2= 1$.

For each $\tilde M \in \{10^{-3},10^{-2},10^{-1},1,10,20\}$ and $\ell \in\{1/2,3/2\}$, we scanned 20 linearly spaced values of the parity angle $\rho$ between 0 and $\pi/2$ and 10 logarithmically spaced values of $\Omega$ from $10^{-5}$ to the maximum throat-forming value of $\Omega$ for that $\tilde M$ (see Fig. \ref{fig:max_Omega} and the corresponding discussion). For each $(\Omega,\tilde M,\rho,\ell)$, we used a bisection search to determine the interval of throat-forming spinor magnitudes $I$ (if it exists) and minimized $-Q_t^2$ over this bounded interval using SciPy's ``minimize\_scalar'' function \cite{SciPyMinimizeScalarDocs}. In each case, the throat-forming region contained an interior point at which $Q_t^2$ was maximized. We refined the largest $Q_t^2$ found in our grid search by maximizing over $(\rho,\Omega,\tilde M)$, seeded at that point, using the derivative-free BOBYQA algorithm \cite{BOBYQA}. The value of $Q_t^2$ never exceeded $\approx 1.4\times10^{-4}$, far from the value $Q_t^2 = 1$ required for a reflection-symmetric wormhole throat. Thus, throughout the searched parameter range, the reflection-symmetric case fails before imposing the full set of throat conditions in eq. \eqref{eq:throat_condition}.

In summary, we find no evidence for reflection-symmetric traversable wormhole solutions of the Einstein-Dirac system. Combined with the generic numerical obstruction of the previous section, our results strongly indicate that the Einstein-Dirac system does not admit traversable wormhole solutions within the static, spherically symmetric, asymptotically flat setting studied in this work, when considering a physically meaningful, positive-frequency Dirac source.

\begin{acknowledgments}
The author is deeply grateful to his advisor, Robert M. Wald, for his guidance and support throughout this work. The author also thanks Pau Figueras for his mentorship and valuable assistance with the numerical searches conducted as part of this research. This work was presented as a thesis to the Department of Physics at the University of Chicago, in partial fulfillment of the requirements for the Ph.D. degree.

Part of this research was conducted with the support of the National Science Foundation Graduate Research Fellowship under Grant No. DGE 1706045. This research was also supported in part by NSF Grant 24-03584 and Templeton Foundation Grant No. 62845 to the University of Chicago.
\end{acknowledgments}

\appendix

  \section{Spin-Weighted Spherical Harmonics}\label{app:spin_weighted_Spherical_Harmonics} In this appendix, we introduce definitions and conventions regarding the spin-weighted spherical harmonics. We also highlight key properties of these functions that are relied upon in the paper. We mostly follow the definitions of Goldberg et al. \cite{SpinSHarmonics}, but make adjustments to the angular derivative operators to align with the conventions of Geroch, Held, and Penrose \cite{GHP:1973} used in this paper.

We will consider the static, spherically symmetric line element used in this paper \begin{equation}
    ds^2 = f(x)^2\dt^2-\frac{\dx^2}{f(x)^2} - r(x)^2\dtheta^2-r(x)^2\sin^2\theta \dphi^2.
\end{equation} We now introduce a null tetrad consisting of four complex vector fields $(l^a,n^a,m^a,\bar m^a)$ where $l^a,n^a$ are real and $l_an^a = -m_a\bar m^a = 1$, with all other inner products vanishing. We choose \begin{equation}
    \begin{aligned}
    &l^a = \frac{1}{\sqrt{2}}\left[f(x)^{-1}\left(\frac{\partial}{\partial t}\right)^a + f(x)\left(\frac{\partial}{\partial x}\right)^a\right], 
    && n^a = \frac{1}{\sqrt{2}}\left[f(x)^{-1}\left(\frac{\partial}{\partial t}\right)^a - f(x)\left(\frac{\partial}{\partial x}\right)^a\right],\\
    &m^a = \frac{1}{\sqrt{2}r(x)}\left[\left(\frac{\partial}{\partial\theta}\right)^a + i\csc\theta\left(\frac{\partial}{\partial\varphi}\right)^a\right] ,
    &&\bar m^a = \frac{1}{\sqrt{2}r(x)}\left[\left(\frac{\partial}{\partial\theta}\right)^a - i\csc\theta\left(\frac{\partial}{\partial\varphi}\right)^a\right],
\end{aligned}
\end{equation} so that $l^a,n^a$ are future-directed null vectors\footnote{As our spacetime is spherically symmetric, it is of Petrov type D (so long as the Weyl tensor is not identically zero), which gives rise to two (repeated) principal null vectors of the Weyl tensor. The vectors $l^a$ and $n^a$ are conveniently chosen to align with these principal null directions.} and $m^a,\bar m^a$ are tangent to the 2-spheres. Much like how vector or tensor components are only meaningful with regard to a particular choice of basis, spinor components are also only defined with reference to a particular tetrad. We can preserve the normalization condition $m_a\bar m^a = -1$ under a rescaling of $m^a$ by an arbitrary phase $e^{i\psi}$. A quantity $\eta$ has spin-weight $s$ if $\eta \to e^{is\psi}\eta$ under $m^a \to e^{i\psi}m^a$. 

This has a convenient, equivalent description in the two-spinor language. The identification (\ref{spinor_identification}) in terms of the null tetrad above is simply \begin{equation}
    \begin{aligned}
    &o^A\bar o^{A^\prime}\leftrightarrow l^a, && \iota^A\bar\iota^{A^\prime} \leftrightarrow n^a, 
    &&& o^A\bar\iota^{A^\prime} \leftrightarrow m^a, &&&& \iota^A\bar o^{A^\prime}\leftrightarrow \bar m^a.
\end{aligned} 
\end{equation}The normalization conditions on the null tetrad are equivalent to the normalization conditions on the spinor dyad $\epsilon_{AB}o^A\iota^B=1$. These are preserved under rescaling by a nowhere-vanishing complex scalar function $\lambda$, i.e., we can rescale our spinor basis by \begin{equation}
    \begin{aligned}
    &o^A \to \lambda o^A ,&& \iota^A \to \lambda^{-1}\iota^A ,
    &&& \bar o^{A^\prime} \to \bar\lambda \bar o^{A^\prime} ,&&&& \bar \iota^{A^\prime} \to \bar \lambda^{-1}\bar \iota^{A^\prime}.\label{spinor_basis_rescale}
\end{aligned} 
\end{equation}This, in turn, rescales our spinor components by \begin{equation}
    \begin{aligned}
    &\Phi_0\to\lambda^{-1}\Phi_0, && \Phi_1 \to \lambda \Phi_1, 
    &&&  X_{0^\prime} \to \bar\lambda  X_{0^\prime}, &&&&  X_{1^\prime} \to\bar \lambda^{-1} X_{1^\prime}.
\end{aligned}
\end{equation} We say that a quantity $\eta$ is of type $\{p,q\}$ if $\eta \to \lambda^p\bar\lambda^q\eta$ under (\ref{spinor_basis_rescale}). $\eta$ is said to have \textit{boost-weight} $\frac{1}{2}(p+q)$ and \textit{spin-weight} $\frac{1}{2}(p-q)$.

The spin-weighted spherical harmonics $_sY_{\ell,m}$, where $s,\ell,m\in\Z+\frac{1}{2}$ subject to $\ell\geq |s|$ and $-\ell\leq m\leq \ell$, comprise an orthonormal basis of spin-weight $s$ functions on $S^2$, under the usual $L^2$ inner product on $S^2$. Letting $D^\ell(\alpha,\beta,\gamma)$ denote the Wigner D-matrices (as defined in \cite{SpinSHarmonics}), a convenient definition of the spin-weighted spherical harmonics is given by \begin{align}
    _sY_{\ell,m}(\theta,\varphi) := \sqrt{\frac{2\ell+1}{4\pi}} D^{\ell}_{-s,m}(\varphi,\theta,0).\label{eq:spin_weighted_spherical_harmonics_def}
\end{align} Although we require only half-integer values of $s$ in this work, one can see that the $s=0$ case aligns with the familiar spherical harmonics that are eigenfunctions of the scalar Laplacian on $S^2$. The spin-weight-$s$ spherical harmonics are orthonormal with respect to the usual $L^2$ inner product on the sphere \begin{equation}
    \int_{S^2}\overline{_{s}Y_{\ell,m}(\theta,\varphi)}\ _{s}Y_{\ell^\prime,m^\prime}(\theta,\varphi)\ \sin\theta \dtheta \dphi = \delta_{\ell\ell^\prime}\delta_{mm^\prime}.
\end{equation}

The angular derivative operators $\edth$ and $\edth^\prime$ are defined by their action on a spin-weight-$s$ quantity $\eta$ as \begin{equation}
    \begin{aligned}
    &\edth\eta  =\frac{1}{\sqrt{2}r(x)}\left[\frac{\partial}{\partial\theta} +i\csc\theta\frac{\partial}{\partial\varphi}-s\cot\theta\right]\eta ,
    &&\edth^\prime\eta  =\frac{1}{\sqrt{2}r(x)}\left[\frac{\partial}{\partial\theta} -i\csc\theta\frac{\partial}{\partial\varphi}+s\cot\theta\right]\eta ,\label{eq:edth_def}
\end{aligned}
\end{equation} where we have rescaled the definitions from \cite{SpinSHarmonics} by $-1/\sqrt{2}r(x)$ and, following Ref. \cite{GHP:1973}, opted to use the symbol $\edth^\prime$ instead of $\bar \edth$ throughout this work.

We note a few crucial properties of the spin-weighted spherical harmonics. Under conjugation and under antipodal mapping, they satisfy \begin{equation}
    \begin{aligned}
    &\overline{_sY_{\ell,m}} = (-1)^{m+s}\ _{-s}Y_{\ell,-m}, 
    && _sY_{\ell,m}(\pi-\theta,\varphi+\pi) = (-1)^\ell\ _{-s}Y_{\ell,m}(\theta,\varphi).\label{eq:spin_weighted_harmonics_properties}
\end{aligned}
\end{equation} The angular derivative operators $\edth,\edth^\prime$ act as spin-raising and spin-lowering operators, respectively:\begin{equation}
    \begin{aligned}
    &\edth\ _sY_{\ell,m} = -\sqrt{\frac{(\ell-s)(\ell+s+1)}{2r(x)^2}}\ _{s+1}Y_{\ell,m},
    &&\edth^\prime\ _sY_{\ell,m} = \sqrt{\frac{(\ell+s)(\ell-s+1)}{2r(x)^2}}\ _{s-1}Y_{\ell,m}.
\end{aligned} 
\end{equation} Finally, for fixed $\ell$, the identity \begin{equation}
    \sum_{m=-\ell}^\ell \ _sY_{\ell,m}\ \overline{_{s^\prime}Y_{\ell,m}}= \frac{2\ell+1}{4\pi}\delta_{ss^\prime}
\end{equation} holds, where the sum is taken over $m\in\Z+\frac{1}{2}$ from $-\ell,\dots,\ell$. This follows directly from the definition (\ref{eq:spin_weighted_spherical_harmonics_def}) in terms of the Wigner D-matrices, using the fact that the Wigner D-matrices are unitary at each point on $S^2$, i.e., \begin{equation}
\begin{aligned}
     I &= D^\ell{D^\ell}^\dagger  \Rightarrow \delta_{ss^\prime} = \delta_{-s,-s^\prime} = \sum_m D^\ell_{-s,m}\overline{D^{\ell}_{-s^\prime,m}}  = \frac{4\pi}{2\ell+1}\sum_m\ _sY_{\ell,m}\ \overline{_{s^\prime}Y_{\ell,m}},
\end{aligned}
\end{equation} which gives the desired result.

\section{Action of Discrete Symmetries on the Dirac Field}\label{app:discrete_symmetries}

In this appendix, we will review the definition of spinor fields as sections of associated spinor bundles, and use this machinery to discuss how discrete symmetries of the underlying spacetime act on spinors. We will then turn to their explicit implementations in our spacetime.

\subsection{General Discussion of the Action of Discrete Symmetries on Dirac Fields}

Spinor fields in curved spacetime are defined using principal and associated fiber bundles \cite{Schuller2015SpinStructuresLecture}. Let $(M,g)$ be a spacetime which admits a spin structure.\footnote{There are possible topological restrictions that prevent the discussion that follows from being well-defined, but we will not address those here and assume that the spacetime admits a spinor structure, as is the case for the wormhole spacetimes under consideration in this paper \cite{Geroch1968SpinorStructureI}.} Start by considering the principal bundle $\text{OLM}^+$ of spatially- and temporally-oriented orthonormal frames on $M$, on which the proper, orthochronous Lorentz group $\text{SO}(1,3)^+$ has a natural, free group action. Denote the projection map on $\text{OLM}^+$ by $\pi_{\text{OLM}^+}$ and the right group action on $e\in \text{OLM}^+$ by $\Lambda \in \text{SO}(1,3)^+$ as $e\blacktriangleleft \Lambda$. Famously, $\text{SO}(1,3)^+$ is double-covered by $\text{SL}(2,\C)$, where the double-covering map $\rho$ has $\ker\rho \cong\mathbb{Z}_2$. A spin-frame bundle consists of a principal $\text{SL}(2,\C)$ bundle $(P,\pi,M)$ and a map $\varphi:P\to \text{OLM}^+$ that satisfies $\pi_{\text{OLM}^+}\circ \varphi = \pi$ and, for all $A\in\text{SL}(2,\C)$ and $p\in P$, we have $\varphi(p)\blacktriangleleft \rho(A) = \varphi(p \triangleleft A),$ where $\triangleleft$ denotes the right group action of $\text{SL}(2,\C)$ on $P$. A spin-bundle is then a complex vector bundle associated to the spin-frame bundle with typical fiber $\C^n$, on which $\text{SL}(2,\C)$ has a left group action given by a representation of $\text{SL}(2,\C)$, and spinor fields are sections of the spin-bundle. In particular, a Weyl spinor is the case when we take $n=2$ and $\text{SL}(2,\C)$ acts by its natural representation on $\C^2$, denoted $(\frac{1}{2},0)$. Similarly, a conjugate Weyl spinor is the case when $n=2$ but $\text{SL}(2,\C)$ acts via its conjugate representation, denoted $(0,\frac{1}{2})$. Both of these representations of $\text{SL}(2,\C)$ are irreducible, and a Dirac spinor is a direct sum of these two representations.

This treatment explains how one should implement Lorentz transformations that are continuously connected to the identity on spinors at each point on the spacetime. In this work, we shall want to consider certain discrete spacetime symmetries: the antipodal map $\Upsilon$ and time-reversal $T$ (which are symmetries of any static, spherically symmetric spacetime), and in certain sections we consider the reflection map $Z$. These are given by \begin{equation}
    \begin{aligned}
    &\Upsilon(t,x,\theta,\varphi)= (t,x,\pi-\theta,\varphi+\pi),
    &&T(t,x,\theta,\varphi)= (-t,x,\theta,\varphi),
    &&&Z(t,x,\theta,\varphi)= (t,-x,\theta,\varphi).\label{eq:discrete_symmetries}
\end{aligned} 
\end{equation}

The antipodal and reflection maps reverse the spatial orientation of the frame, and so we need to consider instead the principal frame bundle $\text{LM}^+$ of temporally-oriented (but not spatially-oriented) frames on $M$. On this bundle, $\text{O}(1,3)^+$ has a natural, free, right group action. $\text{O}(1,3)^+$ is double covered by the semidirect product $\text{SL}(2,\C)\rtimes \mathbb{Z}_2$ (where we can take the nontrivial element of $\Z_2$ to act by complex conjugation on $\text{SL}(2,\C)$) and the double covering map $\rho$ has $\ker\rho \cong \mathbb{Z}_2$.\footnote{This way of writing the covering group is convention-dependent. Once disconnected Lorentz transformations are included, there are inequivalent Pin-type choices, distinguished, for example, by the square of these discrete maps \cite{BergDeWittMoretteGwoKramer2001PinGroups}. In the complex Dirac representation, this distinction may appear as a common phase multiplying the symmetry operators on the Dirac spinor, which does not change the induced action on vectors. In this work, we do not need a classification of these choices, and we fix the relevant conventions by the explicit transformations given later in this appendix.\label{ft:discrete_symmetry_conventions}} We may analogously define the spin-frame bundle and spin-bundle by replacing the old double-covering map by this new one, $\text{OLM}^+$ by $\text{LM}^+$, $\text{SL}(2,\C)$ by $\text{SL}(2,\C)\rtimes \mathbb{Z}_2$, and $\text{SO}(1,3)^+$ by $\text{O}(1,3)^+$. The primary distinction between this extended case and the earlier one is that $\text{SL}(2,\C)\rtimes \mathbb{Z}_2$ does not admit any complex two-dimensional irreducible representations -- the smallest irreducible representation on which the $\text{SL}(2,\C)$ part acts non-trivially is 4-dimensional. If we take this 4-dimensional irreducible representation of $\text{SL}(2,\C)\rtimes\mathbb{Z}_2$, but restrict now to the connected subgroup consisting of just $\text{SL}(2,\C)$, the resulting representation is now reducible into a direct sum of the $(\frac{1}{2},0)\oplus(0,\frac{1}{2})$ representations, so we can identify this irreducible representation of $\text{SL}(2,\C)\rtimes\mathbb{Z}_2$ with the aforementioned representation of a Dirac spinor, since the actions of $\text{SL}(2,\C)$ on both agree. The antipodal or reflection maps exchange the two chiral Weyl bundles and therefore act naturally only on their direct sum, the Dirac bundle. Note that, within a fixed choice of double cover, the kernel of the covering map is $\mathbb{Z}_2$, so each Lorentz transformation is implemented uniquely up to an overall sign. The discrete maps carry a further phase ambiguity from the choice of cover itself (see footnote \ref{ft:discrete_symmetry_conventions}), which we fix by the explicit transformations below.

A similar story holds for time-reversal: dropping the temporal orientation of the frame bundle instead of the spatial orientation. As the resulting group is isomorphic to $\text{O}(1,3)^+$, we obtain the same double covering group and the same conclusions hold: time reversal exchanges the two chiral Weyl representations and hence acts naturally only on a Dirac spinor. We will see this fact manifest in the explicit implementations of these discrete symmetries: none of them can be implemented as a linear map that takes components of a single 2-component Weyl spinor to itself. They all involve either interchanging the two chiral Weyl representations or complex conjugation.

\subsection{Antipodal Symmetry}\label{app:antipodal}
We now turn to the explicit implementation of these symmetries in both the 2-spinor and 4-component Dirac spinor notations, specializing to our geometry and components. First, we address antipodal symmetry, which is a symmetry of any spherically symmetric spacetime. Indeed, $\Upsilon^*g_{ab} = g_{ab}$ in our case, where $\Upsilon$ is defined in eq. (\ref{eq:discrete_symmetries}). This acts on the coordinate differentials by\begin{equation}
    \begin{aligned}
    &\Upsilon^*(\dt) = \dt ,&& \Upsilon^*(\dx) = \dx, 
    &&& \Upsilon^*(\dtheta) = -\dtheta ,&&&& \Upsilon^*(\dphi) = \dphi.\label{coordinate_parity}
\end{aligned}
\end{equation}This same effect is achieved by the following transformation of the spinor dyads \begin{equation}
    \begin{aligned}
    &o^A \mapsto -i\bar o^{A^\prime} ,&& \iota^A \mapsto i\bar\iota^{A^\prime} ,
    &&& \bar o^{A^\prime} \mapsto io^A ,&&&&\bar \iota^{A^\prime} \mapsto -i\iota^A,
\end{aligned}
\end{equation} which, under the identification (\ref{spinor_identification}), agrees with eq. (\ref{coordinate_parity}). This, in turn, induces a mapping on the two-spinor components \begin{equation}
    \begin{aligned}
    &\Phi_0\mapsto i \Upsilon^* X_{1^\prime}, &&\Phi_1\mapsto i\Upsilon^* X_{0^\prime},
    &&&  X_{0^\prime} \mapsto -i \Upsilon^*\Phi_1, &&&&  X_{1^\prime} \mapsto -i \Upsilon^*\Phi_0,\label{eq:sphere_parity_component_map}
\end{aligned}
\end{equation} where $\Upsilon^*$ is interpreted as the usual pullback of a scalar function. As an action on 4-component Dirac spinors, using the identification introduced in eq. (\ref{eq:spinor_formalisms_identification}), this is implemented by \begin{align}
    &\Upsilon^*\Psi \vert_{(t,x,\theta,\varphi)}:=P\Psi\vert_{\Upsilon(t,x,\theta,\varphi)}, && P:= i\gamma^0\gamma^2\gamma^3.\label{4_comp_parity}
\end{align}The matrix $P$ squares to the identity, anticommutes with $\gamma^1$, and commutes with the remaining gamma matrices.

Applying $\Upsilon^*$ twice to our Dirac spinor results in an overall sign flip since the components map back to themselves, albeit evaluated at $\varphi+2\pi$, which for spin-$1/2$ fields yields a sign change. The convention where this operator squares to $1$ (e.g., \cite{PeskinSchroeder:1995}) is instead obtained by composing $\Upsilon^*$ with a constant \textit{internal} phase acting on the Dirac field, $\tilde{\Upsilon}^*:= -i\Upsilon^*$. Because every Dirac bilinear contains one factor of $\Psi$ and one of $\bar{\Psi}$, this phase acts trivially on all tensorial quantities and leaves the induced frame transformation untouched, while $(\tilde{\Upsilon}^*)^2 = +1$. The phase $-i$ does not commute with charge conjugation (a Majorana field could not absorb it), so $\Upsilon^*$ as defined above is the unique implementation (up to sign) compatible with the real structure on spinors \cite{BergDeWittMoretteGwoKramer2001PinGroups}, and it is the one we use throughout.

If $ (\Phi_0, \Phi_1, X_{0^\prime}, X_{1^\prime})$ solves the Dirac equation, then applying the map (\ref{eq:sphere_parity_component_map}) or (\ref{4_comp_parity}) to these components gives a solution as well. We can define sphere-parity even and odd eigenstates as in eq. \eqref{eq:def_of_sphere_parity_even_odd_solutions} \begin{align}
    \Psi_\pm:= \frac{1}{\sqrt{2}}(\Psi \mp i \Upsilon^*\Psi),
\end{align} since under the sphere-parity map we have \begin{equation}
    \Upsilon^*\Psi_\pm = \pm i\Psi_\pm.
\end{equation}

Because we use Dirac spinor components that are adapted to antipodal symmetry heavily in this work, it is useful to review here precisely how these are constructed using the sphere-parity operator $\Upsilon^*$. Suppose that the (rescaled) spinor components have stationary time dependence with frequency $\omega$ and are proportional to spherical harmonics of spin-weight $\pm 1/2$ with indices $(\ell,m)$: \begin{equation}
    \begin{aligned}
    &\hat \Phi_0(t,x,\theta,\varphi) = e^{-i\omega t}\phi_0^{(\omega,\ell,m)}(x)\ _{-1/2}Y_{\ell,m}(\theta,\varphi), \\
    &\hat \Phi_1(t,x,\theta,\varphi) = e^{-i\omega t}\phi_1^{(\omega,\ell,m)}(x)\ _{+1/2}Y_{\ell,m}(\theta,\varphi),\\
    &\hat X_{0^\prime}(t,x,\theta,\varphi) = e^{-i\omega t}\chi_{0^\prime}^{(\omega,\ell,m)}(x)\ _{-1/2}Y_{\ell,m}(\theta,\varphi),\\
    &\hat X_{1^\prime}(t,x,\theta,\varphi) = e^{-i\omega t}\chi_{1^\prime}^{(\omega,\ell,m)}(x)\ _{+1/2}Y_{\ell,m}(\theta,\varphi).
\end{aligned} 
\end{equation} Under the identification of eq. \eqref{eq:spinor_formalisms_identification} and using the transformation of the spin-weighted spherical harmonics under sphere-parity [eq. \eqref{eq:spin_weighted_harmonics_properties}], we have\begin{widetext}
    \begin{equation}
    \begin{aligned}
    \hat \Psi_\pm = \frac{e^{-i\omega t}}{\sqrt{2}}\begin{pmatrix}
        \pm  (-1)^{\ell-1/2}(\phi_1^{(\omega,\ell,m)} \pm i(-1)^{\ell-1/2}\chi_{0^\prime}^{(\omega,\ell,m)})\ _{-1/2}Y_{\ell,m}\\
        \pm(-1)^{\ell-1/2}(\phi_0^{(\omega,\ell,m)} \pm i(-1)^{\ell-1/2} \chi_{1^\prime}^{(\omega,\ell,m)})\ _{+1/2}Y_{\ell,m}\\
        (\phi_0^{(\omega,\ell,m)} \pm i(-1)^{\ell-1/2}\chi_{1^\prime}^{(\omega,\ell,m)})\ _{-1/2}Y_{\ell,m}\\
        (\phi_1^{(\omega,\ell,m)} \pm i(-1)^{\ell-1/2}\chi_{0^\prime}^{(\omega,\ell,m)})\ _{+1/2}Y_{\ell,m} ,
    \end{pmatrix}
\end{aligned}
\end{equation}
\end{widetext} which can be rewritten in terms of the components defined in eq. \eqref{eq:pre_sphere_parity_components} as \begin{widetext}
    \begin{equation}
    \begin{aligned}
    &\hat \Psi_+ = \frac{e^{-i\omega t}}{\sqrt{2}}\begin{pmatrix}
        (-1)^{\ell-1/2}\tilde\psi_0^{(\omega,\ell,m)}(x)\ _{-1/2}Y_{\ell,m}(\theta,\varphi)\\
        (-1)^{\ell-1/2} \tilde\psi_1^{(\omega,\ell,m)}(x)\ _{+1/2}Y_{\ell,m}(\theta,\varphi)\\
        \tilde\psi_1^{(\omega,\ell,m)}(x)\ _{-1/2}Y_{\ell,m}(\theta,\varphi)\\
        \tilde\psi_0^{(\omega,\ell,m)}(x)\ _{+1/2}Y_{\ell,m}(\theta,\varphi)
    \end{pmatrix},\\
    &
    \hat \Psi_- = \frac{e^{-i\omega t}}{\sqrt{2}}\begin{pmatrix}
        -(-1)^{\ell-1/2}\tilde\psi_2^{(\omega,\ell,m)}(x)\ _{-1/2}Y_{\ell,m}(\theta,\varphi)\\
        -(-1)^{\ell-1/2} \tilde\psi_3^{(\omega,\ell,m)}(x)\ _{+1/2}Y_{\ell,m}(\theta,\varphi)\\
        \tilde\psi_3^{(\omega,\ell,m)}(x)\ _{-1/2}Y_{\ell,m}(\theta,\varphi)\\
        \tilde\psi_2^{(\omega,\ell,m)}(x)\ _{+1/2}Y_{\ell,m}(\theta,\varphi)
    \end{pmatrix}.
\end{aligned} 
\end{equation}
\end{widetext}
 A sphere-parity even solution, defined by $\hat \Psi_- = 0$, is characterized by $\tilde\psi_2^{(\omega,\ell,m)}=\tilde\psi_3^{(\omega,\ell,m)} = 0$, while a sphere-parity odd solution, defined by $\hat \Psi_+ = 0$, is characterized by $\tilde\psi_0^{(\omega,\ell,m)}=\tilde\psi_1^{(\omega,\ell,m)} = 0$.

\subsection{Time-Reversal Symmetry}\label{app:time_reversal} Next, we consider time-reversal symmetry. For the static spacetimes we are considering here, we have that $T$, defined in eq. (\ref{eq:discrete_symmetries}), is an isometry $T^*g_{ab} = g_{ab}$, and therefore is a discrete symmetry of our spacetime. The coordinate differentials satisfy \begin{equation}
    \begin{aligned}
    &T^*(\dt)= -\dt,&& T^*(\dx) = \dx, 
    &&& T^*(\dtheta) = \dtheta,&&&&T^*(\dphi) = \dphi.\label{eq:time_reversal_coord_diff}
\end{aligned}
\end{equation} When acting on spinors, we need to be more careful. As discussed in \S \ref{sec:semi_classical_meaning}, the sign of the frequency for a stationary solution influences whether the solution has a meaningful semiclassical interpretation, so we would like time-reversal not to interchange positive/negative-frequency modes, as would na\"ively be the case. This is resolved by demanding that time-reversal act \textit{antilinearly} on spinors \cite{PeskinSchroeder:1995}, i.e., $T^*(\alpha \Psi_1+\Psi_2) = \bar\alpha T^*(\Psi_1)+T^*(\Psi_2)$, where $\Psi_1,\Psi_2$ are spinors, $\alpha\in \C$, and $\bar\alpha$ is the complex conjugate of $\alpha$. The action of time-reversal on our spinor basis is given by \begin{equation}
    \begin{aligned}
    &o^A\mapsto -\iota^A, &&\iota^A\mapsto o^A,
    &&&\bar o^{A^\prime}\mapsto \bar \iota^{A^\prime}, &&&& \bar\iota^{A^\prime} \mapsto -\bar o^{A^\prime},
\end{aligned}
\end{equation} which under the identification (\ref{spinor_identification}) is in agreement with eq. (\ref{eq:time_reversal_coord_diff}).  This induces a map on the spinor components \begin{equation}
    \begin{aligned}
    &\Phi_0\mapsto -T^*\bar\Phi_1, &&\Phi_1\mapsto T^*\bar\Phi_0, 
    &&& X_{0^\prime} \mapsto T^*\bar X_{1^\prime} ,&&&&X_{1^\prime}\mapsto -T^*\bar X_{0^\prime},\label{eq:time_reversal_dirac}
\end{aligned}
\end{equation} where $T^*$ is interpreted as the usual pullback of a scalar function. On 4-component Dirac spinors, using the identification introduced in eq. (\ref{eq:spinor_formalisms_identification}), this is implemented by  \begin{align}
&T^*\Psi\vert_{(t,x,\theta,\varphi)} = (\mathcal{T}\Psi^*)\vert_{T(t,x,\theta,\varphi)}, && \mathcal{T} := -\gamma^1\gamma^3
\end{align} where $\Psi^*$ denotes componentwise complex conjugation. This operator commutes with $\gamma^0$ and anticommutes with the remaining gamma matrices, 
and squares to minus the identity when acting on spinor fields.\footnote{The commutation/anticommutation properties are opposite to what one might have expected the behavior of a time-reversal to be on the gamma matrices, but recall that the Dirac equation contains a factor of $i$ that will also flip sign (since time-reversal is antilinear), which restores the expected behavior. The (anti)commutation properties of this operator only hold when considering the antilinear properties of time-reversal, and are not true of the matrix $\mathcal{T}$ on its own.} This maps Dirac solutions to solutions, and preserves the character of positive/negative-frequency solutions.

It will be useful to express the time-reversal action in terms of the radial functions used for stationary, spherically symmetric spacetimes. Consider a stationary Dirac solution with fixed $\omega$, expanded in spin-weighted spherical harmonics as in eq. \eqref{eq:spinor_component_decomposition}. Because time reversal acts antilinearly, the stationary time dependence is preserved, as is the frequency $\omega$. Complex conjugation of the spin-weight $\pm1/2$ spherical harmonics maps $m\mapsto-m$ up to a factor of $(-1)^{m\pm 1/2}$ [see eq. \eqref{eq:spin_weighted_harmonics_properties}]. Thus, the action induced by eq. \eqref{eq:time_reversal_dirac} on the radial functions is \begin{equation}
    \begin{aligned}
        &\phi_0^{(\omega,\ell,m)}\mapsto (-1)^{m+1/2}\overline{\phi_1^{(\omega,\ell,-m)}},
        &&\phi_1^{(\omega,\ell,m)}\mapsto (-1)^{m+1/2} \overline{\phi_0^{(\omega,\ell,-m)}}, \\
        &\chi_{0^\prime}^{(\omega,\ell,m)} \mapsto (-1)^{m-1/2}\overline{\chi_{1^\prime}^{(\omega,\ell,-m)}}, 
        && \chi_{1^\prime}^{(\omega,\ell,m)} \mapsto (-1)^{m-1/2} \overline{\chi_{0^\prime}^{(\omega,\ell,-m)}}.
    \end{aligned}
\end{equation} Applying this map to the Dirac components introduced in eq. \eqref{eq:pre_sphere_parity_components} yields \begin{equation}
    \begin{aligned}
        &\tilde\psi_0^{(\omega,\ell,m)} \mapsto (-1)^{m+1/2}\overline{\tilde\psi_1^{(\omega,\ell,-m)}}, 
        &&\tilde\psi_1^{(\omega,\ell,m)}\mapsto (-1)^{m+1/2}\overline{\tilde\psi_0^{(\omega,\ell,-m)}},\\
        &\tilde\psi_2^{(\omega,\ell,m)}\mapsto (-1)^{m+1/2}\overline{\tilde\psi_3^{(\omega,\ell,-m)}}, 
        && \tilde\psi_3^{(\omega,\ell,m)} \mapsto (-1)^{m+1/2}\overline{\tilde\psi_2^{(\omega,\ell,-m)}}.
    \end{aligned}
\end{equation} In terms of the spinor components introduced in eq. \eqref{eq:real_spinor_component_definition}, time-reversal acts as \begin{equation}\begin{aligned}
    &\psi_0^{(\omega,\ell,m)} \mapsto i(-1)^{m-1/2} \overline{\psi_0^{(\omega,\ell,-m)}},
    &&\psi_1^{(\omega,\ell,m)} \mapsto i(-1)^{m-1/2} \overline{\psi_1^{(\omega,\ell,-m)}},\\
    &\psi_2^{(\omega,\ell,m)}\mapsto i(-1)^{m+1/2}\overline{\psi_2^{(\omega,\ell,-m)}}, 
    &&\psi_3^{(\omega,\ell,m)}\mapsto i(-1)^{m+1/2}\overline{\psi_3^{(\omega,\ell,-m)}}. \label{eq:time_reversal_radial_map}
\end{aligned}\end{equation} 
Because time-reversal maps the radial parts of Dirac solutions with
frequency/spherical-harmonic indices $(\omega,\ell,m)$ to ones with
$(\omega,\ell,-m)$, and the radial Dirac equation is independent of $m$,
the transformation in eq. \eqref{eq:time_reversal_radial_map} maps the space of solutions of the radial
Dirac equation for a particular $(\omega,\ell)$ to itself. Within each
sphere-parity sector, this map is complex conjugation times a single
constant phase, and the phase may be dropped because the equation is
linear and bilinear observables will be unaffected. Thus, the componentwise conjugate of any solution is again a
solution. Since a first-order linear system is uniquely determined by its
solution space, the radial Dirac operator in each sector must equal its own
complex conjugate; i.e., the radial Dirac equation in the components of
eq. \eqref{eq:real_spinor_component_definition} is necessarily real. This explains why obtaining a real, radial Dirac
equation, as we did in eq. \eqref{eq:real_dirac}, was guaranteed.

\subsection{Reflection Symmetry}\label{app:reflection}
In situations where our spacetime metric has a reflection symmetry about $x=0$, we have that $Z^*g_{ab} =g_{ab}$, where $Z$ is defined by eq. (\ref{eq:discrete_symmetries}), and hence the reflection $x\to-x$ is a spacetime symmetry. This is useful when considering a wormhole in which the geometry on both sides of the throat is identical, and we can redefine the coordinate $x$ by a translation so that $x=0$ is the location of the throat about which the spacetime is symmetric. The action of $Z$ on the coordinate differentials is\begin{equation}
    \begin{aligned}
    &Z^*(\dt) = \dt ,&& Z^*(\dx) = -\dx,
    &&& Z^*(\dtheta) = \dtheta ,&&&& Z^*(\dphi) = \dphi.\label{coord_reflection}
\end{aligned}
\end{equation}  This can be achieved by the following map on the spinor dyad\begin{equation}
    \begin{aligned}
    &o^A\mapsto \bar \iota^{A^\prime},&&\iota^A \mapsto \bar o^{A^\prime}, 
    &&& \bar o^{A^\prime}\mapsto \iota^A ,&&&& \bar\iota^{A^\prime} \mapsto o^A,
\end{aligned}
\end{equation}  which reproduces eq. (\ref{coord_reflection}) under the identification in eq. (\ref{spinor_identification}). This in turn induces a map on the spinor components \begin{equation}
    \begin{aligned}
    &\Phi_0\mapsto iZ^* X_{0^\prime} ,&& \Phi_1\mapsto -iZ^*X_{1^\prime,} 
    &&& X_{0^\prime} \mapsto -iZ^*\Phi_0, &&&& X_{1^\prime} \mapsto iZ^*\Phi_1,
\end{aligned}
\end{equation} where $Z^*$ is interpreted as the usual pullback of a scalar function. As an action on 4-component Dirac spinors, using the identification introduced in eq. (\ref{eq:spinor_formalisms_identification}), this is implemented by  \begin{align}
    &Z^*\Psi \vert_{(t,x,\theta,\varphi)}:=R\Psi\vert_{Z(t,x,\theta,\varphi)}, && R:= i\gamma^0\gamma^1\gamma^2.
\end{align}The matrix $R$ squares to the identity, anticommutes with $\gamma^3$, and commutes with the remaining gamma matrices.

It will be useful to translate the action of the reflection map into an action on the radial functions used for stationary, spherically symmetric solutions. Consider a stationary mode with fixed $(\omega,\ell,m)$, expanded in spin-weighted spherical harmonics as in eq. \eqref{eq:spinor_component_decomposition}, and then rewritten in terms of the real radial components $(\psi_0^{(\omega,\ell,m)},\psi_1^{(\omega,\ell,m)},\psi_2^{(\omega,\ell,m)},\psi_3^{(\omega,\ell,m)})$ introduced in eq. \eqref{eq:real_spinor_component_definition}. Since the reflection map $Z$ leaves $t,\theta,\varphi$ unchanged and only sends $x\mapsto -x$, it does not change $\omega$, $\ell$, or $m$. Its only effect is to evaluate the radial functions at $-x$ and to interchange the two sphere-parity sectors. The induced action is
\begin{equation}
    \begin{aligned}
    &\psi_0^{(\omega,\ell,m)}(x) \mapsto (-1)^{\ell-1/2}\psi_3^{(\omega,\ell,m)}(-x), 
    &&\psi_1^{(\omega,\ell,m)}(x) \mapsto -(-1)^{\ell-1/2}\psi_2^{(\omega,\ell,m)}(-x), \\
    &\psi_2^{(\omega,\ell,m)}(x) \mapsto -(-1)^{\ell-1/2}\psi_1^{(\omega,\ell,m)}(-x),
    &&\psi_3^{(\omega,\ell,m)}(x) \mapsto (-1)^{\ell-1/2}\psi_0^{(\omega,\ell,m)}(-x). \label{eq:reflection_radial_components}
\end{aligned}
\end{equation} One can check directly that applying this transformation twice returns the original radial functions, as expected from $(Z^*)^2=1$ with our conventions.

For the static and spherically symmetric solutions studied in the latter sections of the paper, we take an incoherent sum over $m$ with the same radial functions for each value of $m$. Since $Z$ does not act on the angular coordinates, the same transformation \eqref{eq:reflection_radial_components} applies to each member of the $m$-multiplet. Thus, the reflected spherically symmetric configuration is obtained simply by replacing the common radial functions with the reflected radial functions as above.

\bibliography{apssamp}

\end{document}